\documentclass[11pt]{article} 
\usepackage[margin=1in]{geometry}
\usepackage[utf8]{inputenc}
\usepackage[T1]{fontenc}
\usepackage{float}
\usepackage{framed}
\usepackage{algorithm}
\usepackage[noend]{algpseudocode}
\usepackage{verbatim}
\usepackage{amssymb,amsfonts,amsmath,amsthm}
\usepackage{enumerate}
\usepackage[font=small,labelfont=bf]{caption}
\usepackage{color}
\usepackage{graphicx}
\usepackage{tabulary}
\usepackage{diagbox}
\usepackage[dvipsnames]{xcolor}
\usepackage{pgfplots}
\usepackage[colorlinks=true, linkcolor=blue, urlcolor=blue, citecolor=ForestGreen]{hyperref}
\usepackage{cleveref}
\usepackage{tabularx}

\usepackage{comment}

\usepackage{algorithm}
\usepackage[noend]{algpseudocode}
\usepackage{amsmath,amssymb,amsthm}
\usepackage{bm}
\usepackage{bbm}
\usepackage{bbold}
\usepackage{booktabs}
\usepackage{caption}
\usepackage{enumerate}
\usepackage{float}
\usepackage{subcaption}
\usepackage[normalem]{ulem}
\usepackage[dvipsnames]{xcolor}
\usepackage{xspace}
\usepackage[numbers,sort&compress]{natbib}
\pgfplotsset{compat=1.18}

\newcommand{\mcB}{\mathcal{B}}
\newcommand{\mcC}{\mathcal{C}}

\newcommand{\tO}{\tilde{O}}
\newcommand{\mcP}{\mathcal{P}}

\allowdisplaybreaks

\usepackage{thm-restate}

\newcommand{\neighbor}{\texttt{NN}}
\newcommand{\wnn}{\textsf{W\textsubscript{NN}}}

\newcommand{\merge}{\texttt{Merge}}

\newcommand{\val}{\text{val}}

\usepackage{setspace}

\newcommand{\poly}{\operatorname{poly}}

\newcommand{\argmax}{\mathop{\mathrm{argmax}}}

\newcommand{\dSing}{d_\mathrm{single}}
\newcommand{\dAvg}{d_\mathrm{average}} 
\newcommand{\dCen}{d_\mathrm{cen}}
\newcommand{\dWard}{d_\mathrm{Ward}}

\newcommand{\R}{\mathbb{R}}

\usepackage{footnote}

\newcommand{\cA}{\mathcal{A}}

\newcommand{\cC}{\mathcal{C}}

\newcommand{\cS}{\mathcal{S}}

\newtheorem{theorem}{Theorem}

\newtheorem{lemma}{Lemma}
\newtheorem{definition}{Definition}

\newcommand{\eps}{\epsilon}

\usepackage[framemethod=TikZ]{mdframed}
\newcounter{Frame}

\title{
Parallel Hierarchical Agglomerative Clustering\\ in Low Dimensions
}

\author{%
MohammadHossein Bateni\\%
Google Research\\%
New York, USA%
\and%
Laxman Dhulipala\thanks{Supported by NSF grants CCF-2403235 and CNS-2317194}\\%
University of Maryland\\%
College Park, USA%
\and%
Willem Fletcher\thanks{Supported by NSF grant CCF-2403236}\\%
Brown University\\
Providence, USA%
\and%
Kishen N Gowda\footnotemark[1]\\%
University of Maryland\\
College Park, USA%
\and%
D Ellis Hershkowitz\footnotemark[2]\\%
Brown University\\
Providence, USA%
\and%
Rajesh Jayaram\\%
Google Research\\
New York, USA%
\and%
\hspace{11mm}Jakub Łącki\\%
\hspace{11mm}Google Research\\
\hspace{11mm}New York, USA%
}
\date{}

\begin{document}
\maketitle
\begin{abstract}
Hierarchical Agglomerative Clustering (HAC) is an extensively studied and widely used method for hierarchical clustering in $\mathbb{R}^k$ based on repeatedly merging the closest pair of clusters according to an input linkage function $d$. Highly parallel (i.e., NC) algorithms are known for $(1+\epsilon)$-approximate HAC (where near-minimum rather than minimum pairs are merged) for certain linkage functions that monotonically increase as merges are performed. However, no such algorithms are known for many important but non-monotone linkage functions such as centroid and Ward's linkage.

In this work we show that a general class of non-monotone linkage functions---which include centroid and Ward's distance---admit efficient NC algorithms for $(1+\epsilon)$-approximate HAC in low dimensions. Our algorithms are based on a structural result which may be of independent interest: the height of the hierarchy resulting from any constant-approximate HAC on $n$ points for this class of linkage functions is at most $\operatorname{poly}(\log n)$ as long as $k = O(\log \log n / \log \log \log n)$. Complementing our upper bounds, we show that NC algorithms for HAC with these linkage functions in \emph{arbitrary} dimensions are unlikely to exist by showing that HAC is CC-hard when $d$ is centroid distance and $k = n$.
\end{abstract}

\pagenumbering{gobble}
\newpage

\tableofcontents
\newpage
\pagenumbering{arabic}

\setcounter{page}{1}


\section{Introduction}

Hierarchical Agglomerative Clustering (HAC) is a greedy bottom-up clustering algorithm which takes as input a collection of $n$ points $\mcP \subseteq \mathbb{R}^k$ and a symmetric ``linkage'' function $d : 2^{\mcP} \times 2^{\mcP} \to \mathbb{R}_{\geq 0}$ which gives the ``distance'' between clusters of points.\footnote{One can also study linkage functions that give the ``similarity'' between clusters; see, e.g., \cite{dhulipala2022hierarchical}.} 
HAC begins with the singleton clustering $\mcC = \{\{p\} : p \in \mcP\}$. Then, over the course of $n-1$ steps, it merges the two closest clusters.
In particular, in each step if $A, B \in \mcC$ are the clusters in $\mcC$ minimizing $d(A,B)$, it removes $A$ and $B$ from $\mcC$ and adds $A \cup B$ to $\mcC$.
We say that $d(A,B)$ is the \emph{value} of the merge.
The resulting hierarchy of clusters on $\mcP$ naturally corresponds to a rooted binary tree of clusters called a dendrogram. 
In particular, the dendrogram has $2n-1$ nodes; $A \cup B$ is a node of the dendrogram with children $A$ and $B$ if at some point HAC merges clusters $A$ and $B$.
See \Cref{fig:HAC}. 
If the merged clusters $A$ and $B$ do not satisfy $d(A, B) \leq \min_{C,D \in \mcC}d(C,D)$ in each of the $n-1$ iterations but satisfy $d(A, B) \leq c \cdot \min_{C,D \in \mcC}d(C,D)$ for some fixed $c \geq 1$, then the HAC is said to be $c$-approximate.

HAC has seen widespread adoption in practice because, unlike other clustering algorithms (e.g., $k$-means), it does not require users to pre-specify the number of desired clusters. 
Furthermore, after running the algorithm once, the output dendrogram can be cut to obtain clusterings with varying numbers of clusters.
HAC implementations are available in many widely-used data science libraries such as SciPy~\cite{virtanen2020scipy}, scikit-learn~\cite{pedregosa2011scikit}, fastcluster~\cite{mullner2013fastcluster},  Julia~\cite{Clusteringjl}, R~\cite{hclust}, MATLAB~\cite{mathworks_linkage}, Mathematica~\cite{mathematica} and many more \cite{murtagh2012algorithms,murtagh2017algorithms, shearer1985alglib}. 
Likewise, its adoption is explained by the fact that it performs well for many natural objective functions, both in theory \cite{grosswendt2019analysis, moseley-wang} and practice \cite{bateni2024efficient, dhulipala2023terahac, yu2024parclusterers, yu2025dynhac, dhulipala2022hierarchical, dhulipala2021hierarchical, dhulipala2024optimal, epasto2021clustering, cochez2015twister}.

Different linkage functions give different variants of HAC. In roughly ascending order of complexity, some common linkage functions define the distance between two clusters $A, B \subseteq \mathbb{R}^k$ as:
\begin{itemize}
    \item \textbf{Single-linkage}: the minimum distance $\dSing(A,B) := \min_{(a,b) \in A \times B}||a-b||$.
    \item \textbf{Average-linkage}: the average distance $\dAvg(A,B) := \frac{1}{|A||B|}\sum_{(a,b) \in A \times B} ||a-b||$.
    \item \textbf{Centroid-linkage}: the distance between centroids $\dCen(A,B) := ||\mu(A) - \mu(B)||$ where $\mu(X) := \sum_{x \in X} x / |X|$ is the centroid of $X \subseteq \mathbb{R}^k$.
    \item \textbf{Ward's-linkage}: how much the merge would increase the $k$-means objective $\dWard(A,B) := \Delta(A\cup B)-\Delta(A)-\Delta(B)$ where $\Delta(X) := \sum_{x \in X} ||x-\mu(X)||^2$ is the $k$-means objective.
\end{itemize}

\begin{figure}[h]
    \centering
    \begin{subfigure}[b]{0.19\textwidth}
        \centering
        \includegraphics[width=\textwidth,trim=0mm 0mm 350mm 0mm, clip]{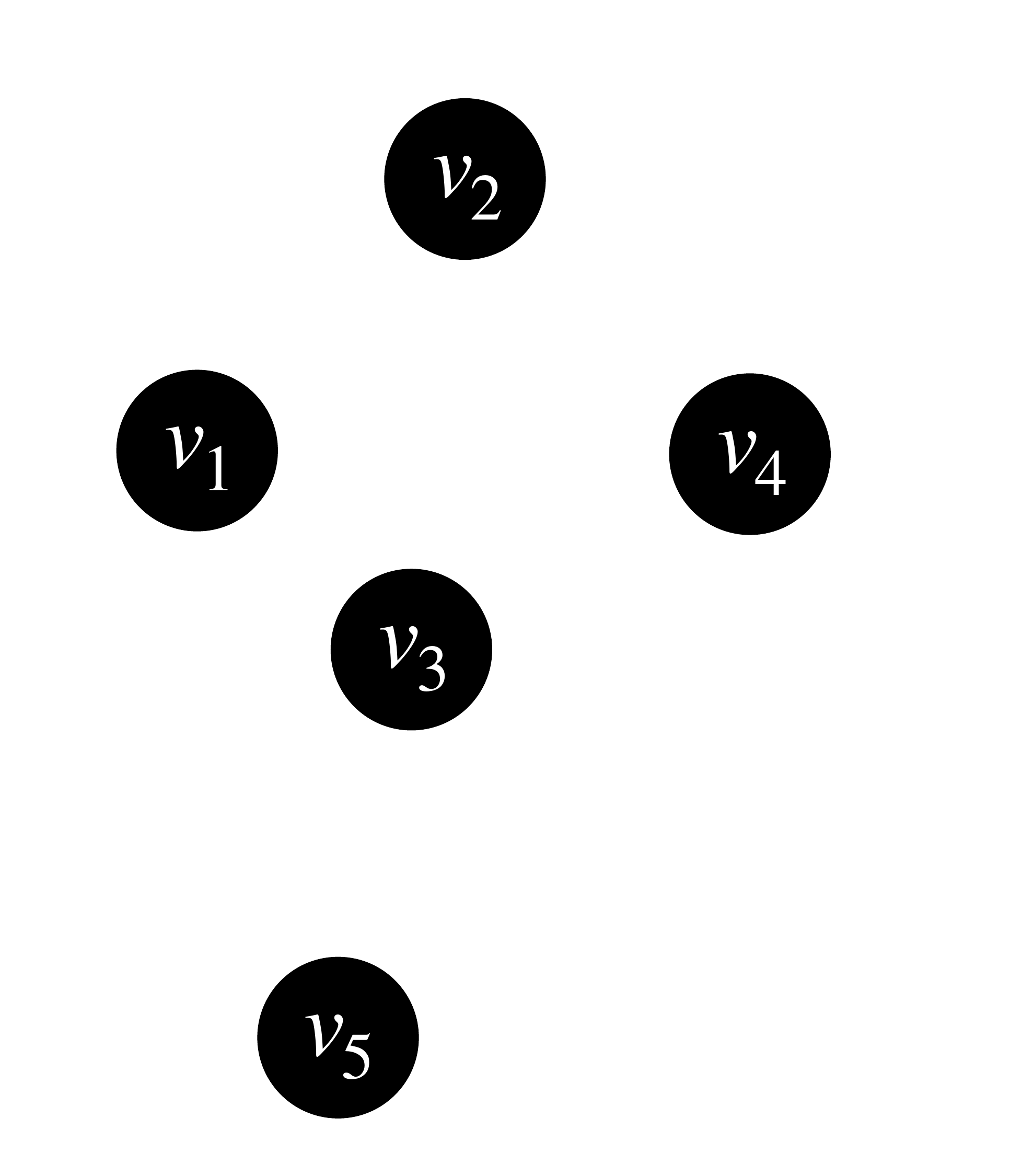}
        \caption{Input $\mcP$.}\label{sfig:hac1}
    \end{subfigure} \hspace{3.5em} 
    \begin{subfigure}[b]{0.19\textwidth}
        \centering
        \includegraphics[width=\textwidth,trim=0mm 0mm 350mm 0mm, clip]{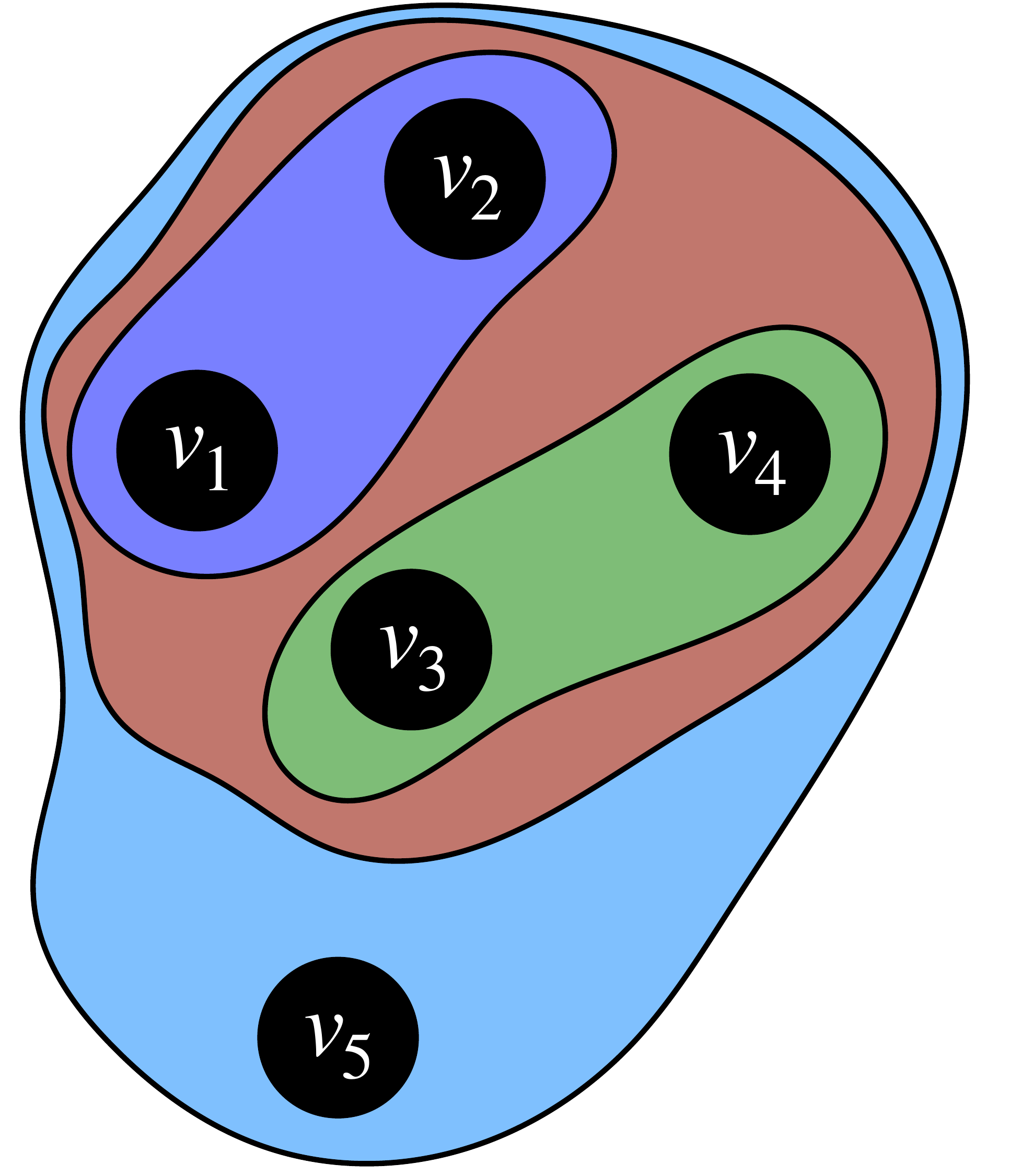}
        \caption{HAC Clustering.}\label{sfig:hac2}
    \end{subfigure}  \hspace{3.5em}
    \begin{subfigure}[b]{0.19\textwidth}
        \centering
        \includegraphics[width=\textwidth,trim=0mm 0mm 320mm 0mm, clip]{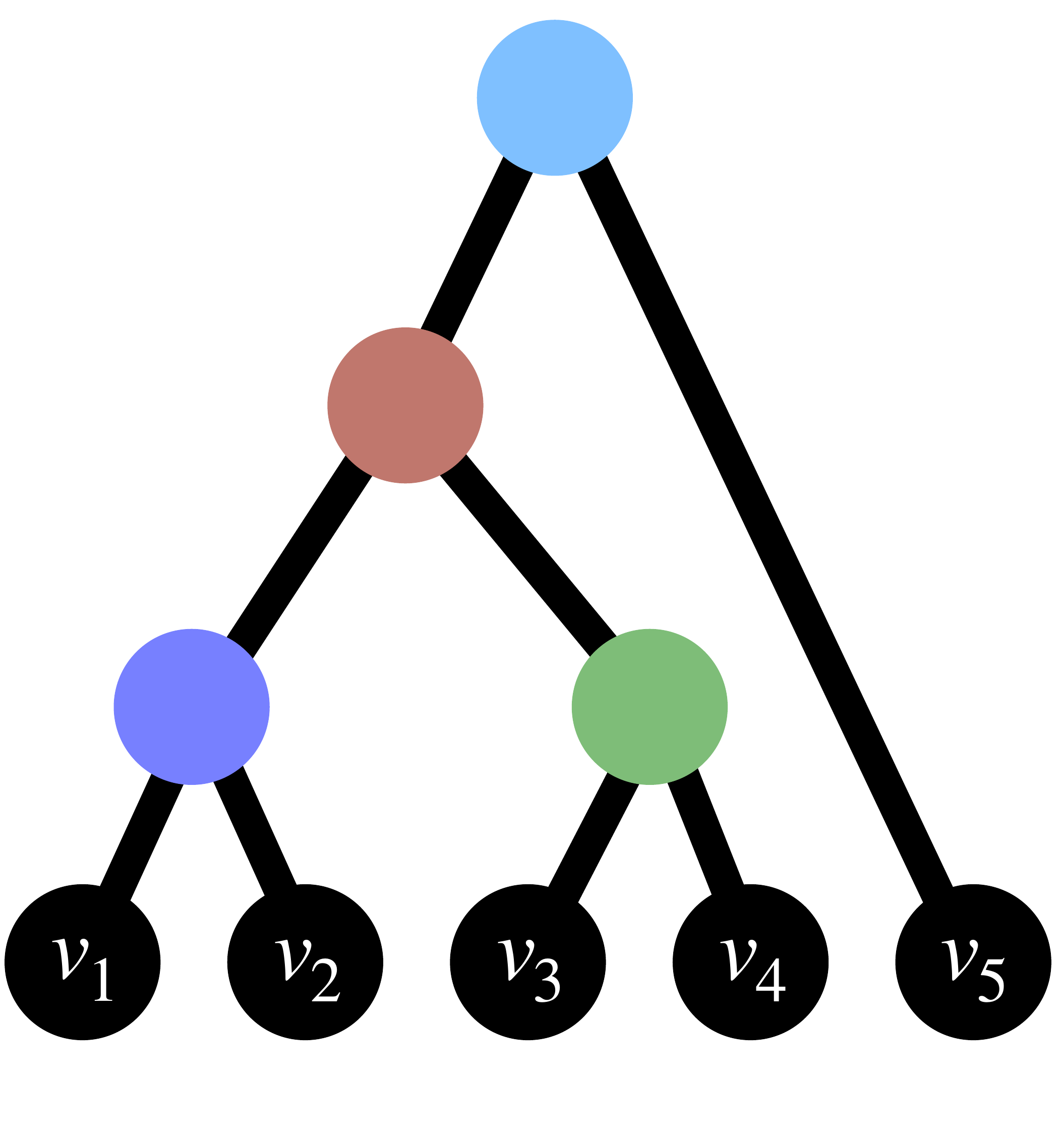}
        \caption{Dendrogram.}\label{sfig:hac3}
    \end{subfigure}   
    \vspace{-0.5em}
    \caption{An example of HAC run on $\mcP \subseteq \mathbb{R}^k$.} \label{fig:HAC}
\end{figure}
Recently, the widespread usage of HAC on large datasets has motivated considerable interest in developing the theory of efficient parallel HAC algorithms. 
For the simpler linkage functions of single and average, efficient parallel algorithms are known. 
More specifically for single-linkage, the problem is reducible to computing the clustering of the Euclidean MST and a recent work gave a work-optimal and poly-logarithmic depth parallel algorithm in the comparison model for this problem~\cite{dhulipala2024optimal}. 
On the other hand, another line of recent work showed that average-linkage HAC on general graphs is $\mathsf{P}$-complete, and thus unlikely to be in $\mathsf{NC}$ \cite{dhulipala2022hierarchical}.
Motivated by this hardness, \cite{dhulipala2022hierarchical} showed that for any constant $\eps > 0$, $(1+\eps)$-approximate average-linkage HAC can be done with $\poly(\log n)$ depth and $\poly(n)$ work\footnote{$f(n) = \poly(n)$ if there exists some constant $c$ such that $f(n) \leq O(n^c)$. Throughout this work we use the standard work-depth model of parallelism, which is equal up to logarithmic factors in the depth to other standard parallel models like different PRAM variants, see, for example, \cite{blelloch2024introduction}.} (i.e., it is in the complexity class NC).\footnote{More generally, they show that this is true even if the input is a graph with arbitrary edge weights.}

Despite the progress over the past few years for single- and average-linkage, there is much less progress on other linkage-functions such as centroid- or Ward's-linkage, especially when $(1+\eps)$-approximation is allowed.
Recently, \cite{lattanzi2020framework} gave an $O(\log^2 n)$-approximate and $\poly(\log n)$-round distributed algorithm for centroid HAC and an $O(1)$-approximate algorithm that performed well in practice but had no provable guarantees on its round-complexity.
Unlike other linkage functions, centroid- and Ward's-linkage are directly related to cluster centroids. Thus, developing fast parallel algorithms for centroid- and Ward's is especially motivated by the important practical applications of clustering that rely on cluster centroids (e.g., when leveraging clustering in practical nearest-neighbor search indices~\cite{subramanya2019diskann, ANNScaling, douze2024faiss}). Furthermore, Ward's-linkage has a close theoretical connection to $k$-means clustering and, in fact, provably yields clusterings that are good approximations of the $k$-means objective under certain assumptions~\cite{grosswendt2019analysis}.


The fact that parallel algorithms for single- and average-linkage exist but not centroid- and Ward's-linkage is explained by the helpful fact that single- and average-linkage are \emph{monotone}, even under arbitrary merges. 
In particular, given a set of clusters $\mcC$, merging \emph{any} pair of clusters in $\mcC$ results in a new set of clusters whose minimum distance according to both $\dSing$ and $\dAvg$ has not decreased. 
This property reduces computing the single-linkage dendrogram to post-processing the edges of the Euclidean minimum spanning tree~\cite{dhulipala2024optimal}. 
For average-linkage, this monotonicity property enables the usage of standard bucketing tricks in parallel algorithms.
In particular, \cite{dhulipala2022hierarchical} used this fact to divide $(1+\eps)$-approximate HAC into phases where the closest pair increases by a multiplicative factor $(1+\eps)$ after each phase. 
However, their work is in the graph setting, where the techniques rely on the fact that the linkage function is determined by the weights of edges between individual nodes, and it is not clear how to extend these ideas to linkage functions defined directly in $\mathbb{R}^k$.


Also, unfortunately, neither centroid- nor Ward's-linkage is monotone under arbitrary merges. For example, even exact centroid merges can reduce the minimum distance between clusters by a multiplicative constant. See, e.g., \Cref{sfig:nonMon1} / \ref{sfig:nonMon2}. Even worse, even just $2$-approximate centroid merges can arbitrarily reduce the minimum distance between two clusters. See \Cref{sfig:nonMon3} / \ref{sfig:nonMon4}. Likewise, Ward's is known to be monotone under minimum merges \cite{grosswendt2019analysis} but one can prove that under even just $(1+\eps)$-approximate merges for arbitrarily small $\eps > 0$, this ceases to be the case. As such, achieving a good notion of progress on which to base a parallel algorithm for both centroid and Ward's appears difficult since for these linkage functions the minimum distance can oscillate wildly over the course of HAC.

\begin{figure}
\begin{subfigure}[b]{0.24\textwidth}
        \centering
        \includegraphics[width=\textwidth,trim=0mm 0mm 0mm 0mm, clip]{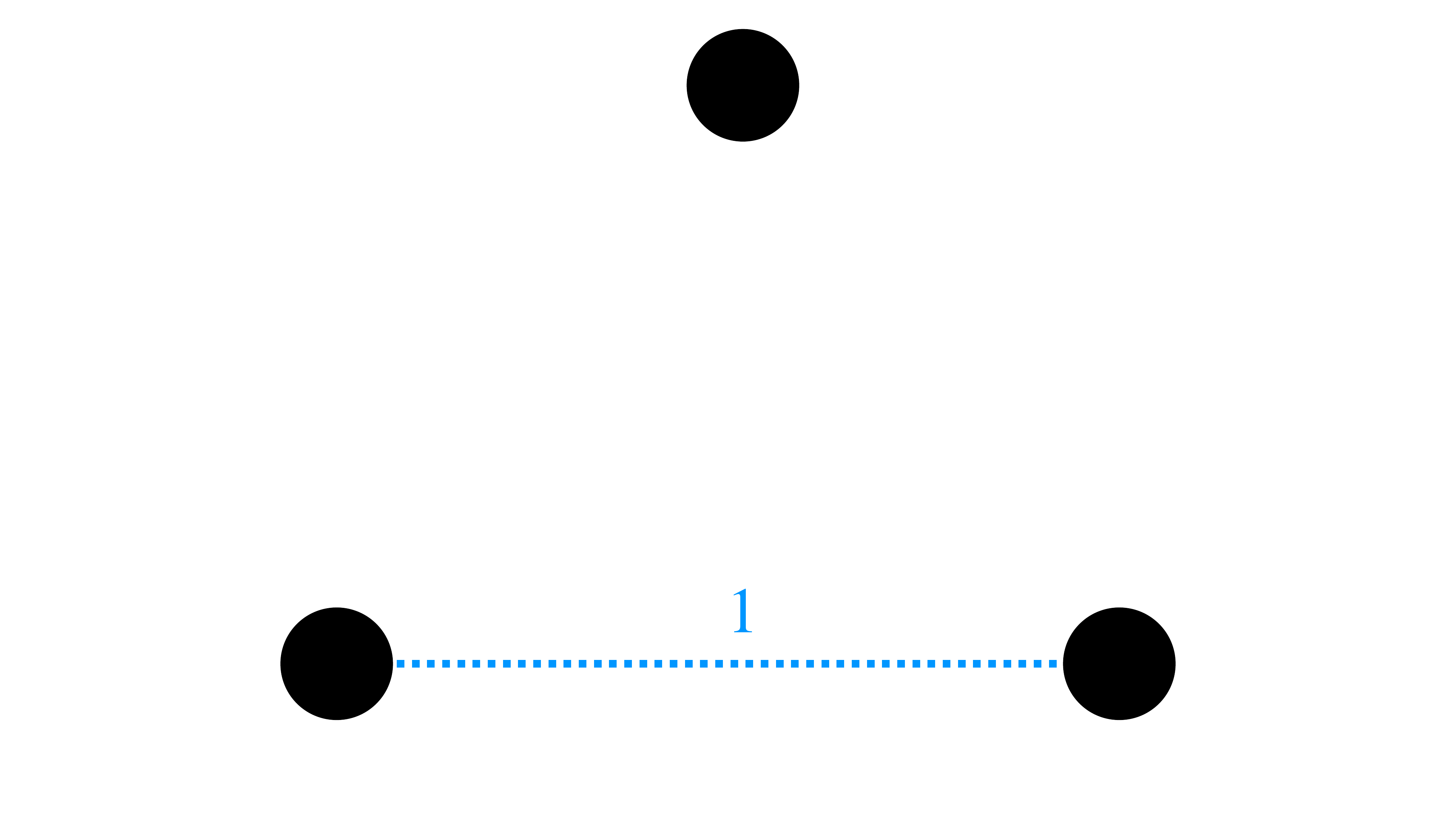}
        \caption{First merge.}\label{sfig:nonMon1}
    \end{subfigure}  \hfill
    \begin{subfigure}[b]{0.24\textwidth}
        \centering
        \includegraphics[width=\textwidth,trim=0mm 0mm 0mm 0mm, clip]{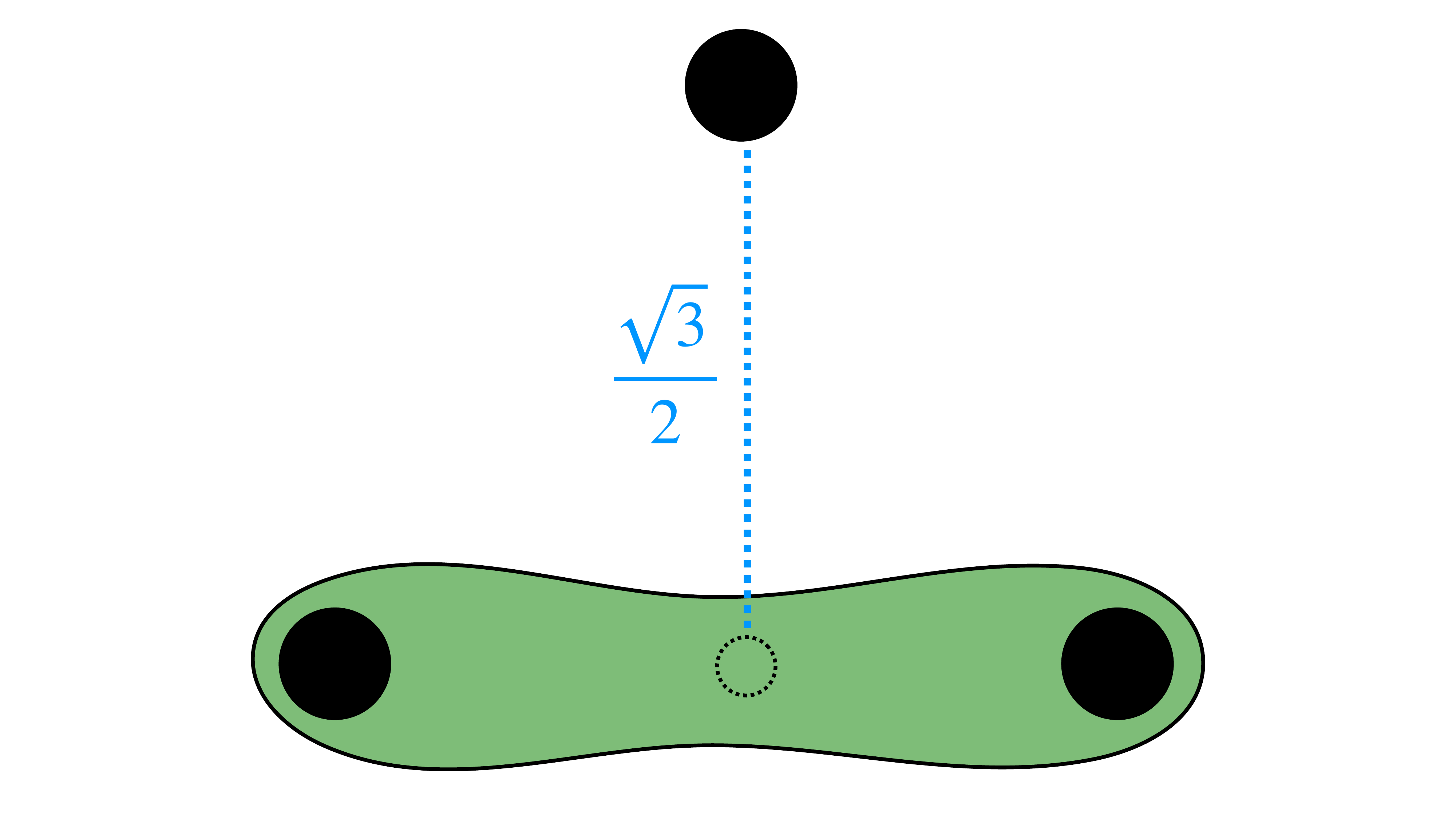}
        \caption{Next merge.}\label{sfig:nonMon2}
    \end{subfigure}   \hfill
    \begin{subfigure}[b]{0.24\textwidth}
        \centering
        \includegraphics[width=\textwidth,trim=0mm 0mm 0mm 0mm, clip]{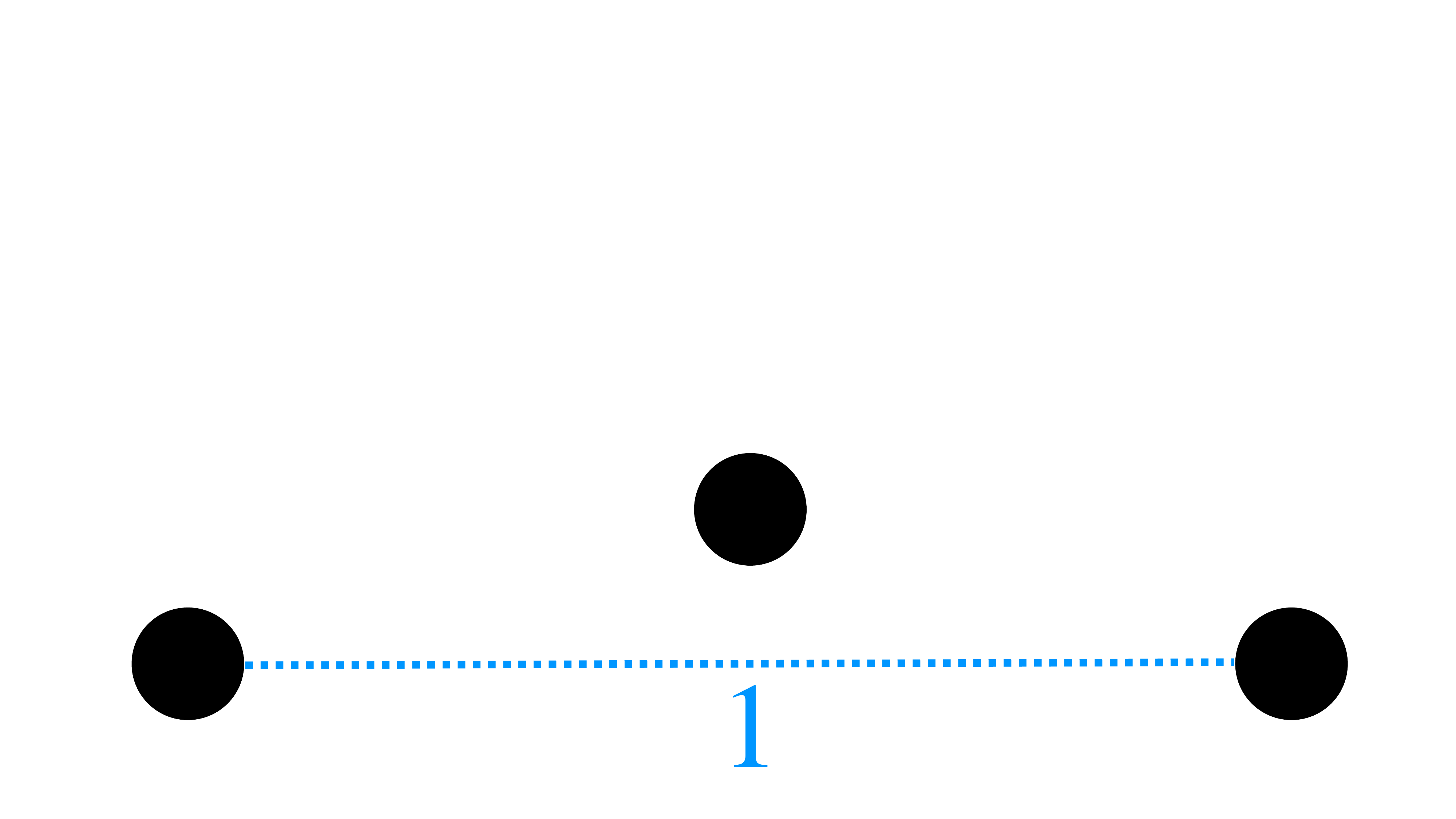}
        \caption{First merge.}\label{sfig:nonMon3}
    \end{subfigure}
    \begin{subfigure}[b]{0.24\textwidth}
        \centering
        \includegraphics[width=\textwidth,trim=0mm 0mm 0mm 0mm, clip]{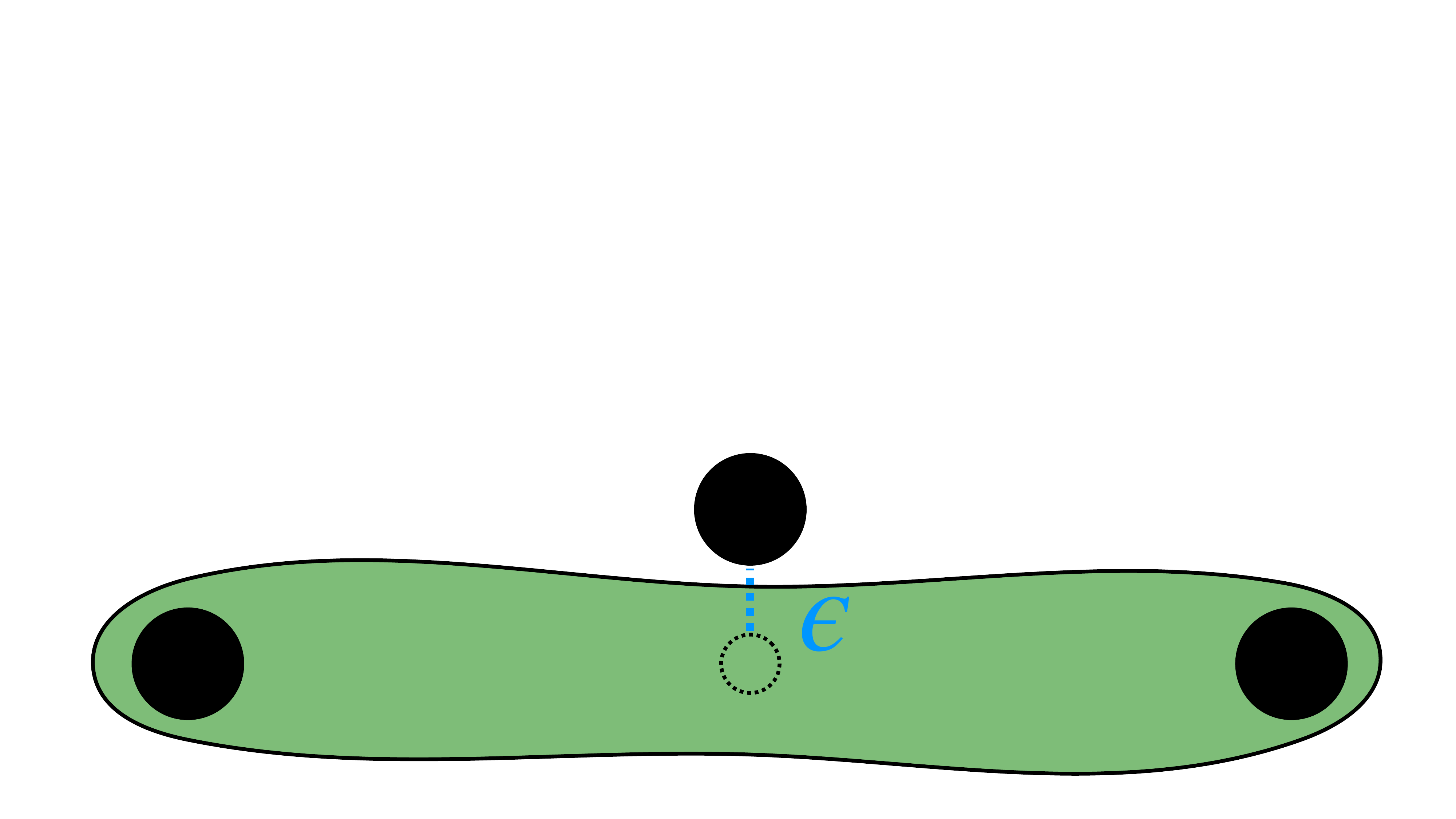}
        \caption{Next merge.}\label{sfig:nonMon4}
    \end{subfigure}
    \caption{Two examples where centroid HAC fails to be monotone.  \ref{sfig:nonMon1} / \ref{sfig:nonMon2} gives $3$ equidistant points in $\mathbb{R}^2$ where the minimum merge reduces from $1$ and $\sqrt{3}/2 < 1$. \ref{sfig:nonMon3} / \ref{sfig:nonMon4} gives $3$ points in $\mathbb{R}^2$ where a $2$-approximate merge reduces the minimum distance from $1$ to an arbitrarily small $\eps$.} \label{fig:nonMon}
\end{figure}

Summarizing, there are no known fast $(1+\eps)$-approximate parallel algorithms for HAC when the linkage function is not monotone. This leads us to the central question of this work:

\begin{quote}\centering \textit{
Are there NC $(1+\eps)$-approximate algorithms for HAC with non-monotone linkage $d$?\\
}\end{quote}

\section{Our Contributions}

In this work, we show that $(1+\eps)$-approximate HAC is in NC for a general class of linkage functions which includes both centroid- and Ward's-linkage whenever the dimensionality is low. Our algorithms are based on a structural result which shows that the height of the dendrogram resulting from HAC with these linkage functions is always small in low dimensions. Complementing this, we show that the assumption of low dimensionality is provably necessary (under standard complexity assumptions) for NC algorithms for these linkage functions.

In more detail, we introduce the notion of \emph{well-behaved linkage functions}. 
Roughly speaking, well-behaved linkage functions are functions which (approximately) exhibit the usual properties of Euclidean distance such as the triangle inequality and well-known packing properties and are additionally stable under ``small'' merges. We show that both centroid and Ward's are well-behaved. Centroid- and Ward's-linkage appear very different in nature---centroid is about the distance between points in Euclidean space and Ward's is about improving a clustering objective---and so it is, perhaps, surprising that the two can be shown to have important common properties which can be exploited for parallelism.

Next, we prove the key structural result of our work which is a proof that well-behaved linkage functions always result in low height dendrograms. In particular, we show that $c$-approximate HAC with a well-behaved linkage function gives a dendrogram of height at most about $\tilde{O}\left((ck)^k\right)$ in 
$\mathbb{R}^k$.\footnote{Throughout this work we use $\tilde{O}$ notation to hide $\poly(\log n)$ terms.} It follows that, for instance, the height of the dendrogram is always at most $\poly(\log n)$ as long as $k = O(\log \log n / \log \log \log n)$ and $c = O(1)$. Likewise, it is not too hard to see that our height bounds are essentially best possible for centroid by considering points placed uniformly $2$ apart around a unit sphere and a ``heavy'' point at the origin in $\mathbb{R}^k$.\footnote{$\Omega(k^k)$ points can be placed on the sphere by packing, and they merge into the center one by one.} 

We then leverage the low height of the dendrogram to design NC algorithms. In particular, we show that if the dendrogram has height $h$ for a well-behaved linkage function, then 
there exist parallel algorithms for $(1+\eps)$-approximate HAC with $\tilde{O}(h\ell^k)$ depth, where $\ell$ is an auxiliary parameter (discussed later) that is always at most $h$.
Combining this with our height bounds, and utilizing state-of-the-art parallel nearest neighbor search (NNS) data structures, yields $\tO(n)$ work NC algorithms for $(1+\epsilon)$-approximate centroid linkage when $k=O(1)$ and Ward's linkage when $k=O(\log\log n/\log\log\log n)$. 
Our algorithms are thus {\em nearly work-efficient}\footnote{A parallel algorithm is \emph{nearly work-efficient} if its work matches that of the best-known sequential algorithm up to polylogarithmic factors.} 
compared to existing results for approximate centroid and Ward's HAC in the sequential setting~\cite{abboudhac, bateni2024efficient}.

Note that, even for $k=1$, it is not clear that parallel algorithms should be possible for centroid linkage. For instance, consider $n$ points placed along the line where the distance between the $i$th and $(i+1)$th point is $1+i \cdot \eps$ for some small $\eps > 0$ as in \Cref{fig:path}. Initially, the $(i+1)$th point would like to merge with the $i$th point for every $i$. As such, there is a chain of ``dependencies'' of length $\Omega(n)$ which must be resolved before the $(n-1)$th point can know if it should next merge with the $(n-2)th$ point or the $n$th point. Such examples would seem to preclude efficient low-depth parallel algorithms.

\begin{figure}
\begin{subfigure}[b]{0.48\textwidth}
        \centering
        \includegraphics[width=\textwidth,trim=0mm 0mm 0mm 280mm, clip]{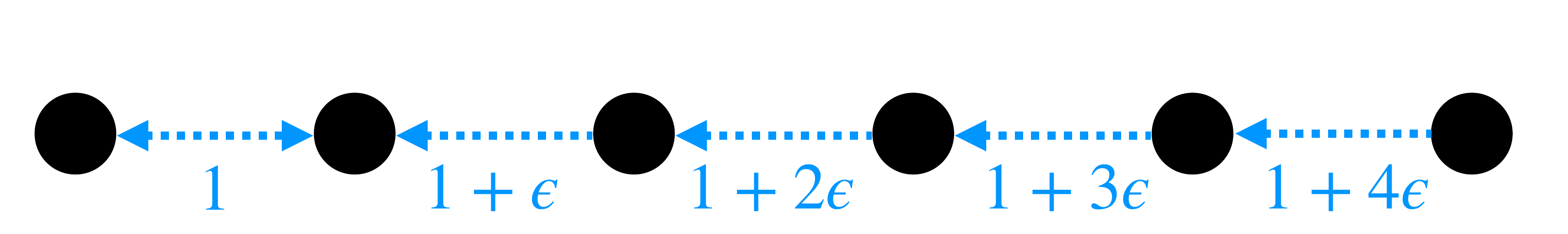}
        \caption{Initial $\mcP$.}\label{sfig:path1}
    \end{subfigure}  \hfill
    \begin{subfigure}[b]{0.48\textwidth}
        \centering
        \includegraphics[width=\textwidth,trim=0mm 0mm 0mm 280mm, clip]{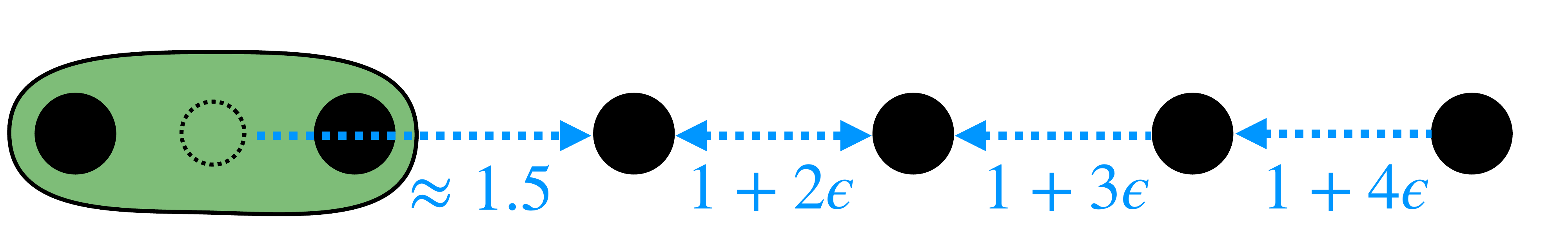}
        \caption{First merge.}\label{sfig:path2}
    \end{subfigure}   \hfill
    \begin{subfigure}[b]{0.48\textwidth}
        \centering
        \includegraphics[width=\textwidth,trim=0mm 0mm 0mm 280mm, clip]{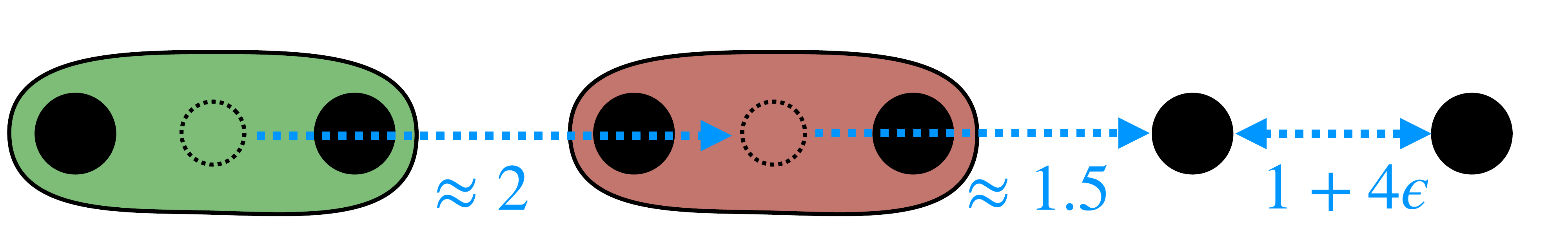}
        \caption{Second merge.}\label{sfig:path3}
    \end{subfigure}\hfill
    \begin{subfigure}[b]{0.48\textwidth}
        \centering
        \includegraphics[width=\textwidth,trim=0mm 0mm 0mm 280mm, clip]{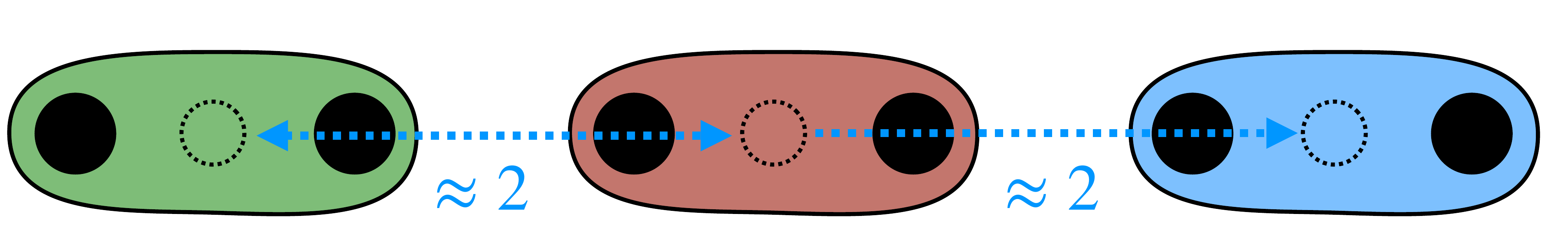}
        \caption{Third merge.}\label{sfig:path4}
    \end{subfigure}
    \caption{An example of centroid HAC with an $\Omega(n)$-length chain of dependencies. Each point/cluster points at its nearest other cluster. Notice that the $(n-1)$th point does not know if it merges left or right until $\Omega(n)$ merges have been performed.} \label{fig:path}
\end{figure}

Lastly, we show that the assumption of low dimensionality is necessary for NC algorithms for well-behaved linkage functions. In particular, we show that HAC with centroid-linkage is hard for the class Comparator Circuit (CC) when $k$ is linear in $n$. CC-hardness is widely believed to rule out NC algorithms \cite{cook2014complexity, mayr1992complexity,subramanian1994new}. While there exist a number of complexity results for HAC \cite{ dhulipala2022hierarchical, bateni2024s}, these results only hold for more general graph versions of HAC; to our knowledge, our hardness result is the first hardness for HAC in Euclidean space.

In what remains of this section, we give a more formal description of our results and discuss some of the techniques and challenges in proving them.

\subsection{Well-Behaved Linkage Functions}

We begin by defining the properties that are satisfied by a well-behaved linkage function.\footnote{We note that most of our proofs work even if the below definitions have larger constants or even extra $\poly (\log \log n)$ or $\poly (\log n)$ slack in some cases; for simplicity of presentation we've generally stated things in terms of small fixed constants. In several places we've noted the tolerance that our definitions allow.}

Our first two properties generalize well-known properties of Euclidean distance to linkage functions. The first of these properties is what we call $\alpha$-packability. In particular, it is well-known that any collection of points in $\mathbb{R}^k$ that are pairwise at least some distance $r$ apart but contained in a radius $R$ ball consists of at most  $\left(\frac{R}{r}\right)^{O(k)}$ points (see \Cref{thm:packPoints}). $\alpha$-packability simply generalizes this to an arbitrary linkage function with a multiplicative slack of $\alpha$. Below and throughout this work, we let $B_d^\mcC(A,R) := \{ C \in \mcC : d(A,C) \leq R\}$ be the radius $R$ ball centered at $A$ with respect to clusters $\mcC$ and linkage distances $d$ where $A$ is the center of the ball.
\begin{restatable}[$\alpha$-Packability]{definition}{packableLinkage} \label{dfn:packableLinkage}
Linkage function $d$ is $\alpha$-packable if for all $r > 0$ and $\mcC \subseteq 2^{\mathbb{R}^k}$ such that $d(B,C) \geq r$ for all distinct $B,C \in \mcC$, we have for every $A \subseteq \mathbb{R}^k$ and $R \geq r$ that
\begin{align*}
    \left|B_d^\mcC(A,R) \right| \leq \alpha \cdot \left(\frac{R}{r} \right)^{O(k)}.
\end{align*}
\end{restatable}
\noindent Proving $O(1)$-packability for centroid is straightforward as centroid distances are just given by Euclidean distances between points in $\mathbb{R}^k$ and these pack as described above. However, it is much less clear for Ward's since Ward's is about greedy improvement of the $k$-means objective. Nonetheless, we show that Ward's is $\alpha$-packable for $\alpha = O(\log n)$ by using alternate characterizations of Ward's linkage. We note that, in fact, the only place we will use low dimension in all of our proofs is to bound the cardinality of balls of the above form.

Next, we appropriately generalize the triangle inequality to linkage functions. Again, since centroid-linkage is given by the Euclidean distance between points, it trivially satisfies the triangle inequality. However, it is not too hard to see that Ward's linkage can be arbitrarily far from satisfying the triangle inequality. Fortunately, our bounds on the dendrogram height (and therefore our definition of well-behaved linkage functions) only require the triangle inequality when the middle cluster is not the minimum size cluster. For this restricted setting, Ward's can be shown to satisfy the triangle inequality approximately. This gives us the following notion of an approximate triangle inequality.
\begin{restatable}[Approximate Triangle Inequality]{definition}{apxMetric} \label{dfn:apxMetric}
We say linkage function $d$ approximately satisfies the triangle inequality if
\begin{align*}
    d(A,C) \leq c_{\Delta} \cdot \left(d(A,B) + d(B,C) \right)
\end{align*}
for some fixed constant $c_{\Delta}$ for every $A, B, C \subseteq \mathbb{R}^k$ such that $|B| \geq \min(|A|,|C|)$.\footnote{In fact, all of our proofs work even if $c_{\Delta} = O(\log \log n)$; only the algorithm requires $c_\Delta=O(1)$.}
\end{restatable}


Our last $3$ properties can be seen as stability properties. In particular, each states that over the course of one or many merges, the linkage function should not vary too wildly.

The first of these is what we call weight-stability. Roughly speaking, a linkage function is weight-stable if when a cluster merges with a much smaller cluster, this ``moves'' the cluster only a small amount. Since we are dealing with general clusters which cannot readily be summarized by a single point, the formal sense of how much a cluster ``moves'' is based on the linkage function distance between the cluster before and after the merge, as described below.\footnote{We note all of our proofs can be made to work even with any fixed constant in front of $\frac{|B|}{|A| + |B|} \cdot d(A, B)$.}
\begin{restatable}[Weight-Stable]{definition}{weightStab} \label{dfn:weightStab}
A linkage function $d$ is weight-stable if for any $A, B \subseteq \mathbb{R}^k$ we have
\begin{align*}
    d(A \cup B, A) \leq \frac{|B|}{|A| + |B|} \cdot d(A, B).
\end{align*}
\end{restatable}
\noindent Roughly speaking, weight-stability holds for centroid because when $A$ merges with $B$, the new centroid can be ``dragged'' at most the relative size of $B$ times how far the centroid of $B$ was from $A$. For Ward's, it follows from the \emph{Lance-Williams} characterization of Ward's linkage distance after a merge \cite{lance1967general} (see~\Cref{lem:LanceWilliamsWards}).

Our next stability property says that after performing a merge, our new distances to other clusters should be (up to the magnitude of the merge performed) at least the average of our prior distances. 
We call this property average-reducibility.\footnote{All of our proofs work if the inequality holds for any fixed convex combination bounded away from $1$. In particular, our proofs work if there is a constant $c \in (0,1)$ such that $d(A \cup B, C) \geq c\cdot d(A,C) + (1-c) \cdot d(B,C) - d(A,B)$.}
\begin{restatable}[Average-Reducibility]{definition}{avgRed} \label{dfn:avgRed}
Linkage function $d$ is average-reducible if for any $A, B, C \subseteq \mathbb{R}^k$ such that $|C| \geq |A|+ |B|$, we have
\begin{align*}
    d(A \cup B, C) \geq \frac{d(A,C) + d(B,C)}{2} - d(A,B).
\end{align*}
\end{restatable}
\noindent The name average-reducibility comes from the fact that average-reducibility closely resembles the well-studied property of reducibility, which states that if $A$ and $B$ are a minimum merge we have $d(A \cup B, C) \geq \min (d(A,C), d(B,C))$ \cite{lance1967general}. For centroid, average-reducibility follows from the fact that the centroid of two centroids lies at the weighted midpoint between the two along with the triangle inequality. For Ward's it is not true in general but, assuming $|C| \geq |A|+ |B|$, it follows from the Lance-Williams characterization.

Our last stability property is arguably the most straightforward. In particular, we say that a linkage function has poly-bounded diameter if, given a collection of clusters each consisting of a single point, the maximum linkage distance between any two subsets of these points is at most polynomially-larger. Below, and throughout this work, for $u, v \in \mathbb{R}^k$, we let $d(u, v) := d(\{u\}, \{v\})$.
\begin{restatable}[Poly-Bounded Diameter]{definition}{boundDiam} \label{dfn:polyDiam}
We say linkage function $d$ has poly-bounded diameter if given any $\mcP \subseteq \mathbb{R}^k$ such that $\Delta = \max_{u,v \in \mcP}d(u, v)$, we have $d(A,B) \leq \poly(\Delta \cdot |\mcP|)$ for any disjoint $A,B \subseteq \mcP$.
\end{restatable}
\noindent The above is true for centroid since all centroids of subsets of a $\mcP \subseteq \mathbb{R}^k$ lie in the convex hull of $\mcP$. It is similarly easy to show for Ward's.

If a linkage function satisfies the above properties, then we say that it is well-behaved.
\begin{restatable}[Well-Behaved]{definition}{wellBeh} \label{dfn:wellBehaved}
We say linkage function $d$ is well-behaved if it is $\alpha$-packable for $\alpha = \tilde{O}(1)$ (\Cref{dfn:packableLinkage}), approximately satisfies the triangle inequality (\Cref{dfn:apxMetric}), is weight-stable (\Cref{dfn:weightStab}), average-reducible (\Cref{dfn:avgRed}) and has poly-bounded diameter (\Cref{dfn:polyDiam}).
\end{restatable}
\noindent As discussed above, we show that both centroid and Ward's are well-behaved.
\begin{restatable}{theorem}{centroidWB}\label{thm:centroidWB}
    The centroid linkage function $\dCen$ is well-behaved (\Cref{dfn:wellBehaved}).
\end{restatable}

\begin{restatable}{theorem}{wardsWB}\label{thm:wardWB}
    Ward's linkage function $\dWard$ is well-behaved (\Cref{dfn:wellBehaved}).
\end{restatable}


We prove the above theorems in \Cref{sec:packablelinkage}.

\subsection{Bounds on Dendrogram Height}
Having formally defined well-behaved linkage functions, we now state our dendrogram height bound which says, roughly, that any $c$-approximate HAC resulting from a well-behaved linkage function in $\mathbb{R}^k$ has height $\tilde{O}\left( (ck)^k\right)$. 

To formally state our result, we must first formalize some points. First, our height bounds will hold regardless of how ties are broken for HAC. In particular, at any given step of a $c$-approximate HAC there might be multiple candidate $A$ and $B$ such that $d(A, B) \leq c \cdot \min_{C, D \in \mcC} d(C,D)$ (even for $c = 1$). We say that \emph{any} $c$-approximate HAC has a dendrogram height at most $h$ if regardless of how we make these choices, the result of HAC is always a dendrogram with height at most $h$. Likewise, we define the aspect ratio of an instance of HAC given by point set $\mcP \subseteq \mathbb{R}^k$ and linkage function $d$ as
\begin{align*}
        \frac{\max_{u,v \in \mcP} d(u,v)}{\min_{u,v \in \mcP} d(u,v)}.
\end{align*}
With the above formalism in hand, we can now formally define  our key structural result bounding the height of HAC dendrograms for well-behaved functions.
\begin{restatable}{theorem}{heightBound}\label{lem:heightBound}
    Suppose $d$ is a well-behaved linkage function (as defined in \Cref{dfn:wellBehaved}). Then any $c$-approximate HAC for $d$ with $\poly(n)$ aspect ratio has a dendrogram of height $\tilde{O} \left(( k  \cdot c)^{O(k)} \right)$ in $\mathbb{R}^k$. 
\end{restatable}
\noindent As an immediate consequence of the above, we have that $c$-approximate HAC for centroid and Ward's has height at most $\tilde{O} \left(( k  \cdot c)^{O(k)} \right)$ in $\mathbb{R}^k$, assuming $\poly(n)$ aspect ratio. More specifically, above, our $\tilde{O}$ notation hides only $O(\alpha \cdot \log n)$ where $\alpha$ is the packability as in \Cref{dfn:packableLinkage}. For centroid we have $\alpha = O(1)$ and for Ward's we have $\alpha = O(\log n)$, giving respective final height bounds of $O(\log n \cdot (k \cdot c)^{O(k)})$ and $O(\log^2 n \cdot (k \cdot c)^{O(k)})$. These height bounds may be interesting in their own right as they, perhaps, explain the practical utility of these linkage functions---if the underlying dimension of the data is low then they tend to produce balanced hierarchies.

\subsubsection{Intuition for Bounds on Dendrogram Height}
Our proof bounding the dendrogram height is the most technical part of our work and is based on a potential function argument. In particular, we fix an arbitrary point $x_0 \in \mcP$ and then argue that the number of times the cluster containing $x_0$ participates in a merge is at most $\tilde{O} \left(( k  \cdot c)^{O(k)} \right)$. In order to do so, we divide HAC into $\tilde{O} \left(( k  \cdot c)^{O(k)} \right)$ phases where in each phase the cluster containing $x_0$ participates in at most $\tilde{O} \left(( k  \cdot c)^{O(k)} \right)$-many merges. To bound the number of merges in each phase, we define a suitable non-increasing potential which starts at $\tilde{O} \left(( k  \cdot c)^{O(k)} \right)$ and reduces by $\Omega(1)$ each time the cluster containing $x_0$ participates in a merge.

In order to motivate and give the intuition behind this potential function argument, in the rest of this section we very roughly sketch the argument for exact (i.e.\ $1$-approximate) centroid HAC. Centroid HAC has a convenient interpretation where, instead of maintaining clusters, we simply maintain the centroids of clusters, associating with each centroid a weight equal to the size of the corresponding cluster. Centroid HAC then repeatedly takes the two closest centroids and merges them into a new centroid which is the weighted average of the merged centroids and whose weight is the sum of the merged centroids.

For the rest of this sketch, let $X$ be the cluster containing $x_0$ and let $x$ be its centroid. Suppose, for the moment, that:
\begin{enumerate}
    \item[(1)] The weight of $x$ is far larger than the weight of any other centroid and;
    \item[(2)] When we merge two centroids $y$ and $z$ that are not equal to $x$, then the newly created centroid is at least as far from $x$ as both $y$ and $z$ were.
\end{enumerate}
Note that, by (1), it follows that when $x$ merges with another centroid, the new centroid (which we will still refer to as ``$x$'') lies essentially at the same place that $x$ was prior to the merge.

Under these assumptions, consider the state of our centroids when $x$ merges with a centroid at some distance $\delta$. Then, by definition of (exact) HAC, we know that every pair of centroids are at distance at least $\delta$. By packability (\Cref{dfn:packableLinkage}), it follows that the number of centroids within distance $2\delta$ of $x$ at this moment is at most $2^{O(k)}$. Since we are assuming that $x$ has large weight and so does not move when it merges, and newly created centroids cannot get closer to $x$, the number of merges that $x$ can perform until there are no other centroids within $2\delta$ of it is at most $2^{O(k)}$. Thus, after performing $2^{O(k)}$ merges, $x$ has its closest centroid distance increase by a factor of $2$ and, assuming the maximum centroid distance of a merge that $x$ can perform is polynomially-bounded, this can happen at most $O(\log n)$ times, giving a height bound of at most $\log n \cdot 2^{O(k)}$. See \Cref{sfig:circ1}/\ref{sfig:circ2}/\ref{sfig:circ3}/\ref{sfig:circ4}/\ref{sfig:circ5}. We next discuss how to dispense with the above $2$ assumptions. 

For (1), notice that all we are using is that the initial position of $x$ when it first performs a distance $\delta$ merge is very close to the position of $x$ when it first performs a distance $\geq 2 \delta$ merge. Since $x$ only merges with centroids at distance at most $2\delta$ from it for this period, it follows that the only way $x$ can move more than $\delta$ from its initial position during this period is if the total weight of what $x$ merges with is on the order of the weight of $x$ itself. In other words, roughly speaking, (1) only fails to hold if $x$ has its weight increase by some multiplicative constant. Since weights are at least $1$ and bounded by $n$, this can only happen $O(\log n)$ times which is still consistent with a $\log n \cdot 2^{O(k)}$ height bound. Summarizing the above, we may divide HAC into about $O(\log n)$ phases where in the $j$th phase $x$ does not drift far from its starting position and always performs merges of distance at most $\approx 2^j$ and in each phase $x$ performs at most $2^{O(k)}$ merges.

Dispensing with (2) is more difficult. Specifically, consider a single phase where we are performing merges of value at most $2\delta$. 
Since (even exact) centroid HAC is not monotone, when two clusters, say $y$ and $z$, at distance about $\delta$ from each other merge, they can move up to (about) $\delta$ closer to $x$. As such, even if there is nothing within $2\delta$ of $x$ at one point in HAC, a pair can later merge, produce a new centroid within $2\delta$ of $x$, and reduce the minimum distance of $x$ back below $2 \delta$, breaking the above argument.

However, notice that in order to generate a new centroid at distance $\delta$ from $x$ we must merge \emph{two} centroids within distance $2\delta$ from $x$ (assuming we are doing exact HAC). More generally, to generate a new centroid at distance $i \delta$ from $x$, we must merge two centroids at distance at most $(i+1)\delta$; see \Cref{sfig:circ6}/\ref{sfig:circ7}. In this way, centroids at distance $i\delta$ could eventually drift within $2\delta$ of $x$ but in order to do so the number of centroids at distance $i \delta$ with which they have to merge is roughly $2^i$. As such, in terms of violating monotonicity, a centroid at distance $i\delta$ ``counts'' for $1/2^i$ of a centroid at, say, distance $\delta$ from $x$. See \Cref{sfig:circ6}/\ref{sfig:circ7}/\ref{sfig:circ8}.

\begin{figure}
\begin{subfigure}[b]{0.24\textwidth}
        \centering
        \includegraphics[width=\textwidth,trim=0mm 0mm 300mm 0mm, clip]{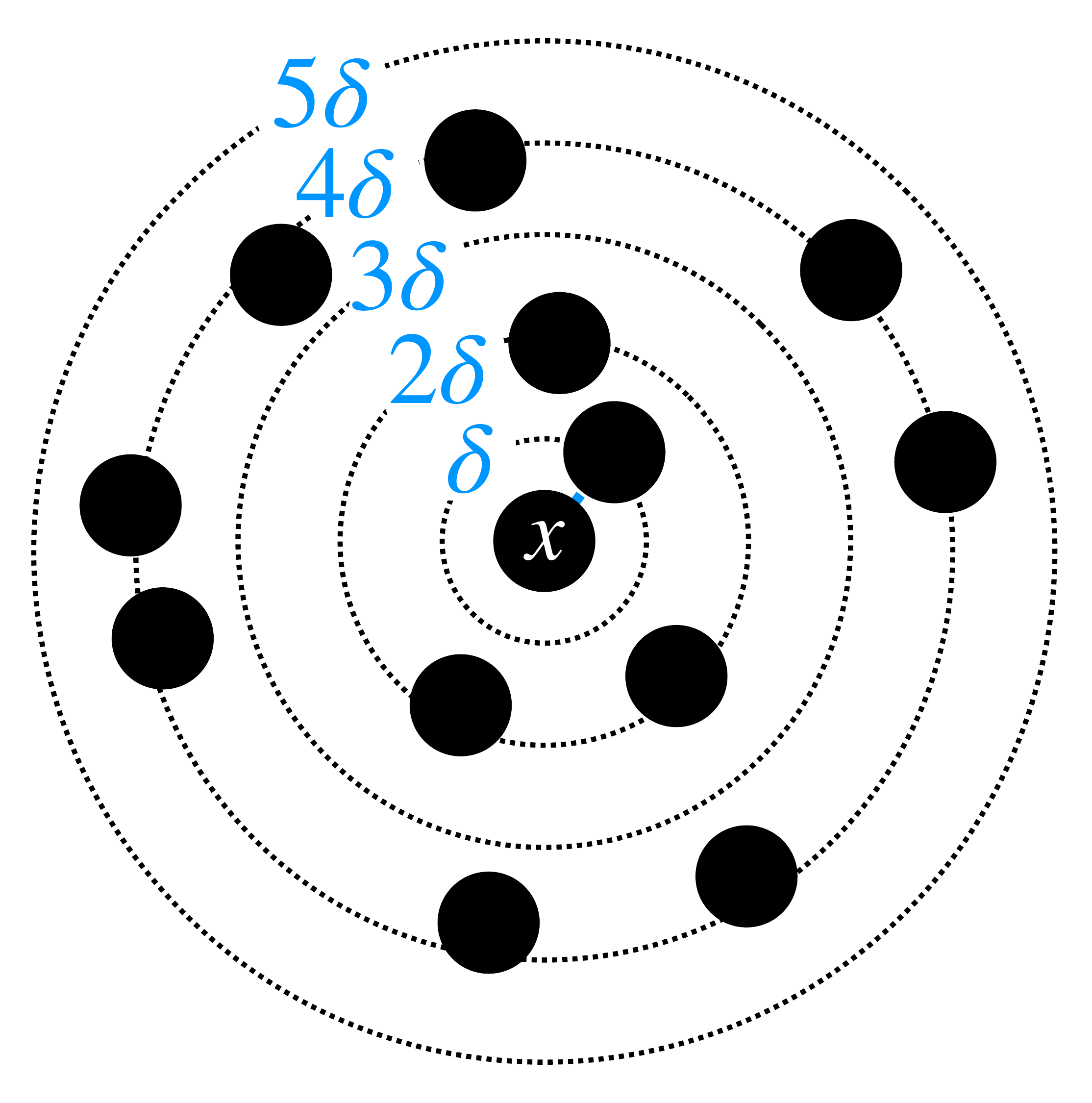}
        \caption{Initial $\mcP$.}\label{sfig:circ1}
    \end{subfigure}  \hfill
    \begin{subfigure}[b]{0.24\textwidth}
        \centering
        \includegraphics[width=\textwidth,trim=0mm 0mm 300mm 0mm, clip]{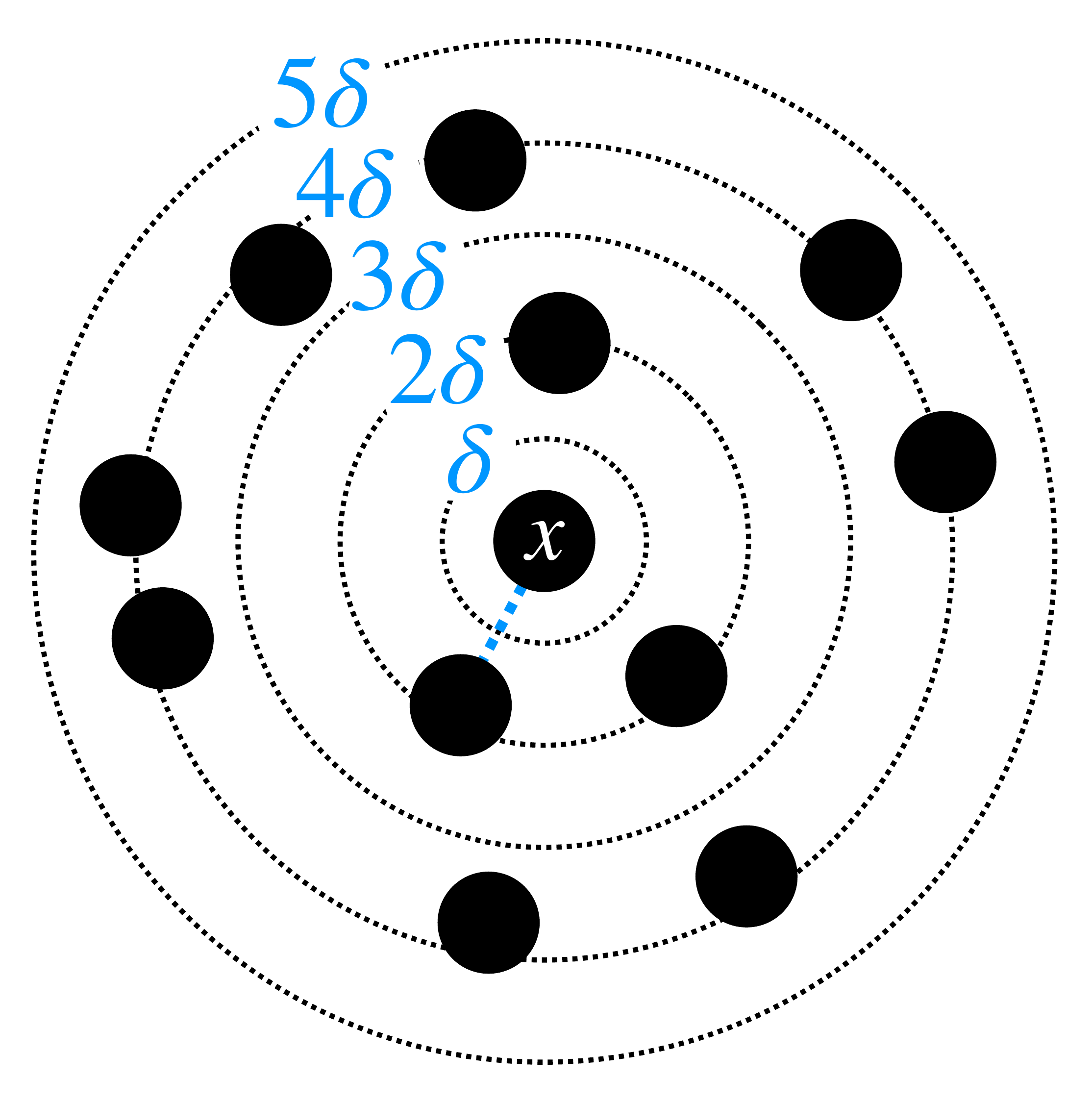}
        \caption{First merge.}\label{sfig:circ2}
    \end{subfigure}   \hfill
    \begin{subfigure}[b]{0.24\textwidth}
        \centering
        \includegraphics[width=\textwidth,trim=0mm 0mm 300mm 0mm, clip]{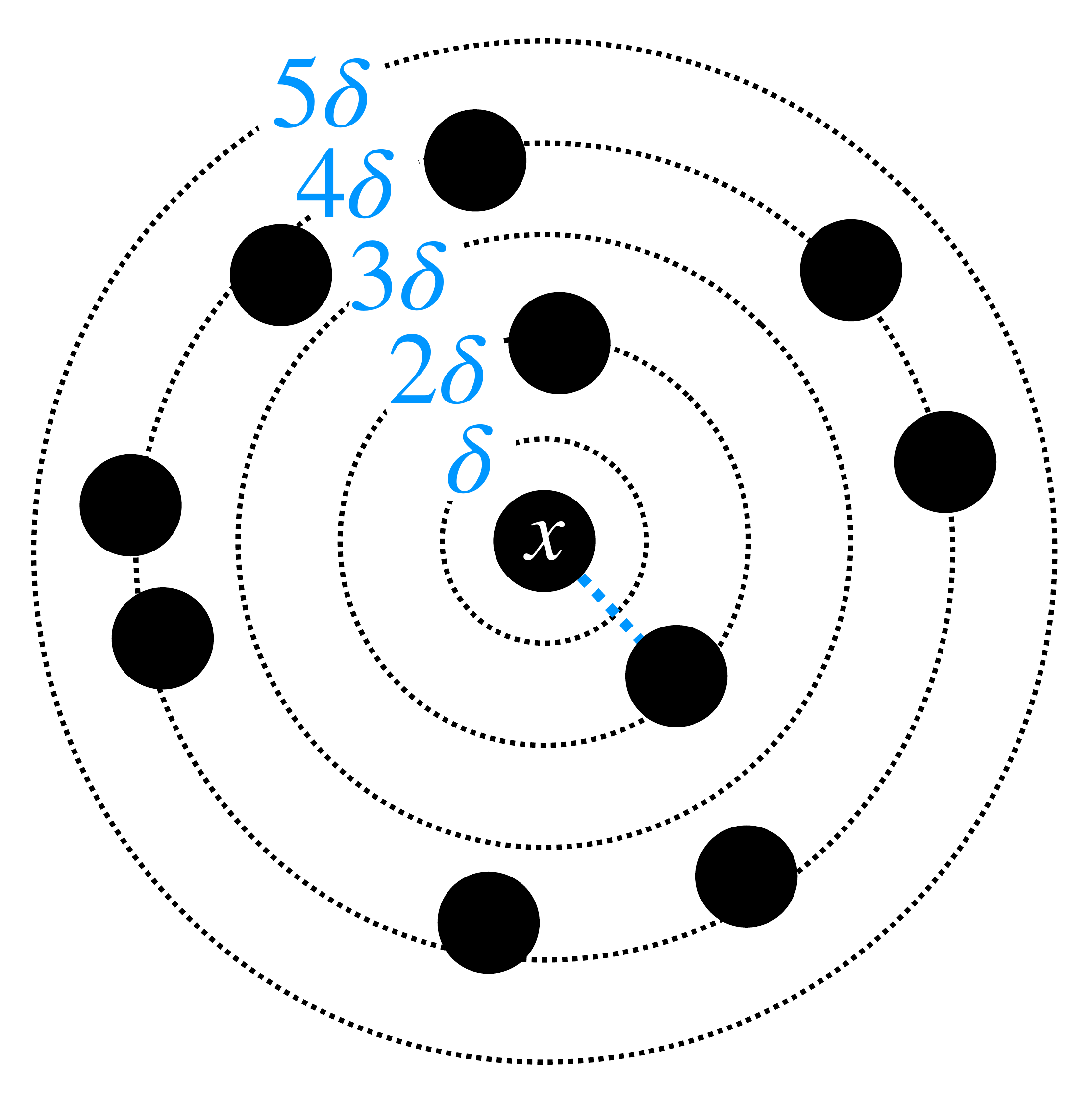}
        \caption{Second merge.}\label{sfig:circ3}
    \end{subfigure}\hfill
    \begin{subfigure}[b]{0.24\textwidth}
        \centering
        \includegraphics[width=\textwidth,trim=0mm 0mm 300mm 0mm, clip]{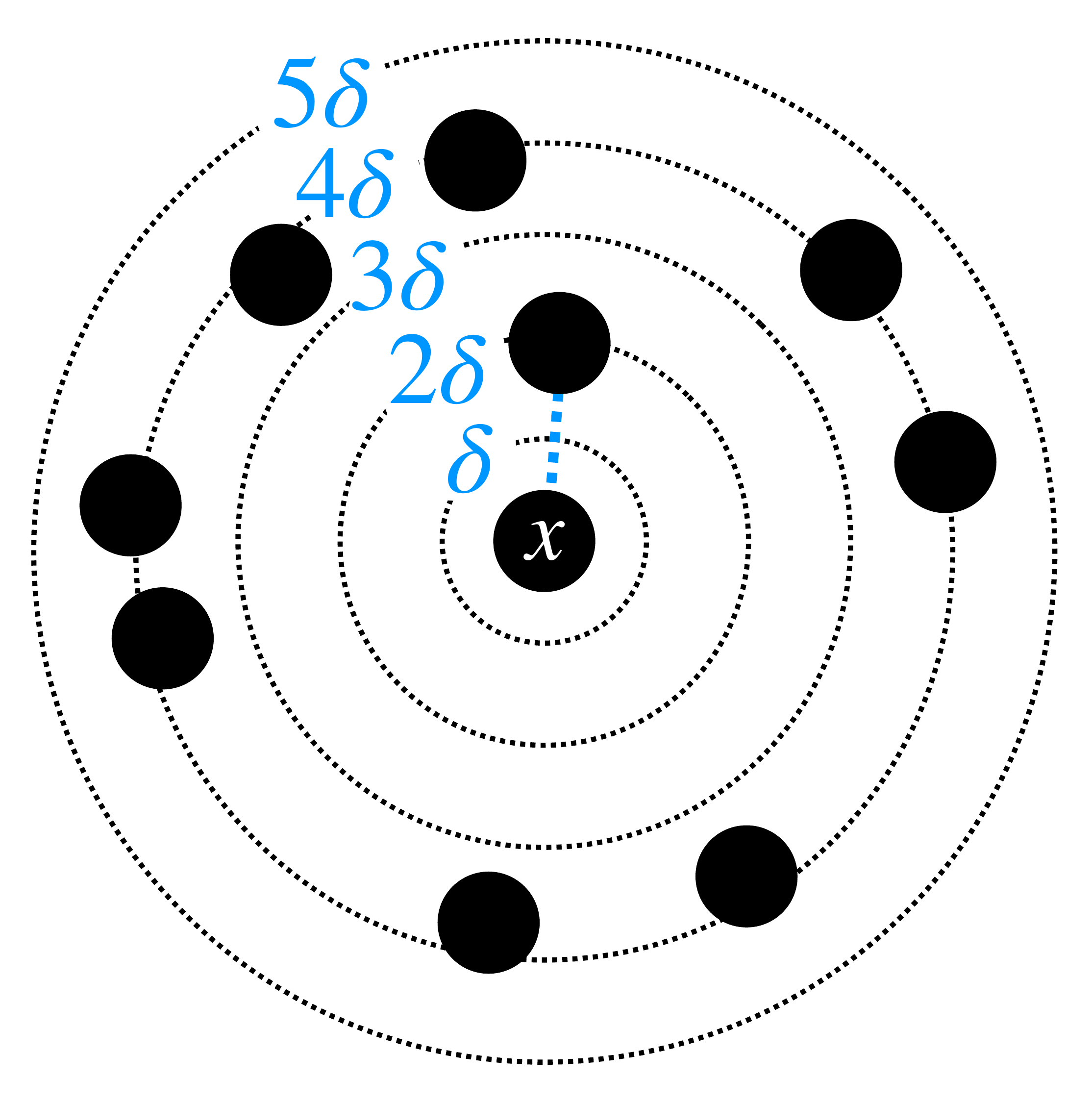}
        \caption{Third merge.}\label{sfig:circ4}
    \end{subfigure}
    \begin{subfigure}[b]{0.24\textwidth}
        \centering
        \includegraphics[width=\textwidth,trim=0mm 0mm 300mm 0mm, clip]{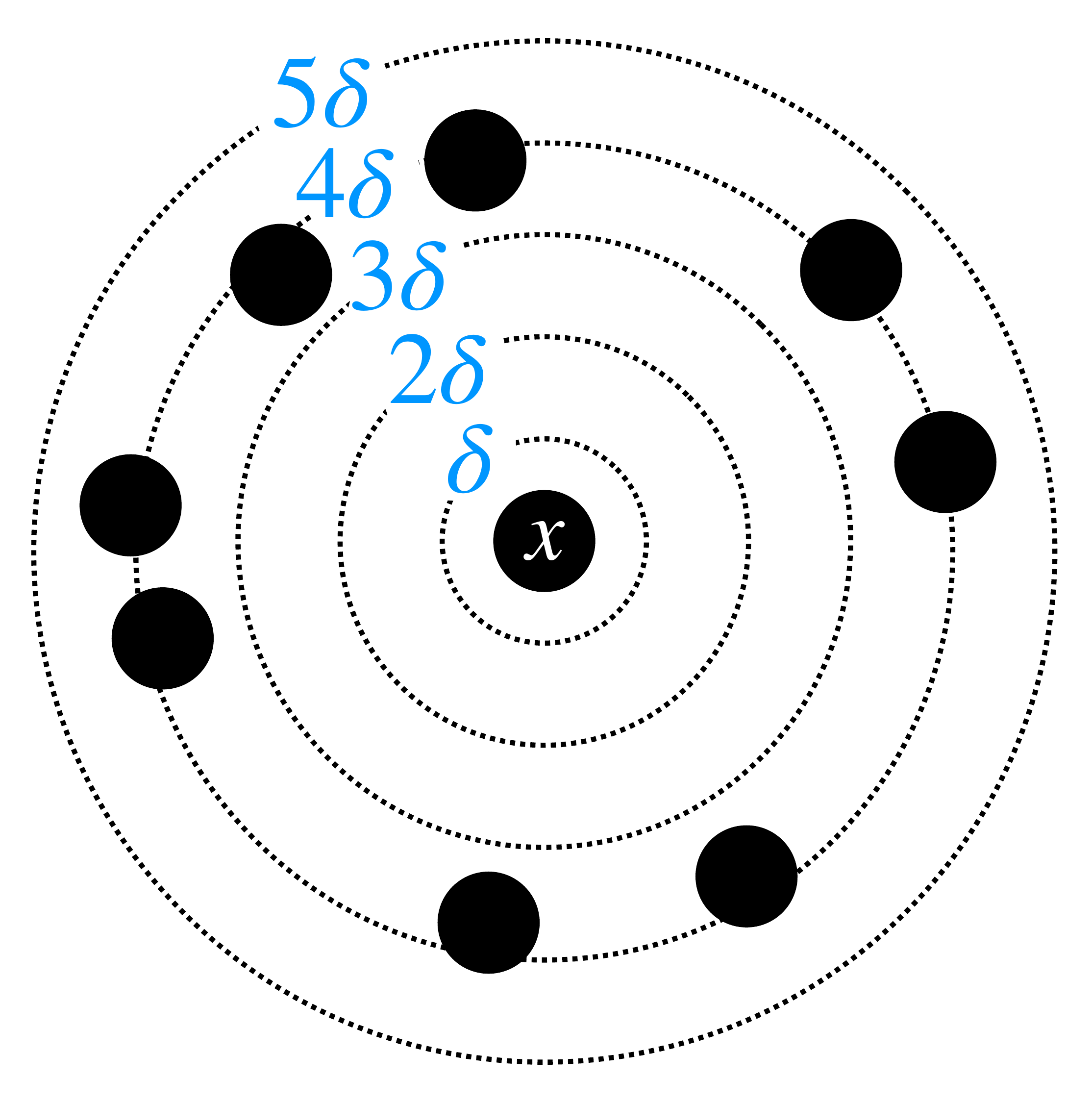}
        \caption{Empty $2\delta$ ball.}\label{sfig:circ5}
    \end{subfigure}  \hfill
    \begin{subfigure}[b]{0.24\textwidth}
        \centering
        \includegraphics[width=\textwidth,trim=0mm 0mm 300mm 0mm, clip]{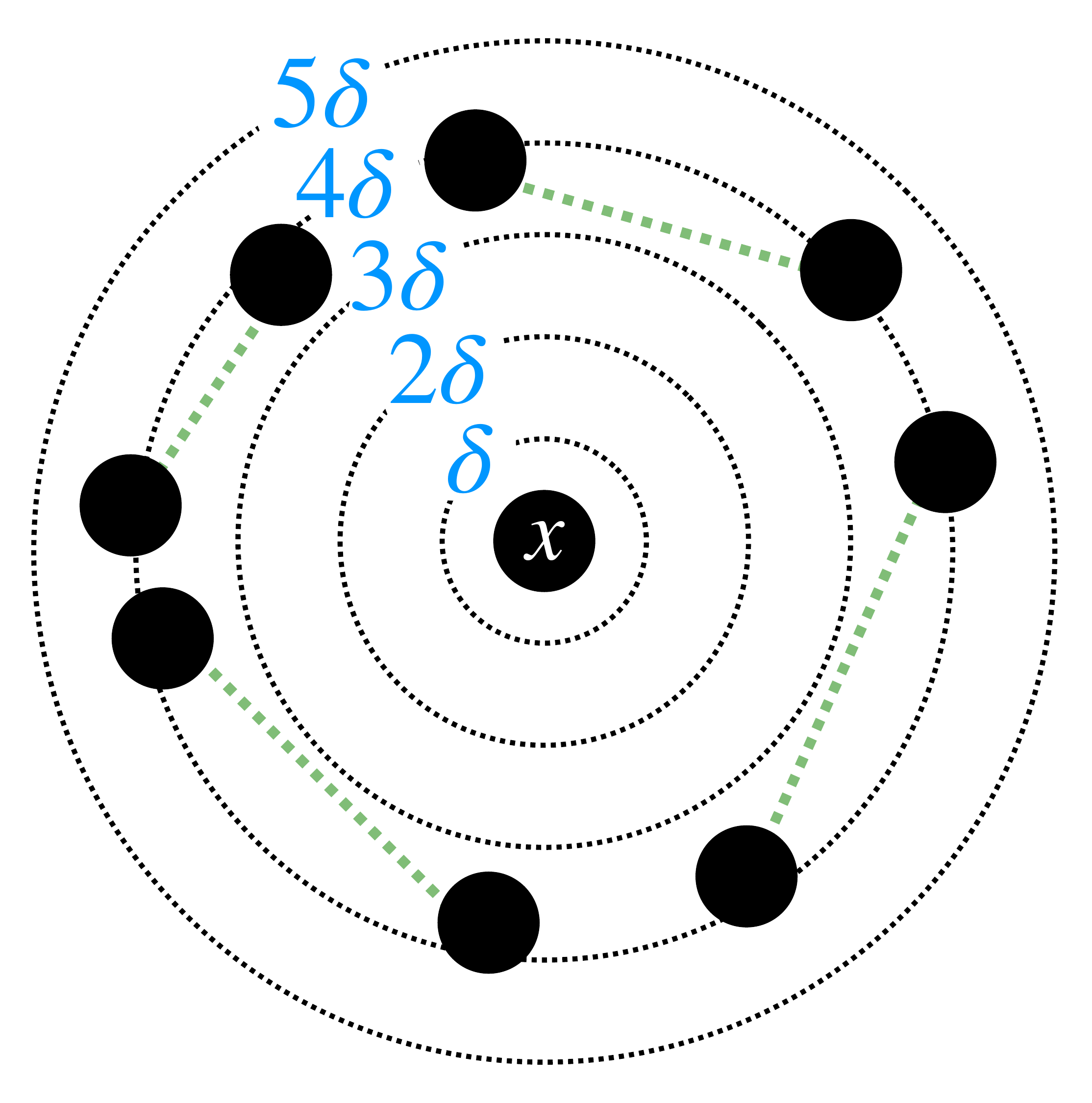}
        \caption{First pairs merge.}\label{sfig:circ6}
    \end{subfigure}   \hfill
    \begin{subfigure}[b]{0.24\textwidth}
        \centering
        \includegraphics[width=\textwidth,trim=0mm 0mm 300mm 0mm, clip]{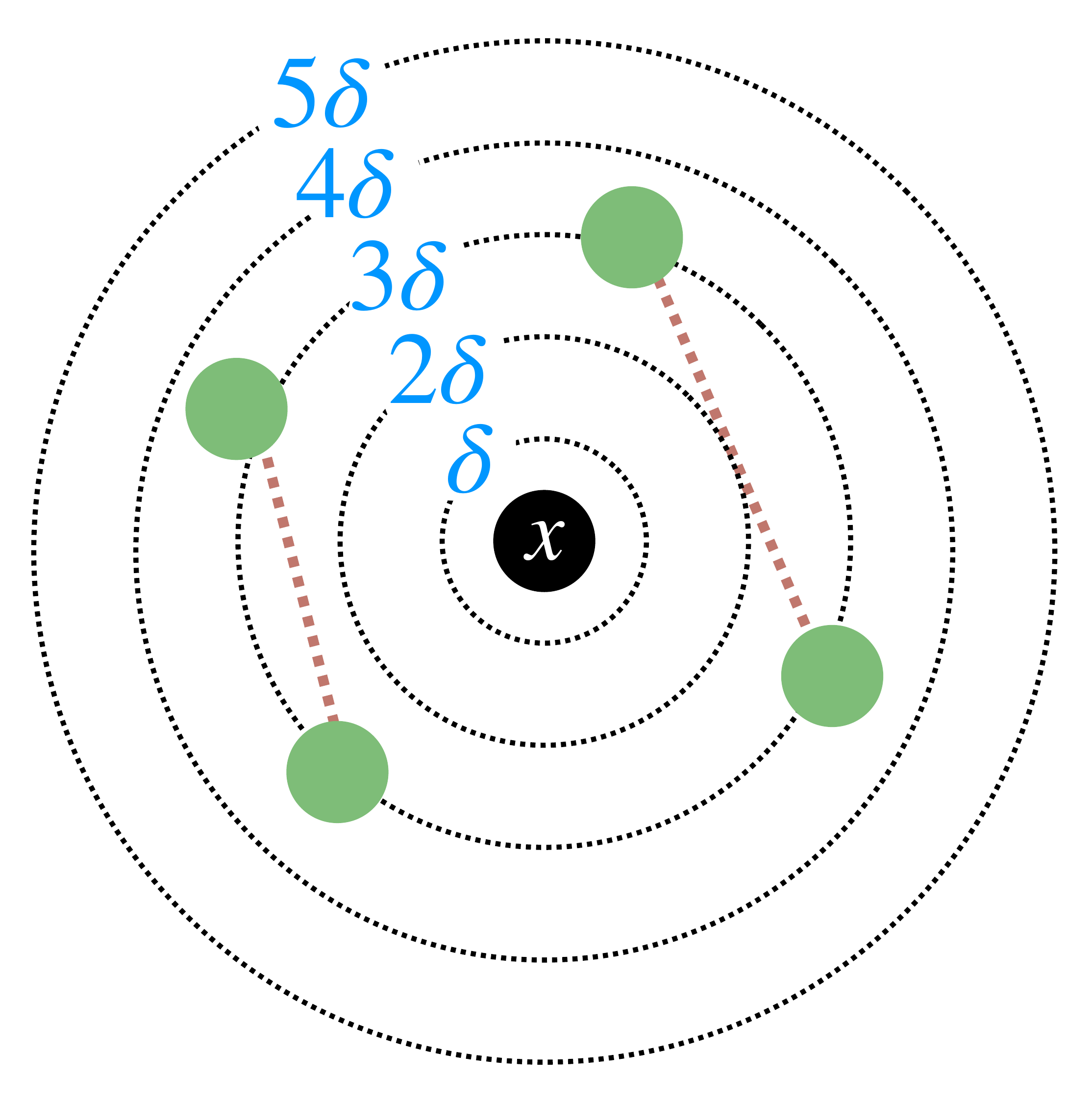}
        \caption{Second pairs merge.}\label{sfig:circ7}
    \end{subfigure}\hfill
    \begin{subfigure}[b]{0.24\textwidth}
        \centering
        \includegraphics[width=\textwidth,trim=0mm 0mm 300mm 0mm, clip]{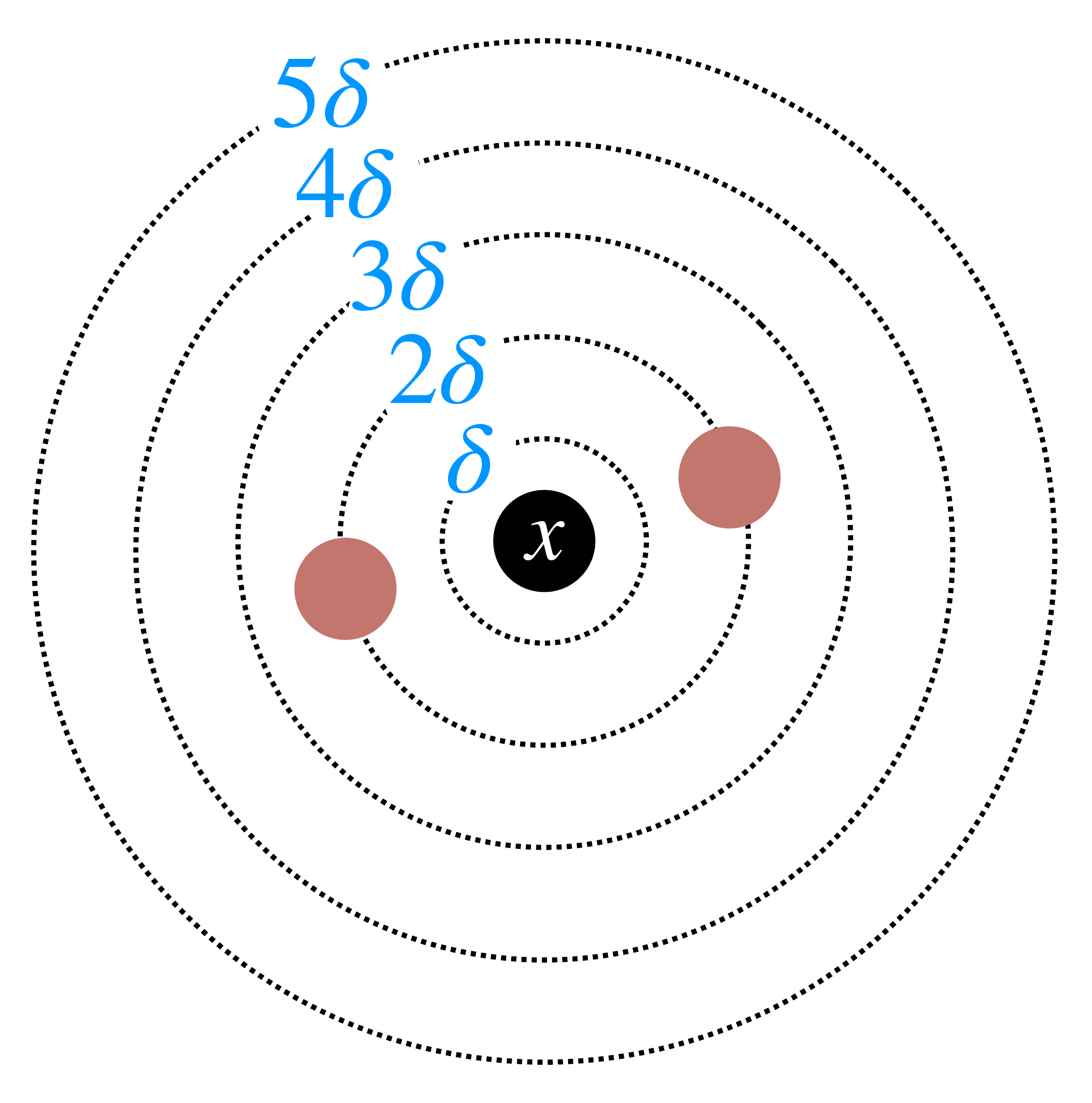}
        \caption{New centroids in ball.}\label{sfig:circ8}
    \end{subfigure}
    \caption{An illustration of our height bound argument for centroid. \Cref{sfig:circ1}/\ref{sfig:circ2}/\ref{sfig:circ3}/\ref{sfig:circ4}/\ref{sfig:circ5} show $x$ merging until there is nothing within its $2\delta$ ball. \Cref{sfig:circ6}/\ref{sfig:circ7}/\ref{sfig:circ8} shows how new centroids can enter its ball but only by merging off in pairs where $8$ centroids at distance $4\delta$ become only $2$ centroids at distance $2\delta$.} \label{fig:circ}
\end{figure}
 
The key idea of our height bound is to summarize this by assigning to each centroid $y$ a value, which is (up to scaling by $\delta$) negatively exponential in the centroid's distance from $x$
\begin{align*}
    \val(y) \approx \exp\left(-\frac{\|x-y\|}{\delta}\right),
\end{align*}
and which represents how much this centroid ``threatens'' our monotonicity at $x$; see \Cref{dfn:val} for a formal definition. We can, in turn, define our potential which is, roughly, the total threat to our monotonicity as
\begin{align*}
    \phi \approx \sum_{y}\val(y),
\end{align*}
where the above sum is over all our current centroids; see \Cref{dfn:potential} for a formal definition.

Lastly, we can show by our packing properties that $\phi$ is at most about $k^{O(k)}$ at the beginning of our phase and, by the convexity of $\exp$ and average-reducibility, is non-increasing. Furthermore, notice that any centroid within distance $2\delta$ of $x$ contributes $\Omega(1)$ to the potential and so when it merges with $x$ in our phase, it reduces our potential by $\Omega(1)$. As such, in a given phase in which $x$ merges with points at distance at most $2\delta$, the number of merges $x$ can participate in should be at most about $k^{O(k)}$.

Summarizing, we show our height bound by dividing HAC into (roughly) $O(\log n)$ phases where in the $j$th phase $x$ does not drift far from its starting point and performs merges of value at most $\approx 2^j$. We bound the number of merges in each phase as at most $k^{O(k)}$ by using the above potential function argument to bound the extent to which monotonicity can be violated.

Formalizing this argument and stating it in a way that works for Ward's linkage and, more generally, well-behaved linkage functions requires overcoming several challenges. The largest of these is the fact that, above, we are implicitly using the triangle inequality in several places. However, the triangle inequality  only holds for well-behaved functions under assumptions about cluster sizes and, even then, only multiplicatively approximately. This multiplicative constant prohibits us from naively applying the approximate triangle inequality as the multiplicative factor compounds each time we apply it. Likewise, we have to deal with approximation in HAC and the fact that, generally, distances according to well-behaved linkage functions cannot necessarily be summarized by distances between points in Euclidean space as we have done above.

The specifics of this argument and how we overcome these challenges are detailed in \Cref{sec:heightBound}.

\subsection{Parallel HAC for Low-Height Dendrograms}
Next, we leverage our bound on dendrogram height to design efficient parallel algorithms. In particular, the algorithm's depth is bounded in terms of the dendrogram height $h$ and an auxiliary parameter $\ell$, which we call the \emph{bounce-back length}. Intuitively, $\ell$ captures the non-monotonicity of the linkage function: specifically, it is the maximum number of merges a cluster undergoes---after an initial merge---before all linkage values to the cluster rise back to at least the value of the initial merge. While we defer a formal definition to \Cref{sec:algo}, we note that $\ell \le h$ always, and for certain linkage functions, such as Ward's, $\ell$ is constant, yielding stronger guarantees. 

We now state our main algorithmic result. 
\begin{restatable}{theorem}{conBound}\label{thm:conBound}
    Fix $\eps > 0 $ and suppose $d$ is well-behaved,
    and that any $(1+\eps)$-approximate HAC dendrogram has height at most $h$ and bounce-back length at most $\ell$. Then, there exists a $(1+\eps)$-approximate parallel HAC algorithm in $\mathbb{R}^k$ with $\tO(\wnn \cdot nh\ell^{O(k)})$ work and $\tilde{O}\left(h\ell^{O(k)}\right)$ depth, where $\wnn$ denotes the work of computing the $\tO(\ell^{O(k)})$ nearest neighbors of a cluster, assuming $\poly(n)$ aspect ratio.
\end{restatable}

Plugging our height bound (\Cref{lem:heightBound}), 
and utilizing state-of-the-art parallel NNS data structures in low dimensions (resulting in $\wnn=\tO(1)$), we obtain a $\tO(n)$ work and $\tO(1)$ depth algorithm for $(1+\eps)$-approximate centroid and Ward's HAC in $\R^k$. The centroid result holds for $k=O(1)$, while the Ward's result extends to $k = O(\log\log n/\log\log\log n)$, owing to stronger bounds on $\ell$.

\begin{restatable}{theorem}{parAlgCent}\label{thm:parAlgCent}
For $k = O\left(1\right)$ and $\epsilon>0$, $(1+\epsilon)$-approximate centroid HAC in $\mathbb{R}^k$ can be solved in $\tO(n)$ expected work and $\tilde{O}(1)$ depth with high probability, assuming $\poly(n)$ aspect ratio.
\end{restatable}

\begin{restatable}{theorem}{parAlgWard}\label{thm:parAlgWard}
For $k = O\left(\frac{\log \log n}{ \log \log \log n} \right)$ and $\epsilon>0$, $(1+\epsilon)$-approximate Ward's HAC in $\mathbb{R}^k$ can be solved in $\tO(n)$ expected work and $\tilde{O}(1)$ depth with high probability, assuming $\poly(n)$ aspect ratio.
\end{restatable}

\noindent We also note that the $(1+\eps)$-approximate centroid result directly implies 
$(1+\eps)^2$-approximate algorithm for \emph{squared centroid} HAC, an alternate but well-studied variant~\cite{Gower1967AComparison}.



\subsubsection{Intuition for Parallel HAC}
The algorithm proceeds in phases using standard geometric-thresholding (or bucketing): each phase begins with all linkage distances at least the lower threshold and performs merges with linkage value at most the upper threshold. Within a phase, the algorithm executes multiple rounds of synchronized parallel merges until all remaining linkage distances exceed the phase's upper threshold. However, because the linkage functions we consider are neither monotone nor reducible, the linkage distances can drop below the lower threshold during merging. 
To handle this, the algorithm starts with a set of clusters and, in parallel, repeatedly merges each cluster with its nearest neighbor until the linkage distances ``bounce back'' into the current threshold range. 
The parameter $\ell$ precisely bounds the number of such merges until a cluster bounces back, and since these merges must trace a path in the dendrogram, $\ell \le h$ always. 
This allows us to afford performing these merges sequentially within each phase without affecting the overall depth.

The main technical challenge is to show that these parallel merges do not interfere with one another, and that they can be sequentialized to prove the required approximation guarantee. 
If the linkage function satisfied the triangle inequality, then for any cluster $A$, its entire \emph{bounce-back path} (i.e., the set of clusters it merges with until bouncing back; see \Cref{dfn:BBPath}) would lie within an $O(\ell)$-radius ball centered at cluster $A$, up to a factor of the upper threshold of the current phase. Moreover, any other cluster whose bounce-back path might intersect that of $A$ would also lie within an $O(\ell)$-ball around $A$. Hence, as long as the algorithm selects clusters that are sufficiently far apart, each will proceed along its bounce-back path independently of others, regardless of whether merges are performed in parallel or sequentially. This allows us to sequentialize the execution: we can order the selected clusters arbitrarily and apply their sequence of merges one after the other. Additionally, by the packability property of well-behaved linkage functions, each selected cluster excludes at most $O(\ell^k)$ others from being processed in the same round. As a result, at least a $\Omega(1/\ell^k)$ factor of clusters make progress in each round, which suffices to bound the total number of rounds within a phase.

Although well-behaved linkage functions do not in general satisfy the triangle inequality, we show that analogous properties still hold for them.
The full algorithm and details are presented in \Cref{sec:algo}.

\subsection{Impossibility of Parallelism in High Dimensions}
Complementing our algorithmic results, we show that low dimensionality is necessary for parallelizing HAC for well-behaved linkage functions, even for approximate HAC. Specifically, as mentioned above, we show that it is CC-hard to $(1+\epsilon)$-approximate centroid-linkage HAC in linear dimensions for $\eps$ sufficiently small. CC-hardness is widely believed to rule out NC algorithms \cite{cook2014complexity, mayr1992complexity,subramanian1994new}.

For our purposes, we do not need to define the class CC but rather can simply use the definition of CC-hardness based on logspace reductions, as follows.
\begin{restatable}[CC-Hard]{definition}{CCHard}
    A problem is CC-hard if all problems of CC are logspace-reducible to it.
\end{restatable}

The decision version of HAC whose CC-hardness we show is as follows.
\begin{restatable}[$(1+\eps)$-Approximate Promise Decision HAC]{definition}{decisionHAC}
    We are given an instance of HAC consisting of $\mcP \subseteq \mathbb{R}^k$, linkage function $d$, $a,b,c \in \mcP$, and a guarantee that $a$, $b$, and $c$ will always merge together in the same relative order for every $(1+\eps)$-approximate HAC. Decide if $a$ and $b$ merge into the same cluster together before $c$ merges into the same cluster as $a$ or $b$.
\end{restatable}

Note that $(1+\eps)$-approximate promise decision HAC is easier than $(1+\eps)$-approximate HAC since we only have to solve it for $a$, $b$, and $c$ that are guaranteed to always merge together in the same order. Thus, running $(1+\eps)$-approximate HAC solves $(1+\eps)$-approximate promise decision HAC. The following summarizes our CC-hardness result for HAC.
\begin{restatable}{theorem}{ccHardness}\label{thm:ccHardness}
$(1+1/n^7)$-approximate promise decision HAC with $\dCen$ is CC-hard in $\mathbb{R}^n$.
\end{restatable}
Since centroid is well-behaved, the above result rules out NC algorithms for well-behaved linkage functions in linear dimensions in general (assuming $CC$-hardness rules out NC algorithms).

\subsubsection{Intuition for Hardness}

We show CC-hardness by reducing from the telephone communication problem (TCP). An instance of TCP consists of a capacity $\kappa$ and $n$ calls. Each call has a start and end time, and so can be represented by an interval. A call is accepted and serviced for its entire duration if there are less than $\kappa$ ongoing calls at its start time. Otherwise, it is dropped.

Our reduction consists of a central point $C$ and, for each call $i$, points $S_i$ and $R_i$. The $S_i$ will be mutually orthogonal and each successive one (in order of start time) will be slightly further from $C$ so that they merge in order. $R_i$ will be placed outside of $S_i$ so that when $S_i$ merges it will have to decide between merging with $C$ and $R_i$. Every time an $S_i$ merges with $C$, it drags $C$ slightly further away from all $S_j$ where $j > i$. We want to set up our HAC instance so that each $S_i$ merges with $C$ if and only if call $i$ is accepted in the TCP instance. HAC determines how many active phone calls there are by how far $S_i$ is from $C$. However, we need a way to adjust this distance when a phone call ends as we are not able to unmerge points from $C$. To accomplish this, we will also include a point $F_i$ for each call which will merge with $C$ when that call ends if and only if $S_i$ did not merge with $C$. Then, the number directions in which $C$ is off center is the number of calls that have finished plus the number of active calls. We are then able to set the distance between $C$ and each $S_i$ so that $S_i$ merges with $C$ if and only if call $i$ is accepted. Thus, if we are able to solve HAC in CC, we are also able to solve TCP.

We formally prove our hardness result in \Cref{sec:hardness}.

\section{Proving Linkage Functions are Well-Behaved}\label{sec:packablelinkage}
In this section we prove that the linkage functions that we study in this work are both well-behaved (\Cref{dfn:wellBehaved}).

Throughout this section we will make use of the following well-known fact regarding packing points in Euclidean space. For completeness, we give a proof in \Cref{sec:DefProofProps}.
\begin{restatable}[Packing Points in $\mathbb{R}^k$, Folklore]{theorem}{packPoints}\label{thm:packPoints}
    Let $\mcP \subseteq \mathbb{R}^k$ be a collection of points that satisfy $||u- v|| \geq r$ for every $u, v \in \mcP$ and there exists some $x \in \mathbb{R}^k$ such that $\mcP \subseteq B(x, R) =\{y : ||y-x|| \leq R\}$. Then $|\mcP| \leq \left(\frac{R}{r}\right)^{O(k)}$.
\end{restatable}

\subsection{Proofs for Centroid}
We start with our proofs for centroid.
\subsubsection{Centroid is Well-Behaved}

We now prove centroid is well-behaved. In particular, we prove the following.
\centroidWB*
\noindent In what follows, we show the necessary properties for being well-behaved.
\begin{lemma}\label{lem:centPack}
    Centroid linkage $\dCen$ is $1$-packable (\Cref{dfn:packableLinkage}).
\end{lemma}
\begin{proof}
    This is immediate from the definition of $\dCen$, \Cref{thm:packPoints} and the definition of $\alpha$-packability (\Cref{dfn:packableLinkage}).
\end{proof}

\begin{lemma}\label{lem:centApxMetric}
    Centroid linkage $\dCen$ satisfies the triangle inequality (\Cref{dfn:apxMetric}) (and therefore approximately satisfies the triangle inequality).
\end{lemma}
\begin{proof}
    This is immediate from the fact that $\dCen$ is given by the Euclidean distances between points in $\mathbb{R}^k$ which, in turn, satisfy the triangle inequality. In particular, for any $A, B, C \subseteq \mathbb{R}^k$, we have
    \begin{align*}
        \dCen(A, C) &= ||\mu(A) - \mu(C)|| \\
        &\leq ||\mu(A) - \mu(B)|| + ||\mu(B) - \mu(C)|| \\
        &= \dCen(A, B) + \dCen(B, C)
    \end{align*}
    where above we applied the triangle inequality for Euclidean space. Thus, it satisfies the triangle inequality and therefore also approximately satisfies the triangle inequality.
\end{proof}

\begin{lemma}\label{lem:centWS}
    Centroid linkage $\dCen$ is weight-stable (\Cref{dfn:weightStab}).
\end{lemma}
\begin{proof}
    Given any $A, B, C \subseteq \mathbb{R}^k$, we have
    \begin{align*}
        \dCen(A \cup B, A) & = ||\mu(A \cup B) - \mu(A)||\\
        &= \left|\left|\frac{|A|}{|A|+|B|}\mu(A) + \frac{|B|}{|A| + |B|}\mu(B) - \mu(A)\right|\right|\\
        &= \left|\left|\frac{|B|}{|A| + |B|}\mu(B) - \frac{|B|}{|A| + |B|}\mu(A)\right|\right|\\
        &= \frac{|B|}{|A| + |B|}||\mu(B) - \mu(A)||\\
        &= \frac{|B|}{|A| + |B|}\dCen(A, B)
    \end{align*}
    as required where in the second line we applied the definition of the centroid and in the second-to-last line we applied the homogeneity of the Euclidean norm $||\cdot||$.
\end{proof}

\begin{lemma}\label{lem:centAR}
    Centroid linkage $\dCen$ is average-reducible (\Cref{dfn:avgRed}).
\end{lemma}
\begin{proof}
    Consider any $A, B, C \subseteq \mathbb{R}^k$. By the triangle inequality for $\dCen$ (\Cref{lem:centWS}) we have
    \begin{align*}
        \dCen(C, A) \leq \dCen(C, A \cup B) + \dCen(A \cup B, A).
    \end{align*}
    Symmetrically, for $B$ we have
    \begin{align*}
        \dCen(C, B) \leq \dCen(C, A \cup B) + \dCen(A \cup B, B).
    \end{align*}
    Summing these two inequalities we get 
    \begin{align}\label{eq:gasfsa}
        \dCen(C, A) &+ \dCen(C, B) \nonumber \\&\leq 2 \cdot \dCen(C, A \cup B) + \dCen(A \cup B, A) + \dCen(A \cup B, B).
    \end{align}

    \noindent Next, observe that, by definition of $\dCen$, we have
    \begin{align*}
        \dCen(A,A \cup B) + \dCen(A \cup B,B) = \dCen(A, B)
    \end{align*}
    and so combining this with \Cref{eq:gasfsa} we have
    \begin{align}
        \dCen(C, A) + \dCen(C, B) \leq 2 \cdot \dCen(C, A \cup B) + \dCen(A, B).
    \end{align}
    Solving for $\dCen(C, A \cup B)$ and using the symmetry of $\dCen$, we get
    \begin{align*}
        d(A \cup B, C) \geq \frac{\dCen(A,C) + \dCen(B,C)}{2} - \dCen(A,B)
    \end{align*}
    as required.
\end{proof}

\begin{lemma}\label{lem:cenPolyDiam}
    Centroid linkage $\dCen$ has poly-bounded diameter (\Cref{dfn:polyDiam}).
\end{lemma}
\begin{proof}
    Let $\Delta = \max_{u,v \in \mcP} \dCen(u,v) = ||u-v||$. Fix an arbitrary $x_0 \in \mcP$ and let $B = B\left(x_0, \Delta\right)$ be all points in $\mathbb{R}^k$ within $\Delta$ of $x$. We will show that for any $A \subseteq \mathbb{R}^k$, we have $\mu(A) \in B$. This, along with the triangle inequality, proves the lemma. 
    
    Consider the function $f(x) = ||x_0-x||$. Observe that by definition of $\Delta$, we have $p \in B$ for every $p$ which is to say $f(p) \leq \Delta$. Thus, by the convexity of $f$ and Jensen's inequality we have
    \begin{align*}
        f(\mu(A)) \leq \frac{1}{|A|}\sum_{a \in A}f(a) 
        &\leq \frac{1}{|A|}\sum_{a \in A}\Delta
        = \Delta
    \end{align*}
    as required.
\end{proof}
Combining the above lemmas proves that centroid is well-behaved (\Cref{thm:centroidWB}).


\subsection{Proofs for Ward's}
We now move on to proving that Ward's linkage function is 
well-behaved.

\subsubsection{Ward's Preliminaries}\label{sec:wardPrelim}

Throughout this section we will make use of alternate forms of Ward's. The first of these is an approximation of Ward's linkage from \cite{grosswendt2019analysis} which says that, roughly, Ward's is the squared centroid distance times the smaller cluster size.
\begin{restatable}[Ward's Approximation, \cite{grosswendt2019analysis}]{lemma}{wardAPX}\label{lem:simpropsmallcluster}
Given $A, B \subseteq \mathbb{R}^k$ we have
\begin{align*}
    \frac{1}{2}\min\{|A|,|B|\} \cdot \|\mu(A)-\mu(B)\|^2 \le \dWard(A,B) \le \min\{|A|,|B|\} \cdot \|\mu(A)-\mu(B)\|^2.
\end{align*}
\end{restatable}

\noindent Additionally, we will make use of the Lance-Williams~\cite{lance1967general} update form of Ward's.
\begin{restatable}[Lance-Williams Form \cite{lance1967general}]{lemma}{LanceWilliamsWards}\label{lem:LanceWilliamsWards}
    Given $A,B,C \subseteq \mathbb{R}^k$, we have $\dWard(A\cup B,C)$ is
    \begin{align*}
        \frac{|A|+|C|}{|A|+|B|+|C|}\dWard(A,C)+\frac{|B|+|C|}{|A|+|B|+|C|}\dWard(B,C)-\frac{|C|}{|A|+|B|+|C|}\dWard(A,B).
    \end{align*}
\end{restatable}

We will also use the following folklore bound which says that, up to a factor of $2$, squared Euclidean distances satisfy the triangle inequality.

\begin{restatable}[Approximate Triangle Inequality for Squared Euclidean Distances]{lemma}{squaredEucTri}\label{lem:squaredEucTri}
    Given any points $a,b,c \in \mathbb{R}^k$, we have
    \begin{align*}
        \|a-c\|^2 \leq 2 \cdot \left(\|a-b\|^2 + \|b-c\|^2 \right).
    \end{align*}
\end{restatable}
\noindent For the sake of completeness, we give proofs of all of the above facts in \Cref{sec:DefProofProps}.

\subsubsection{Ward's is Well-Behaved}

In this section we prove Ward's linkage function is well-behaved.
\wardsWB*
\noindent In what follows, we show the necessary properties for being well-behaved.

We begin by proving that Ward's linkage is $O(\log n)$-packable. At a high-level, given a set of clusters, that are Ward's distance at least $r$ and at most $R$ apart, the idea is to partition them into at most $O(\log n)$ subsets, each containing clusters of sizes within a constant factor. By \Cref{lem:simpropsmallcluster}, the Ward's linkage between a cluster in this partition and any other cluster of the same partition, or a cluster of a larger size, is essentially proportional to the (squared) distance between their centroids. Thus, we apply the packing properties for points in Euclidean space (\Cref{thm:packPoints}) to bound the number of clusters within each part.
 
\begin{restatable}{lemma}{wardsPackable} \label{thm:wardsPack}
Ward's linkage $\dWard$ is $O(\log n)$-packable (\Cref{dfn:packableLinkage}).
\end{restatable}
\begin{proof}
    Consider our ball of clusters $B_{\dWard}^\mcC(A,R)$ for which for any distinct $X,Y \in B_{\dWard}^\mcC(A,R)$, we know $\dWard(X,Y) \geq r$. 

    Let $\mcC_i$ denote the set of clusters in $B_{\dWard}^\mcC(A,R)$ with sizes in the range $[2^{i-1},2^i]$. We will show that $|\mcC_i| = O((R/r)^{O(k)})$ for all $i$. Since $i \le \log n$, the theorem follows.

    Consider any cluster $X\in \mcC_i$ and a cluster $Y$ such that $|Y|\ge |X|$. By \Cref{lem:simpropsmallcluster}, we have:
    \begin{align*}
        \dWard(X,Y) \le |X|\cdot \|\mu(X) - \mu(Y)\|^2 &\implies \|\mu(X)-\mu(Y)\|^2 \ge r/2^i\\
        &\implies \|\mu(X) - \mu(Y)\| \ge \sqrt{r/2^i},
    \end{align*}
     and,
    \begin{align*}
        \dWard(X,Y) \ge |X| \cdot \|\mu(X)-\mu(Y)\|^2/2 &\implies \|\mu(X)-\mu(Y)\|^2 \le 2R/2^{i-1}\\
        &\implies \|\mu(X)-\mu(Y)\| \le \sqrt{R/2^{i-2}}.
    \end{align*}
     Thus, by the packing property for points in Euclidean space (\Cref{thm:packPoints}), we have $|\mcC_i| = O((R/r)^{O(k)})$.
\end{proof}

\begin{restatable}{lemma}{wardsApxTri}
Ward's linkage $\dWard$ approximately satisfies the triangle inequality (\Cref{dfn:apxMetric}).
\end{restatable}
\begin{proof}
    Consider $A, B, C \subseteq \mathbb{R}^k$ where $|B| \geq \min(|A|, |C|)$. By the squared Euclidean triangle inequality (\Cref{lem:squaredEucTri}), we have
    \begin{align*}
        \|\mu(A) -\mu(C) \|^2 \leq 2 \cdot \left(\|\mu(A) -\mu(B) \|^2 + \|\mu(B) -\mu(C) \|^2 \right).
    \end{align*}
    Likewise, by our approximation for Ward's (\Cref{lem:simpropsmallcluster}) we have
    \begin{align*}
        \dWard(A,C) \le \min\{|A|,|C|\} \cdot \|\mu(A)-\mu(C)\|^2.
    \end{align*}
    and so by $\min(|A|,|C|) \leq \min(|B|,|C|)$ and $\min(|A|,|C|) \leq \min(|A|,|B|)$ and another application of our Ward's approximation (\Cref{lem:simpropsmallcluster}) we have
    \begin{align*}
        \dWard(A,C) &\leq 2\min\{|A|,|C|\} \cdot \|\mu(A) -\mu(B) \|^2 + 2\min\{|A|,|C|\} \cdot \|\mu(B) -\mu(C) \|^2\\
        &\leq 2\min\{|A|,|B|\} \cdot \|\mu(A) -\mu(B) \|^2 + 2\min\{|B|,|C|\} \cdot \|\mu(B) -\mu(C) \|^2\\
        &\leq 4 \cdot(\dWard(A,B) + \dWard(B,C))
    \end{align*}
    as required.
\end{proof}

\begin{lemma}\label{thm:wardsWeightStable}
Ward's linkage $\dWard$ is weight-stable (\Cref{dfn:weightStab}).
\end{lemma}
\begin{proof}
Consider $A, B \subseteq \mathbb{R}^k$. By the Lance-Williams form (\Cref{lem:LanceWilliamsWards}), the symmetry of $\dWard$ and $\dWard(A,A)=0$ we have 
\begin{align*}
        \dWard(A \cup B, A) &= \frac{2|A|}{2|A|+|B|}\dWard(A,A)+\frac{|B|+|A|}{2|A|+|B|}\dWard(B,A)-\frac{|A|}{2|A|+|B|}\dWard(A,B) \\
        &= \frac{|B|}{2|A|+|B|}\dWard(A,B)\\
        &\leq \frac{|B|}{|A|+|B|}\dWard(A,B)
    \end{align*}
    as required.
\end{proof}

\begin{lemma}\label{thm:wardsAverageReducible}
Ward's linkage $\dWard$ is average-reducible (\Cref{dfn:avgRed}).
\end{lemma}
\begin{proof}
    Suppose we have three clusters $A,B,C \subseteq \mathbb{R}^k$ such that $|C|\ge |A|+|B|$. 
    Then, by the Lance-Williams form (\Cref{lem:LanceWilliamsWards}), we have $\dWard(A\cup B,C)$ is 
    \begin{align*}
        \frac{|A|+|C|}{|A|+|B|+|C|}&\dWard(A,C)+\frac{|B|+|C|}{|A|+|B|+|C|}\dWard(B,C)-\frac{|C|}{|A|+|B|+|C|}\dWard(A,B)\\
        &\ge \frac{|A|+|C|}{2|C|}\dWard(A,C)+\frac{|B|+|C|}{2|C|}\dWard(B,C)-\dWard(A,B)\\
        &\ge \frac{\dWard(A,C) + \dWard(B,C)}{2}-\dWard(A,B)
    \end{align*}
    as required.
\end{proof}

\begin{lemma}\label{lem:wardsPolyDiam}
Ward's linkage $\dWard$ has poly-bounded diameter (\Cref{dfn:polyDiam}).
\end{lemma}
\begin{proof}
Let $\Delta = \max_{u,v \in \mcP} \dWard(u,v)$ be the maximum Ward's distance between a pair of initial points. By \Cref{lem:simpropsmallcluster}, we have that
\begin{align*}
    \Delta \leq \max_{u, v \in \mcP}\|u-v\|^2.
\end{align*}

Next, consider $A, B \subseteq \mcP$. Our goal is to show $\dWard(A,B) \leq \poly(n \cdot \Delta)$. However, by \Cref{lem:simpropsmallcluster}, we know that $\dWard(A,B) \leq \min\{|A|,|B|\} \cdot \| \mu(A) -\mu(B)\|^2 \leq n \cdot \| \mu(A) -\mu(B)\|^2$. Thus, it suffices to argue that $\| \mu(A) -\mu(B)\|^2 \leq \poly\left(n \cdot \max_{u, v \in \mcP}\|u-v\|^2\right)$. However, this is immediate from the fact that $\dCen$ is poly-bounded (\Cref{lem:cenPolyDiam}).
\end{proof}

Combining the above lemmas shows that Ward's is well-behaved (\Cref{thm:wardWB}).



\section{Height Bounds for HAC with Well-Behaved Linkage Functions}\label{sec:heightBound}
We now prove that any well-behaved linkage function gives rise to a low height dendrogram in low dimensions, as summarized below.

\heightBound*
For the rest of this section, we fix a well-behaved linkage function $d$, where $c_\Delta$ is the constant according to which $d$ approximately satisfies the triangle inequality (\Cref{dfn:apxMetric}) and $\alpha = \tilde{O}(1)$ is the parameter according to which $d$ is $\alpha$-packable (\Cref{dfn:packableLinkage}). Likewise, we fix a $c$-approximate HAC for $d$. We let $\mcP$ be the initial point set to which we are applying HAC. Let $A_i$ and $B_i$ be the $i$th pair of clusters merged by HAC and let $\delta_i := d(A_i, B_i)$ be the value of this merge. Let $\mcC_i$ be all clusters just before the $i$th merge.

Fix an arbitrary point $x_0 \in \mcP$. Let $X_i$ be the cluster containing $x_0$ just before the $i$th merge is performed. If $X_i$ participates in the $i$th merge (that is, $X_i \in \{A_i, B_i\}$), then we let $\bar{X}_i$ be the cluster with which $X_i$ merges (that is, $\bar{X}_i$ is the one element of $\{A_i, B_i\} \setminus \{X_i\}$). To prove the above theorem, it suffices to argue that the cluster containing $x_0$ participates in at most $\tilde{O}\left((c \cdot k)^{O(k)}\right)$-many merges and so we proceed to do so for the rest of this section.

\subsection{Dividing HAC into Phases}

In order to argue that the cluster containing $x_0$ participates in boundedly-many merges, we will divide the merges that HAC performs into $\tilde{O}\left((c \cdot k)^{O(k)}\right)$ phases where the number of merges that the cluster containing $x_0$ participates in is at most $\tilde{O}\left((c \cdot k)^{O(k)}\right)$ in each one of these phases. The $j$th phase is defined as a contiguous sequence of merges in which the maximum merge value has not increased too much, the weight of $X_i$ has not significantly increased and every merge we perform occurs close to where $X_i$ ``started'' in this phase.

More formally, the $j$th phase consist of the merges indexed by $I_j := [s_j, f_j]$ where $s_j = f_{j-1}+1$. We let $\tilde{\delta}_j := \max_{i \leq s_j} \delta_i$ be the largest value of a merge performed up to the beginning of the $j$th merge, let $i_j := \argmax_{i \leq s_j} \delta_i$ be the corresponding index and let $\tilde{X}_j := X_{s_j}$ be the cluster containing $x_0$ at the beginning of this phase. Then, $I_j$ is defined recursively as follows.

\begin{definition} [$j$th phase, $I_j$]\label{dfn:phase}
The indices of merges in the $j$th phase are $I_j = [s_j, f_j]$ where $s_j = f_{j-1} +1$ and $f_{j}$ is the maximum integer greater than or equal to $s_{j}$ where:
\begin{enumerate}
    \item \textbf{Small Merges:} $\delta_i \leq 2 \cdot \tilde{\delta}_j$ for every  $i \in I_j$;
    \item \textbf{Small Size Increase:} $|X_{i+1}| < \frac{3}{2} \cdot |\tilde{X}_j|$ for every $i \in I_j$; and;
    \item \textbf{Small Drift:} $d(\tilde{X}_j, \bar{X}_i) \leq 3 c_{\Delta} \cdot \tilde{\delta}_j$ for every $i \in I_j$ in which $X_i$ participates in a merge.

\end{enumerate}
and as a base case we have $f_0 = 0$.
\end{definition}
Observe that the $j$th phase for $j > 0$ starts because one of the above conditions is not met.


\subsection{Bounding the Number of Merges in Each Phase with a Potential Function}
We proceed to argue in each phase the cluster containing $x_0$ participates in a small number of merges.

We argue this by way of a potential function that, up to scaling, approximately captures the minimum distance of a node to the initial position of the cluster containing $x_0$ in this phase. In particular, for the rest of this section we fix a phase $j$ and, as before, let $\tilde{X}_j$ be the cluster containing $x_0$ at the beginning of the $j$th phase and let $\tilde{\delta}_j$ be the maximum merge distance up to the $j$th phase. For the rest of the section we let $\exp_2(x) := 2^x$.

\begin{definition}[Cluster Value $\val_j$]\label{dfn:val} We define the value of cluster $A$ in the $j$th phase as:
\begin{align*}
    \val_j(A) := \exp_2 \left(- \frac{d(\tilde{X}_j,A)}{4 \cdot \tilde{\delta}_j} \right).
\end{align*}
\end{definition}
\noindent In other words, up to scaling by ($\Theta$ of) the maximum merge done so far, the value of $A$ is negatively exponential in its distance from $\tilde{X}_j$ according to the linkage function $d$.

We now argue that a small distance merge results in a new cluster whose value is smaller than its constituent clusters by a multiplicative constant (under some assumptions about cluster sizes).
\begin{lemma}\label{lem:centroidConvex}
Given any $A, B \subseteq \mathbb{R}^k$ where $d(A, B) \leq 2 \tilde{\delta}_j$ and $|A| + |B| \leq |\tilde{X}_j|$, we have 
\begin{align*}
\val_j(A \cup B) \leq \frac{1}{\sqrt{2}} \left(\val_j(A) + \val_j(B) \right).
\end{align*}
\end{lemma}
\begin{proof}
By assumption $d$ is well-behaved and, in particular, is average-reducible (\Cref{dfn:avgRed}). Since, by assumption $|\tilde{X}_j| \geq |A| +|B|$, we may apply average-reducibility to get
\begin{align*}
    d(A \cup B, \tilde{X}_j) \geq \frac{d(A, \tilde{X}_j)+d(B, \tilde{X}_j)}{2} - d(A, B)
\end{align*}
and since $d(A, B) \leq 2 \tilde{\delta}_j$ by assumption, it follows that 
\begin{align}\label{eq:fsssa}
    d(A \cup B, \tilde{X}_j) \geq \frac{d(A, \tilde{X}_j)+d(B, \tilde{X}_j)}{2} - 2\tilde{\delta}_j.
\end{align}
Thus, we have 
\begin{align*}
    \val_j(A \cup B) &=\ \exp_2\left(\frac{- d(\tilde{X}_j, A \cup B)}{4 \tilde{\delta}_j} \right) \\
    &\leq \exp_2\left(\frac{- d(\tilde{X}_j, A) -  d(\tilde{X}_j, B)}{8\tilde{\delta}_j} + 1/2\right)\\
    &= \sqrt{2} \cdot \exp_2\left(\frac{- d(\tilde{X}_j, A) -  d(\tilde{X}_j, B)}{8\tilde{\delta}_j} \right)\\
    & \leq \frac{1}{\sqrt{2}} \left(\exp_2\left(\frac{- d(\tilde{X}_j, A)}{\tilde{4\delta}_j} \right) + \exp_2\left(\frac{-d(\tilde{X}_j, B)}{4\tilde{\delta}_j} \right) \right)\\
    & = \frac{1}{\sqrt{2}} \left(\val_j(A) + \val_j(B) \right)
\end{align*}
where the second line comes from \Cref{eq:fsssa}, the fourth line comes from the convexity of $\exp_2(-x)$ and the last line comes from the definition of $\val_j$ as given in \Cref{dfn:val}.
\end{proof}

We now define our potential function, which, up to ignoring large clusters, is just the sum of all clusters' values.
\begin{definition}[Potential $\phi_j$]\label{dfn:potential}
We define the potential of clusters $\mcC$ in the $j$th phase as the sum of the value of all clusters whose size is at most $|\tilde{X}_j|/2$:
\begin{align*}
    \phi_j(\mcC) := \sum_{A \in  \mcC : |A| < |\tilde{X}_j|/2} \val_j(A)
\end{align*}
\end{definition}

As a consequence of \Cref{lem:centroidConvex}, we get that our potential function is non-increasing as we merge our clusters. For the below, recall that $\mcC_i$ is all clusters just before the $i$th merge.
\begin{lemma}\label{lem:phiMono}
$\phi_j(\mcC_i)$ is non-increasing in $i$ in the $j$th phase. That is, $\phi_j(\mcC_{i}) \geq \phi_j(\mcC_{i+1})$ for $i  \in I_j$.
\end{lemma}
\begin{proof}
Consider the $i$th merge for $i \in I_j$ and suppose it merges clusters $A_i$ and $B_i$. It suffices to show that $\phi_j(\mcC_{i+1}) - \phi_j(\mcC_{i}) \leq 0$ and so we do so for the rest of the proof.

If $|A_i \cup B_i| \geq |\tilde{X}_j|/2$ then the clusters which contribute their value to $\phi_j(\mcC_{i})$ must be a superset of those which contribute their value to $\phi_j(\mcC_{i+1})$ and so $\phi_j(\mcC_{i+1}) - \phi_j(\mcC_{i}) \leq  0$.

On the other hand, if $|A_i \cup B_i| < |\tilde{X}_j|/2$ then we must have $|A_i| < |\tilde{X}_j|/2$ and $|B_i| < |\tilde{X}_j|/2$ and so by our definition of our potential $\phi_j$, we have 
\begin{align}\label{eq:sfasf}
    \phi_j(\mcC_{i+1}) - \phi_j(\mcC_{i}) =  \val_j(A_i \cup B_i) - \val_j(A_i) - \val_j(B_i).
\end{align}

Furthermore, it follows that $|A_i| + |B_i| \leq |\tilde{X}_j|$ and since this a merge in the $j$th phase, by assumption we know $d(A_i, B_i) \leq 2 \tilde{\delta}_j$. Thus, we may apply \Cref{lem:centroidConvex} to see that 
\begin{align}\label{eq:asgfasg}
    \val_j(A_i \cup B_i) \leq \frac{1}{\sqrt{2}}\left(\val_j(A_i) + \val_j(B_i)\right).
\end{align}
Combining \Cref{eq:sfasf} and \Cref{eq:asgfasg} with the non-negativity of $\val_j$ gives
\begin{align*}
    \phi_j(\mcC_{i+1}) - \phi_j(\mcC_{i}) \leq  \left(\frac{1}{\sqrt{2}}-1 \right) \left( \val_j(A_i) + \val_j(B_i)\right) \leq 0
\end{align*}
as required.
\end{proof}

We next observe that our potential does not start too large at the beginning of a phase.
\begin{lemma}\label{lem:phiInit}
$\phi_j$ is at most $\tilde{O}  \left(( k \cdot c)^{O(k)} \right)$ at the start of the $j$th phase. That is, $\phi_j(\mcC_{s_j}) \leq \tilde{O}  \left(( k \cdot c)^{O(k)} \right)$.
\end{lemma}
\begin{proof}
Recall that $i_j$ is the index of the largest merge performed up to iteration $i$. By definition of $i_j$ and the fact that we are performing a $c$-approximate HAC, we have that every pair of clusters just before the $i_j$th iteration is at least $\tilde{\delta}_{j}/c$ apart according to $d$. That is, for every $A,B \in \mcC_{i_j}$ we have $d(A,B) \geq \tilde{\delta}_{j}/c$.

Thus, applying the fact that $d$ is well-behaved and, in particular, $\alpha$-packable (\Cref{dfn:packableLinkage}) for $\alpha = \tilde{O}(1)$, we have that for any $x > 0$ that
\begin{align}\label{eq:pack}
    \left|B_d^{\mcC_{i_j}}(\tilde{X}_j, x \cdot \tilde{\delta}_j)\right| \leq \alpha \cdot \left(cx \right)^{O(k)} = \alpha \cdot  \exp_2\left(O(k\log c) + O(k \log x)\right).
\end{align}
We can upper bound the value of $\phi_j(\mcC_{i_j})$ by way of such a series of balls. For $l \geq 0$ let 
\begin{align*}
B_l := B_d^{\mcC_{i_j}}(\tilde{X}_j, \tilde{\delta}_j \cdot 2^l)    
\end{align*}
be the radius $\tilde{\delta}_j \cdot 2^l$ ball of clusters centered at $\tilde{X}_j$. Similarly, we let $S_0 = B_0 $ and for $l \geq 1$ we let the $l$th ``shell'' be
\begin{align*}
S_l := B_{l} \setminus B_{l-1}.  
\end{align*}
Observe that since $S_l \subseteq B_l$ by \Cref{eq:pack}, we have that 
\begin{align*}
|S_l|\leq \alpha \cdot \exp_2\left(O(k\log c) + O(kl)\right).    
\end{align*}
On the other hand, we have that for $l \geq 1$, each cluster in $S_l$ is at least $2^{l-1} \cdot \tilde{\delta}_j$ from $\tilde{X}_j$ according to $d$ and so contributes at most $\exp_2(-2^{l-3})$ to $\phi_j({\mcC}_{i_j})$. It follows that 
\begin{align}\label{eq:asfasf}
\phi_j(\mcC_{i_j}) & \leq \alpha \cdot \sum_{l} \exp_2(-2^{l-3}) \cdot \exp_2\left(O(k\log c) + O(kl)\right) \nonumber \\
    &= \alpha \cdot  \sum_{l}  \exp_2\left(-2^{l-3} + O(k\log c) + O(kl)\right)
    \end{align}
To bound this sum, we let $\beta = \Theta( \log \log c + \log k)$ for an appropriately large hidden constant. We then have for $l \geq \beta$ that
\begin{align*}
\exp_2\left(-2^{l-3} + O(k\log c) + O(kl)\right) \leq \exp_2(-l)
\end{align*}
and so using this, $\alpha = \tilde{O}(1)$, and \Cref{eq:asfasf}, we then get
\begin{align*}
    \phi_j(\mcC_{i_j}) &\leq \alpha \cdot \sum_{l \leq \beta} \exp_2\left(-2^{l-3}  + O(k\log c) + O(kl)\right) + \alpha \cdot \sum_{l > \beta}  \exp_2\left(-2^{l-3}  + O(k\log c) + O(kl)\right) \\
    &\leq \alpha \cdot \beta \cdot  \exp_2\left(O(k\log c) + O(k\beta)\right) + \alpha \cdot \sum_{l > \beta} \exp_2(-l)\\
    & = \alpha \cdot \beta  \cdot c^{O(k)} \cdot 2^{O(\beta k)}\\
    & =  \tilde{O}  \left(( k \cdot  c)^{O(k)} \right).
\end{align*}
Lastly, $i_j \leq s_j$ and so by \Cref{lem:phiMono} we have $\phi_j(\mcC_{s_j}) \leq \phi_j(\mcC_{i_j})$, giving our lemma.
\end{proof}

Concluding, we bound the number of merges of the cluster containing $x_0$ in each phase.
\begin{lemma}\label{lem:mergesPerPhase}
$x_0$'s cluster participates in at most $\tilde{O}  \left(( k  \cdot c)^{O(k)} \right)$ merges in each phase.
\end{lemma}



\begin{proof}
    Fix a $j$ and consider the $j$th phase. We will show that each time the cluster containing $x_0$ is merged in phase $j$, the potential $\phi_j$ decreases by at least a constant.
        
    Since $\phi_j$ is at most $\tilde{O}  \left(( k \cdot c)^{O(k)} \right)$ just before the $s_j$th merge phase  by \Cref{lem:phiInit}, $\phi_j$ is always $\geq 0$, and $\phi_j$ is non-increasing by \Cref{lem:phiMono}, this can only happen at most $\tilde{O}  \left(( k \cdot c)^{O(k)} \right)$ times which suffices to show the lemma.

    Fix $i \in I_j$ in which the cluster $X_i$ containing $x_0$ participates in a merge. Recall that we notate the cluster with which $X_i$ merges by $\bar{X}_i$. Our goal is to show that $\phi_j(\mcC_{i}) - \phi_j(\mcC_{i+1}) = \Omega(1)$. By definition of a phase, we know that 
    \begin{align*}
        |\bar{X}_i| < \frac{1}{2} \cdot |\tilde{X}_j|.
    \end{align*}
    since otherwise, by the fact that $\tilde{X}_j \subseteq X_i$, we would have $|X_{i+1}| = |X_i| + |\bar{X}_i| \geq |\tilde{X}_j| +  \frac{1}{2} |\tilde{X}_j| \geq \frac{3}{2}|\tilde{X}_j|$ which would contradict the fact that the $i$th merge occurs in the $j$th phase as per \Cref{dfn:phase}.
    
    Likewise, again since $\tilde{X}_j \subseteq X_i$, we know that that 
    \begin{align*}
        |X_i| \geq |\tilde{X}_j| > |\tilde{X}_j|/2.
    \end{align*}
    It follows by the definition of our potential $\phi_j$ (\Cref{dfn:potential}) that before the $i$th merge $\bar{X}_i$ contributes to our potential but $X_i$ does not. Similarly, after the $i$th merge $X_{i+1} = X_i \cup \bar{X}_i$ does not contribute to $\phi_j$ (nor do $X_i$ and $\bar{X}_i$ since they are no longer clusters in our set of clusters). Putting this together, we have    \begin{align}\label{eq:saga}
        \phi_j(\mcC_{i}) - \phi_j(\mcC_{i+1}) = \val_j(\bar{X}_i).
    \end{align}

    On the other hand, since $i \in I_j$, we know by our definition of the $j$th phase that $d(\tilde{X}_j, \bar{X}_i) \leq 3c_{\Delta} \cdot \tilde{\delta}_j$ and so by definition of $\val_j$ (\Cref{dfn:val}) and the fact that $c_{\Delta}$ is a constant, we get
    \begin{align}\label{eq:daghu}
        \val_j(\bar{X}_i) &= \exp_2 \left(-\frac{d(\tilde{X}_j, \bar{X}_i)}{4 \cdot \tilde{\delta}_j} \right) \geq \exp_2\left(-\frac{3}{4}c_{\Delta} \right)= \Omega(1)
    \end{align}
    Combining \Cref{eq:saga} and \Cref{eq:daghu}, we get $\phi_j(\mcC_{i}) - \phi_j(\mcC_{i+1}) \geq \Omega(1)$ as required.
    \end{proof}

\subsection{Bounding the Number of Phases}
We now bound the total number of phases. Recall that, by definition of a phase (\Cref{dfn:phase}), a phase starts because one of our conditions fails to be met.  In particular, we say that phase $j$ starts because of a large merge if $\delta_{s_j} > 2 \cdot \tilde{\delta}_{j-1}$, because of a large size increase if $|X_{s_j}| \geq \frac{3}{2} \cdot |\tilde{X}_{j-1}|$ and because of a large drift if $d(\tilde{X}_{j-1}, \bar{X}_{s_j}) > 3 c_{\Delta} \cdot \tilde{\delta}_{j-1}$. We proceed to bound the number of phases that fail each of our conditions.

The number of phases that start because of a large merge is at most $O(\log n)$ since each such phase increases the largest merge we've performed by a multiplicative constant and the largest merge we can perform is polynomially-bounded by our polynomial aspect ratio and poly-bounded diameter.
\begin{lemma}\label{lem:boundLargeMerges}
 The number of phases that start because of a large merge is at most $O(\log n)$, assuming $\poly(n)$ aspect ratio.
\end{lemma}
\begin{proof}
    Recall that our $\poly(n)$ aspect ratio implies
    \begin{align*}
        \frac{\max_{u,v \in \mcP} d(u,v)}{\min_{u,v \in \mcP} d(u,v)} = \poly(n).
    \end{align*}

    If phase $j$ starts because of a large merge, we know that $\delta_{s_j} > 2 \cdot \tilde{\delta}_{j-1}$ and so  $\tilde{\delta}_{j} > 2 \cdot \tilde{\delta}_{j-1}$. It follows that if $l$ phases start because of a large merge we have that there exist a pair of clusters $A, B \subseteq \mcP$ such that
    \begin{align}\label{eq:sagasaf}
        d(A, B) \geq 2^l \cdot \min_{u,v} d(u,v).
    \end{align}

    On the other hand, observe that since $d$ is well-behaved we know that it has poly-bounded diameter (\Cref{dfn:polyDiam}). In particular, for any clusters $A, B \subseteq \mcP$ we know that 
    \begin{align}\label{eq:agdhwr}
        d(A, B) \leq \poly\left(\max_{u,v \in \mcP} d(u,v) \cdot n \right).
    \end{align}
    Combining \Cref{eq:sagasaf} and \Cref{eq:agdhwr}, we have
    \begin{align*}
        l \leq \log \left( \frac{\poly\left(\max_{u,v \in \mcP} d(u,v) \cdot n \right)}{\min_{u,v} d(u,v)} \right)
    \end{align*}
    which is at most $O(\log n)$ by our assumption of $\poly(n)$ aspect ratio.
\end{proof}

The number of phases that start because of large size increases is $O(\log n)$ since each such phase increases the size of our cluster by a multiplicative constant.
\begin{lemma}\label{lem:boundSizeMerges}
 The number of phases that start because of a large size increase is at most $O(\log n)$.
\end{lemma}
\begin{proof}
    Observe that if phase $j$ starts because of a large size increase we have $|X_{s_j}| \geq \frac{3}{2} \cdot |X_{s_{j-1}}|$. However, since the size of a cluster is at most $n$, at least $1$, and non-decreasing over the course of our merges, we have that the number of such phases is at most $O(\log n)$.
\end{proof}

In order to argue that the number of phases that start because of a large drift is small, we will argue that as long as our size has only increased by a small amount, every merge happens close to $\tilde{X}_j$. In particular, we will use the following helper lemma.
\begin{lemma}\label{lem:tri}
 Suppose $X_i$ participates in a merge in the $j$th phase. Then, if $|X_i| \leq (1+\eps)\cdot |\tilde{X}_j|$ for $\eps = 1/\tilde{\Theta} \left(( k \cdot c)^{O(k)} \right)$ (for an appropriately large hidden poly-log), we have $d(\tilde{X}_j, \bar{X}_i) \leq 3 c_{\Delta}\cdot \tilde{\delta}_j$.
\end{lemma}
\begin{proof}
    Fix a $j$ and let $Y_1, Y_2 \ldots, Y_z$ be, in order, all clusters that $\tilde{X}_j$ merges with in the $j$th phase up to but not including the cluster $\bar{X}_i$. Likewise, let $X_{\leq l}$ for $l \in [0, z]$ be $\tilde{X}_j \cup \bigcup_{j \leq l}Y_j$ be $\tilde{X}_j$ after performing the first $l$ of these merges. Note that $X_{\leq z} = X_i$.
    
    By assumption and the fact that $\tilde{X}_j \subseteq X_{\leq l}$ for every $l$, we know that 
    \begin{align*}
        |X_{\leq l-1}|+ |Y_l| = |X_{\leq l}| \leq (1+\epsilon) \cdot |\tilde{X}_j| \leq (1+\epsilon) \cdot|X_{\leq l-1}|
    \end{align*}
    and so by $|X_{\leq l}| \geq |X_{\leq l-1}|$ we get 
    \begin{align*}
        |Y_l| \leq \eps \cdot |X_{\leq l}|
    \end{align*}
    and so
    \begin{align}\label{eq:asfasfsa}
        \frac{|Y_l|}{|X_{\leq l}|} \leq \eps.
    \end{align}
    
    Note that since we merge $X_{\leq l-1}$ and $Y_l$ in the $j$th phase, we have by definition of a phase (\Cref{dfn:phase}) that $d(X_{\leq l -1}, Y_l) \leq 2\tilde{\delta}_j$. Furthermore, since $d$ is well-behaved, it is weight-stable (\Cref{dfn:weightStab}) and, in particular, it follows that for each $l$ we have by \Cref{eq:asfasfsa} that 
    \begin{align}\label{eq:sagasfsag}
        d(X_{\leq l-1}, X_{\leq l}) &= d(X_{\leq l-1}, X_{\leq l-1} \cup Y_l)\nonumber\\
        &\leq \frac{|Y_l|}{|X_{\leq l-1}| + |Y_l|} \cdot d(X_{\leq l-1}, Y_l) \nonumber\\
        &= \frac{|Y_l|}{|X_{\leq l}|} \cdot d(X_{\leq l-1}, Y_l)   \nonumber\\
        &\leq \eps \cdot \tilde{\delta}_j.
    \end{align}

    Since $d$ is well-behaved, it also approximately satisfies the triangle inequality (\Cref{dfn:apxMetric}). Let $c_{\Delta}$ be the constant according to which $d$ approximately satisfies the triangle inequality for any $3$ clusters $A$, $B$ and $C$ where $|B| \geq \min(|A|, |C|)$. We have that for any $l_1, l_2, l_3 \leq z$ where $l_1 \leq l_2 \leq l_3$ that $|X_{\leq l_2}| \geq |X_{\leq l_1}| = \min(|X_{\leq l_1}|, |X_{\leq l_3}|)$ and so we have the approximate triangle inequality:
    \begin{align}\label{eq:tria}
        d(X_{\leq l_1}, X_{\leq l_3}) \leq c_{\Delta} \cdot (d(X_{\leq l_1}, X_{\leq l_2}) + d(X_{\leq l_2}, X_{\leq l_3})).
    \end{align}
    
    Thus, we have by \Cref{eq:sagasfsag} and $z-1$ applications of \Cref{eq:tria} in a binary-tree-like fashion (see \Cref{fig:bbpath-hops} for a similar application) that 
    \begin{align*}
        d(\tilde{X}_j, X_{i}) &= d(\tilde{X}_j, X_{\leq z})\\
        &\le c_\Delta\left(d(\tilde{X}_j, X_{\leq z/2}) + d(X_{\le z/2}, X_{\leq z})\right)  \\ 
        &\le c_\Delta\left(c_\Delta \left(d(\tilde{X}_j, X_{\leq z/4}) + d(X_{\le z/4}, X_{\leq z/2})\right) + c_\Delta\left(d(X_{\le z/2}, X_{\leq 3z/4}) + d(X_{\le 3z/4}, X_{\leq z})\right)\right)  \\ 
        &\;\;\vdots\\
        &\leq (2c_{\Delta})^{\log z} \cdot \eps\cdot\tilde{\delta}_j\\
        &= \eps \cdot z \cdot (c_{\Delta})^{\log z} \cdot \tilde{\delta}_j.
    \end{align*}
    
    Lastly, since $d(X_{i}, \bar{X}_{i}) \leq 2 \tilde{\delta}_j$ since we merge $X_{i}$ and $\bar{X}_{i}$ in the $j$th phase, and since $\min(|\tilde{X}_j|, |\bar{X}_{i}|) \leq |\tilde{X}_j| \leq |X_{i}|$, we may apply the approximate triangle inequality one more time to get 
    \begin{align}\label{eq:sgdagsa}
        d(\tilde{X}_j, \bar{X}_{i}) &\leq c_{\Delta}(d(\tilde{X}_j, X_{i}) + d(X_{i}, \bar{X}_{i}))\nonumber\\
        &\leq c_{\Delta}(\eps \cdot z \cdot (c_{\Delta})^{\log z} \cdot \tilde{\delta}_j + 2\tilde{\delta}_j))
    \end{align}

    By \Cref{lem:mergesPerPhase} we know that $z \leq \tilde{O}  \left(( k \cdot c)^{O(k)} \right)$. Moreover, since $c_\Delta$ is a constant and by our assumptions on $\eps$, we therefore have
    \begin{align}\label{eq:hasfsa}
        \eps \cdot z \cdot (c_{\Delta})^{\log z}  \leq 1.
    \end{align}
    Combining \Cref{eq:sgdagsa} and \Cref{eq:hasfsa}, we get
    \begin{align*}
        d(\tilde{X}_j, \bar{X}_{i}) \leq 3c_{\Delta} \cdot \tilde{\delta}_j
    \end{align*}
    as required.
\end{proof}

Using the contrapositive of the above helper lemma, we can bound the number of phases that occur because of a large drift.
\begin{lemma}\label{lem:boundDriftMerges}
 The number of phases that start because of a large drift is at most $\tilde{O}  \left(( k  \cdot c)^{O(k)} \right)$.
\end{lemma}
\begin{proof}
    By the contrapositive of \Cref{lem:tri} we have that the $j$th phase starts because of a large drift only if $|X_{s_j}| \geq (1+\eps)|X_{s_{j-1}}|$ for $\eps = 1/\tilde{\Theta}  \left(( k  \cdot c)^{O(k)} \right)$. In other words, we multiply the size of the cluster containing $x_0$ by $(1+\eps)$. By our choice of $\eps$ and the fact that the maximum cluster size is at most $n$, this can happen at most $\tilde{O}  \left(( k  \cdot c)^{O(k)} \right)$ times.
\end{proof}

Putting our bounds together, we get the following bound on the total number of phases. 

\begin{lemma}\label{lem:boundPhases}
The number of phases is at most $\tilde{O}  \left(( k  \cdot c)^{O(k)} \right)$ assuming polynomial aspect ratio.
\end{lemma}
\begin{proof}
The result is immediate from \Cref{lem:boundSizeMerges}, \Cref{lem:boundLargeMerges} and \Cref{lem:boundDriftMerges}.
\end{proof}

\subsection{Concluding our Height Bound}
We now conclude our height bound.
\heightBound*
\begin{proof}
    Fix an arbitrary $x_0 \in \mcP$. By \Cref{lem:boundPhases}, the total number of phases (as defined in \Cref{dfn:phase}) is at most $\tilde{O}  \left(( k  \cdot c)^{O(k)} \right)$. By \Cref{lem:mergesPerPhase}, the number of merges in which the cluster containing $x_0$ participates in each phase is at most $\tilde{O}  \left(( k  \cdot c)^{O(k)} \right)$. Thus, the total number of merges in which the cluster containing $x_0$ participates is at most $\tilde{O}  \left(( k  \cdot c)^{O(k)} \right)$.
\end{proof}

\section{Parallel Algorithms for HAC with  Low-Height Dendrograms}\label{sec:algo}
In this section, we present a parallel algorithm for computing $(1+\epsilon)$-approximate HAC with depth proportional to the dendrogram's height $h$ and an auxiliary parameter $\ell$ (called the \emph{bounce-back length}; see \Cref{dfn:BBPath}). Specifically, we show the following.

\conBound*

\subsection{Algorithm Description}
By scaling, we assume the minimum linkage between any two points in input $\mcP$ is at least $1$. The algorithm operates in phases, maintaining a set of active clusters, $\mathcal{C}$, initially containing each point as a singleton cluster. At the beginning of phase $t$, all pairwise linkage values between clusters are at least $(1+\epsilon)^t$. Within each phase $t$, the algorithm performs merges only between clusters whose linkage is less than $(1+\epsilon)^{t+1}$. We refer to $(1+\epsilon)^t$ and $(1+\epsilon)^{t+1}$ as the (lower and upper) thresholds for phase $t$. Since linkage $d$ is well-behaved and the aspect ratio $\rho$ is polynomially bounded, the algorithm terminates in at most $O(\log \rho)=O(\log n)$ phases (see \Cref{lem:algo_num_phases}).

Each phase consists of synchronized rounds of parallel merges. At the beginning of each round, the algorithm maintains the invariant that the minimum linkage between clusters remains at least $(1+\epsilon)^t$. During a round, the algorithm selects a subset of clusters, and each selected cluster first merges with its nearest neighbor. If the linkage value of the newly-formed clusters to some neighbors drops below the lower threshold of the current phase, it iteratively merges with its new nearest neighbor until all remaining neighbors have linkage value of at least $(1+\epsilon)^{t}$ again.
This iterative merging is crucial because linkage $d$ may not be monotonic or reducible, meaning a single merge could reduce linkage values with some neighbors below the current phase's threshold. Hence, each selected cluster continues to merge sequentially with its closest neighbor, thus ensuring the approximation guarantee, until linkage values with all its neighbors ``bounce back'' to being within the phase's thresholds again, thereby maintaining the round-invariant. 

We define the \emph{bounce-back length} $\ell$ to be an upper bound on the number of these \emph{out-of-phase} merges any cluster performs within a round. 
More generally, $\ell$ can be defined in an algorithm-independent way: starting from a merge of value $v$, it is the maximum number of merges performed by that cluster before all subsequent linkage values return to at least $v$. For our purposes, however, an algorithm-specific definition suffices---and this is always no larger than the general one.

By definition, $\ell\le h$, though tighter bounds may hold for certain linkage functions. This allows the algorithm to afford performing these merges sequentially.
However, for correctness, it is essential that the parallel merge sequences remain independent: merges in one sequence must not interfere with those in another. A central challenge, then, is to maintain this independence while ensuring that a significant fraction of clusters merge in each round. The latter is essential to bound the total number of rounds within each phase.

We now formalize the notion of ``locally-optimal paths'' taken by clusters within a round. Let $\cC_t$ denote the set of clusters present at the start of phase $t$, where each cluster in $\cC_t$ has at least one neighbor with linkage less than $(1+\epsilon)^{t+1}$. 
With $\cC_{t,0}=\cC_t$, we define $\cC_{t,r}$ to be the subset of $\cC_{t,r-1}$ consisting of clusters that still have at least one neighbor with linkage less than $(1+\eps)^{t+1}$ at the start of round $r$ in phase $t$. 
It is important to note that while other clusters not in $\cC_{t,r-1}$ might also have a neighbor with linkage less than $(1+\eps)^{t+1}$ at the start of round $r$, their closest neighbor must be a cluster that was part of $\cC_{t,r-1}$ and participated in a merge before or during round $r-1$. This is because such external clusters would have previously had all neighbors at a linkage of at least $(1+\eps)^{t+1}$; their current smaller linkage value is a direct result of a merge involving an active cluster from $\cC_{t,r-1}$. Therefore, processing only clusters within $\cC_{t,r-1}$ is sufficient.
%
\begin{restatable}[Locally-Optimal Path]{definition}{LOPath} \label{dfn:LOPath} Given a cluster $A \in \mcC_{t,r}$, the \emph{locally-optimal path} of $A$ is defined as the permutation $(B_1, B_2, \ldots)$ of the clusters in $\mcC_{t,r} -A$. satisfying the following condition: if $A_0=A$ and $A_i := A \cup \bigcup_{j \leq i} B_j$, then for all $i\ge 0$ and $j \geq i+1$,
\begin{align*}
    d(A_i, B_{i+1}) \leq d(A_i, B_j).
\end{align*}
\end{restatable}
\begin{restatable}[Bounce-Back Path]{definition}{BBPath} \label{dfn:BBPath} Given a cluster $A \in \mcC_{t,r}$ with locally-optimal path $(B_1, B_2, \ldots)$, we define the bounce-back path $\pi_A = (B_1, B_2, \ldots, B_l)$ as the maximum-length prefix of $(B_1, B_2, \ldots)$ satisfying, for all $i \in [1, l-1]$,
\begin{align*}
    d(A_0, B_1) &< (1+\epsilon)^{t+1}, \text{ and}\\
    d(A_i, B_{i+1}) &< (1+\epsilon)^t,
\end{align*}
where $A_0=A$ and $A_i := A \cup \bigcup_{j \le i} B_j$.
\end{restatable}
Intuitively, the locally-optimal path of a cluster $A$ captures the ideal sequence of merges that would occur if only cluster $A$ were allowed to greedily merge with its best available neighbor at each step. The bounce-back path is the initial segment of this path consisting of out-of-phase merges: it includes all merges performed before $A$’s linkage values with its remaining neighbors rise back above the threshold $(1+\epsilon)^{t}$ for the current phase.

Note that $l \le \ell$. If the linkage function $d$ satisfies the triangle inequality, then for all $i\in [1,l]$, we have $$d(A,B_i) \le \ell(1+\epsilon)^{t+1}.$$ This implies that the ball $B_d^{\mcC_{t,r}}(A,\ell(1+\epsilon)^{t+1})$ contains the entire bounce-back path of $A$. Moreover, the ball $B_d^{\mcC_{t,r}}(A,3\ell(1+\epsilon)^{t+1})$ contains all clusters whose bounce-back paths could potentially intersect with that of $A$. Therefore, by selecting clusters such that no selected cluster lies within the $O(\ell(1+\epsilon)^{t+1})$-radius ball of another, we ensure that their merge sequences remain unaffected by merges performed by other selected clusters. 

However, linkage functions such as Ward's do not necessarily satisfy the triangle inequality. Nevertheless, we show that for well-behaved linkage functions the bounce-back path of a cluster $A$ is contained within a ball of radius $O(\ell^{O(1)}(1+\epsilon)^{t+1})$ centered at $A$ (see \Cref{lem:BBPathlBall}). Unfortunately, this property alone is insufficient: a small cluster might appear in the bounce-back paths of two large, well-separated clusters, preventing a direct application of the approximate triangle inequality (see \Cref{dfn:apxMetric}). To rule out such cases, we introduce the notion of \emph{bounce-back shells}.
\begin{restatable}[Bounce-Back Shell]{definition}{BBShell} \label{dfn:BBShell} Given a cluster $A \in \mcC_{t,r}$ with bounce-back path $\pi_A = (B_1, B_2, $ $\ldots, B_l)$, we define the bounce-back shell $\pi_A^+$ as all $C \in \mcC_{t,r}$ such that 
\begin{align*}
    d(\{A\}\cup\pi_A, C) \leq 5c_\Delta^3\cdot\ell^{\log(c_\Delta)}\cdot (1+\epsilon)^{t+1},
\end{align*}
where $d(\{A\}\cup\pi_A,C):=\min(d(A,C),\min_{i\in[l]}d(B_i,C))$.
\end{restatable}
Essentially, selecting a cluster $A$ excludes not only all other clusters within a ball of radius $O(\ell^{O(1)}(1+\epsilon)^{t+1})$ centered at $A$, but also all clusters within similarly sized balls centered at each cluster in its bounce-back path $\pi_A$. Given this, we show that selecting clusters that do not mutually appear in each other's bounce-back shells is sufficient to guarantee that their merge sequences remain non-interfering (see \Cref{lem:well_defined}).

Finally, by the $\alpha$-packability of well-behaved linkage functions (see \Cref{dfn:packableLinkage}), we show that the number of clusters in the bounce-back shell of a cluster is bounded by $O(\alpha\ell^{O(k)})$. Consequently, in each round, at least a $\Omega(1/\ell^{O(k)})$ fraction of clusters can be selected to merge along their bounce-back paths. This guarantees that within $\tO(h\ell^{O(k)})$ rounds, the phase completes---i.e., all remaining linkage values exceed $(1+\epsilon)^{t+1}$, and the algorithm advances to the next phase.

Having developed the necessary intuition and formal definitions, we now describe the algorithm in detail. In phase $t$, the algorithm proceeds as follows:
\begin{itemize}
    \item While $\cC_{t,r}$ is non-empty, select a maximal subset $\cS_{t,r} \subseteq \cC_{t,r}$ such that for every pair of distinct clusters $A,B\in \cS_{t,r}$, neither appears in the bounce-back shell of the other.
    \item In parallel, for each cluster $S\in \cS_{t,r}$, merge $S$ with its nearest neighbor $S'$.
    \item If the new cluster $S \cup S'$ has a nearest neighbor at linkage less than $(1+\epsilon)^t$, continue merging it with its nearest neighbor, until all linkage values to neighboring clusters are at least $(1+\epsilon)^t$.
\end{itemize}
The overall algorithm is summarized in \cref{alg:PARALLEL_HAC}. The primitive $\neighbor(C)$ refers to a subroutine that returns the (exact) nearest neighbor of cluster $C$ among the active clusters $\cC$. The primitive $\merge(A,B)$ performs the standard bookkeeping tasks---updating the active cluster set $\cC$, updating the output merge sequence, etc.

To obtain the sequential list of merges that defines the HAC output, we order the merges by phase, with earlier phases appearing first. Within each phase, merges are ordered by round, and within each round, we process the clusters $S \in S_{t,r}$ in an arbitrary order, appending the full sequence of merges associated with each $S$.

\begin{algorithm}[t]
\caption{Parallel HAC for Low-Height Dendrograms}
\begin{algorithmic}[1]
    \State \textbf{Input:} Well-behaved linkage function $d$, $P \subseteq \mathbb{R}^k$ with minimum pairwise linkage $\ge 1$, $\epsilon > 0$
    \State \textbf{Output:} Sequence of merges defining HAC
    \State $\cC \gets \{ \{ p \} : p \in P \}$, $t \gets 0$
    \While{$|\cC| > 1$}
        \State $\cC_{t,0} \gets \{C \in \cC : d(C,\neighbor(C)) < (1+\epsilon)^{t+1}\}$ \Comment{$\neighbor(C)$ is nearest neighbor in current $\cC$} 
        \State $r \gets 0$
        \While{$|\cC_{t,r}|>0$}
            \State // $\pi_A^+$ denotes the bounce-back shell of cluster $A$
            \State $\cS_{t,r} \gets$ maximal subset of $\cC_{t,r}$ such that $A \notin \pi^+_{B}$ and $B \notin \pi^+_A$ for all $A,B \in \cS_{t,r}$
            \ForAll{$S \in \cS_{t,r}$ in parallel} \label{line:parallel_merge}
                \State $S \gets \merge(S, \neighbor(S))$ \label{line:mis}\Comment{update $\cC$}
                \While{$d(S,\neighbor(S)) < (1+\epsilon)^{t}$}
                    \State $S \gets \merge(S, \neighbor(S))$\Comment{update $\cC$}\label{line:merge}
                \EndWhile
            \EndFor
            \State $r \gets r+1$
            \State $\cC_{t,r} \gets \{C \in \cC_{t,r-1} : d(C,\neighbor(C)) < (1+\epsilon)^{t+1}\}$\label{line:construct_active_sets}
        \EndWhile
        \State $t \gets t + 1$
    \EndWhile
\end{algorithmic}
\label{alg:PARALLEL_HAC}
\end{algorithm}

\subsection{Correctness}
In this section, we prove that \cref{alg:PARALLEL_HAC} produces a valid $(1+\epsilon)$-approximate HAC merge sequence. We begin by showing that the merges performed in each round are well-defined. Specifically, we prove that the parallel merge sequences are independent: each cluster can follow its bounce-back path without being affected by merges performed by other clusters in the same round. As a first step, we show that the bounce-back path of a cluster is contained within a ball of radius $O(\ell^{O(1)}(1+\epsilon)^{t+1})$ centered at the cluster.

\begin{figure}
    \centering
    \includegraphics[width=\linewidth]{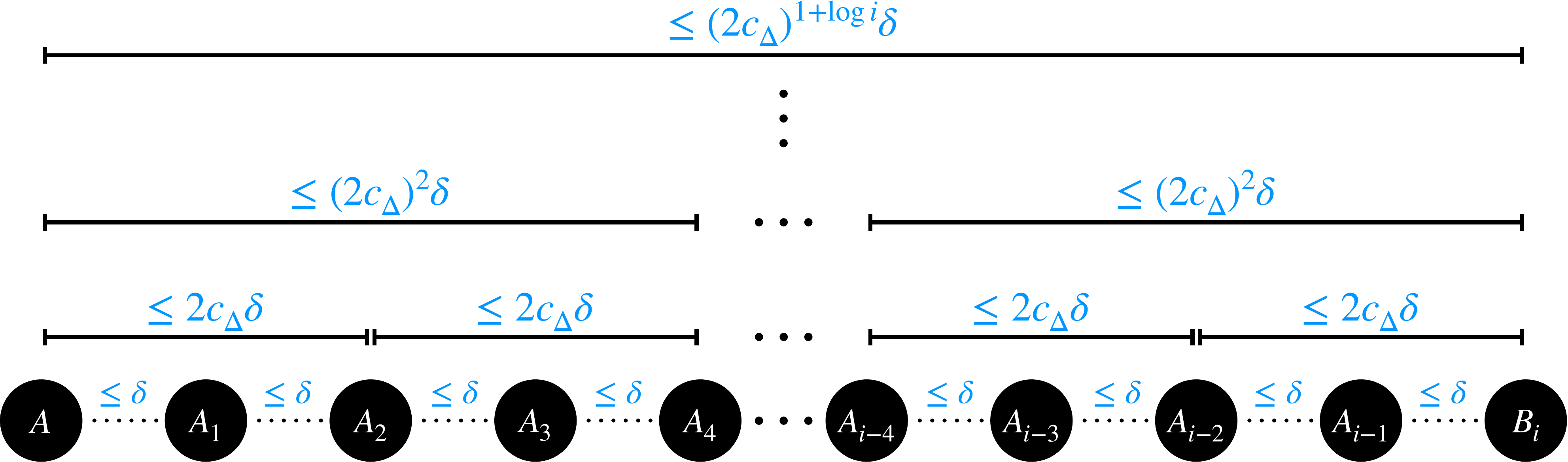}
    \caption{Illustration of the recursive application of the approximate triangle inequality to bound the distance from cluster $A$ to the $i^\text{th}$ cluster on its bounce-back path.}
    \label{fig:bbpath-hops}
\end{figure}

\begin{restatable}{lemma}{BBPathlBall}\label{lem:BBPathlBall}
    Let $A \in \cC_{t,r}$ be a cluster with bounce-back path $\pi_A = (B_1,B_2,\ldots,B_l)$. Then for all $i\in [1,l]$,
    $$d(A,B_i) \le 2c_\Delta\cdot\ell^{\log(c_\Delta)}\cdot(1+\epsilon)^{t+1}.$$
\end{restatable}
\begin{proof}
    Let $A_0=A$, and define $$A_i := A \cup \bigcup_{j \leq i} B_j$$ to be the cluster formed after the first $i$ merges along the bounce-back path $\pi_A$. 
    
    Since we are in phase $t$, each merge along the bounce-back path satisfies
    \begin{align*}
        d(A_{i-1},B_i) \le (1+\epsilon)^{t+1}, \text{ for all $i\in [1,l]$.}
    \end{align*}
    Because $d$ is weight-stable (\Cref{dfn:weightStab}), it follows that
    \begin{align*}
        d(A_{i-1},A_{i}) \le \frac{|B_i|}{|A_{i-1}|+|B_i|}\cdot d(A_{i-1},B_i) \le (1+\epsilon)^{t+1}.
    \end{align*}
    We now bound the distance from $A$ to each $B_i$ using the $c_{\Delta}$-approximate triangle inequality (see \Cref{dfn:apxMetric}), applied recursively along the merge sequence. Since each merge increases cluster size, we have $|A_i|\ge |A|$ for all $i$. Thus, by approximate triangle inequality (applied in a binary-tree-like fashion; see \Cref{fig:bbpath-hops}), for each $i\in [1,l]$, we obtain:
    \begin{align*}
        d(A,B_i) &\le c_\Delta(d(A,A_{i/2})+d(A_{i/2},B_i))\\
        &\le c_\Delta\left(c_\Delta\left(d(A,A_{i/4})+d(A_{i/4},A_{i/2})\right) + c_\Delta\left(d(A_{i/2},A_{3i/4}) + d(A_{3i/4},B_i)\right)\right)\\
        &\;\;\vdots\\
        &\le (2c_\Delta)^{\log(i)+1}\cdot(1+\epsilon)^{t+1}\\ 
        &= 2c_\Delta\cdot i^{\log(c_\Delta)}\cdot(1+\epsilon)^{t+1}.
    \end{align*}
    Since $l \le \ell$, the bound follows.
\end{proof}

Next, we prove that the bounce-back paths of distinct clusters in $\cS_{t,r}$ are sufficiently far apart. 

\begin{lemma} \label{lem:seperation}
    Let $X, Y \in \cS_{t,r}$ with bounce-back paths 
    $$\pi_X = \{X'_1, X'_2, \ldots, X'_x\} \quad \text{and} \quad \pi_Y = \{Y'_1, Y'_2, \ldots, Y'_y\},$$ 
    and bounce-back shells $\pi_X^+$ and $\pi_Y^+$. Also, define the sequence of merged clusters,
    $$
    X_0 := X, \quad X_i := X \cup \bigcup_{j \le i} X'_{j} \quad \text{for } i \in [1, x].
    $$
    The sequence of clusters $Y_j$ is defined similarly. Then, for all $i \in [1, x]$, $j \in [1, y]$, and for all $A \in \{X_i, X'_i\}$, $B \in \{Y_j, Y'_j\}$, we have 
    $$d(A, B) \ge (1+\epsilon)^{t+1}.$$
\end{lemma}
\begin{proof}
    Fix $i \in [1,x]$ and $j \in [1,y]$. From the proof of \Cref{lem:BBPathlBall}, we have:
    \begin{align*}
        d(X, X_i),\;d(X, X'_i),\;d(Y, Y_j),\;d(Y, Y'_j) \le 2c_\Delta \cdot \ell^{\log(c_\Delta)} \cdot (1+\epsilon)^{t+1}.
    \end{align*}
    From the definition of $\cS_{t,r}$, we also have:
    \begin{align*}
        d(X, Y),\;d(X'_i, Y),\; d(X, Y'_j)  \ge 5c_\Delta^3 \cdot \ell^{\log(c_\Delta)} \cdot (1+\epsilon)^{t+1},
    \end{align*}

    We now consider all possible combinations of $A \in \{X_i, X'_i\}$ and $B \in \{Y_j, Y'_j\}$. In each case, we show that $d(A, B) \ge (1+\epsilon)^{t+1}$.

    \begin{description}
      \item[Case 1: $A = X_i$, $B = Y_j$] ~\\
      First, we lower-bound $d(X_i, Y)$:
      \begin{align*}
      d(X, Y) \le c_\Delta &\left(d(X, X_i) + d(X_i, Y)\right) \\
      \Rightarrow \quad d(X_i, Y) &\ge \frac{1}{c_\Delta} \cdot d(X, Y) - d(X, X_i) \\
      &\ge \frac{1}{c_\Delta} \cdot \left(5c_\Delta^3 \ell^{\log(c_\Delta)} (1+\epsilon)^{t+1}\right) - 2c_\Delta \ell^{\log(c_\Delta)} (1+\epsilon)^{t+1} \\
      &\ge 3c_\Delta^2 \cdot \ell^{\log(c_\Delta)} (1+\epsilon)^{t+1}.
      \end{align*}
    
      Using this:
      \begin{align*}
      d(X_i, Y) \le c_\Delta (d(X_i, Y_j) &+ d(Y_j, Y)) \\
      \Rightarrow \quad d(X_i, Y_j) &\ge \frac{1}{c_\Delta} \cdot d(X_i,Y) - d(Y, Y_j) \\
      &\ge \frac{1}{c_\Delta} \cdot \left(3c_\Delta^2 \ell^{\log(c_\Delta)} (1+\epsilon)^{t+1}\right) - 2c_\Delta \ell^{\log(c_\Delta)} (1+\epsilon)^{t+1} \\
      &\ge (1+\epsilon)^{t+1}.
      \end{align*}
    
      \item[Case 2: $A = X_i$, $B = Y'_j$] ~\\
      Again, apply the approximate triangle inequality:
      \begin{align*}
      d(X, Y'_j) \le c_\Delta (d(X, X_i) &+ d(X_i, Y'_j)) \\
      \Rightarrow \quad d(X_i, Y'_j) &\ge \frac{1}{c_\Delta} \cdot d(X, Y'_j) - d(X, X_i) \\
      &\ge \frac{1}{c_\Delta} \cdot \left(5c_\Delta^3 \ell^{\log(c_\Delta)} (1+\epsilon)^{t+1}\right) - 2c_\Delta \ell^{\log(c_\Delta)} (1+\epsilon)^{t+1} \\
      &\ge (1+\epsilon)^{t+1}.
      \end{align*}
    
      \item[Case 3: $A = X'_i$, $B = Y_j$] ~\\
      Symmetric to Case 2. We have $d(X'_i, Y_j) \ge (1+\epsilon)^{t+1}$ by the same argument.
    
      \item[Case 4: $A = X'_i$, $B = Y'_j$] ~\\
      Assume WLOG that $|X'_i| \le |Y'_j|$. Then:
      \begin{align*}
      d(X'_i,Y) \le c_\Delta (d(X'_i, Y'_j) &+ d(Y, Y'_j)) \\
      \Rightarrow \quad d(X'_i, Y'_j) &\ge \frac{1}{c_\Delta} \cdot d(X'_i,Y) - d(Y, Y'_j) \\
      &\ge \frac{1}{c_\Delta} \cdot \left(5c_\Delta^3 \ell^{\log(c_\Delta)} (1+\epsilon)^{t+1}\right) - 2c_\Delta \ell^{\log(c_\Delta)} (1+\epsilon)^{t+1} \\
      &\ge (1+\epsilon)^{t+1}.
      \end{align*}
    \end{description}
\end{proof}

It follows that every cluster in $\cS_{t,r}$ proceeds to merge exactly along its bounce-back path.
\begin{lemma} \label{lem:well_defined}
    In round $r$ of phase $t$, the algorithm merges every cluster $A \in \cS_{t,r}$ precisely along its bounce-back path $\pi_A$.
\end{lemma}
\begin{proof}
    Fix a cluster $X \in \cS_{t,r}$ with bounce-back path $\pi_X = (X'_1, X'_2, \ldots, X'_x)$. We show that $X$ merges with each $X'_i$ in order, and with no other clusters, during round $r$ of phase $t$.
    
    By definition, the bounce-back path is a prefix of the locally-optimal path and represents the best available merges for $X$---unless a better option is created by merges involving other clusters.
    
    Let $Y\in \cS_{t,r}\setminus\{X\}$ with bounce-back path $\pi_Y=(Y'_1,Y'_2,\ldots,Y'_y)$. By \Cref{lem:seperation}, we have $d(X_i,Y'_i)\ge (1+\eps)^{t+1}$ for all $i,j$, so the cluster containing $X$ never prefers any cluster on the bounce-back path of $Y$.
    Moreover, \Cref{lem:seperation} also guarantees that $d(X_i,Y_j) \ge (1+\eps)^{t+1}$ for all intermediate clusters $Y_j$ formed along $Y$'s merge sequence. Thus, $X$ never prefers any cluster created by another merge sequence during the same round. 
    It follows that the merge sequence for $X$ proceeds exactly along $\pi_X$.
\end{proof}

Now that we have shown \Cref{alg:PARALLEL_HAC} is well-defined, we proceed to prove that it produces a $(1+\epsilon)$-approximate HAC.

\begin{lemma} \label{lem:algo_correctness}
    \Cref{alg:PARALLEL_HAC} gives a $(1+\epsilon)$-approximate HAC. 
\end{lemma}
\begin{proof}
    We prove the claim by induction on the rounds within each phase.

    Assume that up to the beginning of round $r$ of phase $t$, all merges performed have been $(1+\epsilon)$-approximate. By the round-invariant, we know that at the start of round $r$, all clusters in $\cC$ have pairwise linkage values at least $(1 + \epsilon)^t$. 

    To analyze the merges in round $r$, we consider an arbitrary sequential order on the clusters in $\cS_{t,r}$. This is valid because, by \Cref{lem:seperation}, the bounce-back paths of clusters in $\cS_{t,r}$ are disjoint and non-interfering. In particular, the merges performed by any cluster do not affect the merge sequence of any other cluster in $\cS_{t,r}$. Therefore, we may analyze the merges one cluster at a time, as if they were performed sequentially. 
    
    Let $X, Y \in \cS_{t,r}$ be two such clusters, and suppose $X$ appears before $Y$ in this sequential order. Let $X_x$ denote the final cluster formed after merging all clusters on the bounce-back path of $X$ with $X$. By construction, 
    $$d(X_x, \neighbor(X_x)) \ge (1+\epsilon)^t.$$
    Moreover, by \Cref{lem:seperation},
    $$d(X_x, Y) \ge (1+\epsilon)^{t+1}.$$
    Thus, all linkage values involving $Y$ are still at least $(1+\epsilon)^t$ after $X$ has completed its merge sequence. The same argument applies to any cluster that appears later in the ordering.
    
    We now show that the sequence of merges performed by the algorithm for $X$ is $(1+\epsilon)$-approximate. Let $\pi_X=(X'_1,X'_2,\ldots,X'_x)$ denote the bounce-back path of $X$. Define $X_0 = X$, and for each $i \in [1, x]$, let $X_i := X\cup \bigcup_{j \le i} X'_{j}$.
    \begin{itemize}
        \item Before the first merge, all linkage values are at least $(1+\eps)^t$, so the optimal pair also has linkage at least $(1+\eps)^t$ as well.
        Since $d(X_0, X'_1) < (1+\epsilon)^{t+1}$, the first merge is $(1+\epsilon)$-approximate.
    
        \item For $i \in [2, x]$, we have $d(X_{i-1}, X'_i) < (1+\epsilon)^t$ by definition of the bounce-back path. Since the only linkage values that may fall below $(1+\eps)^t$ involve the cluster containing $X$, and the algorithm follows the greedy locally-optimal path, it performs the optimal merge at each step.
    \end{itemize}
    
    Therefore, the merge sequence for $X$ is $(1+\epsilon)$-approximate, and appending it to the prior merge sequence preserves the inductive invariant. By induction, the entire merge sequence produced by the algorithm is $(1+\epsilon)$-approximate in each phase.

    A similar inductive argument applies across phases, completing the proof that all merges performed by \Cref{alg:PARALLEL_HAC} are $(1+\epsilon)$-approximate.
\end{proof}

\subsection{Work and Depth}
In this section, we analyze the work and depth of \Cref{alg:PARALLEL_HAC}. We begin by bounding the number of phases executed by the algorithm.

\begin{lemma}\label{lem:algo_num_phases}
    The number of phases in \Cref{alg:PARALLEL_HAC} is at most $O(\log n)$.
\end{lemma}
\begin{proof}
The proof follows by a similar line of argument as \Cref{lem:boundLargeMerges}. By scaling, we assumed $\min_{u,v\in \mcP}d(u,v)=1$, and our $\poly(n)$ aspect ratio assumption implies
\begin{align*}
    \max_{u,v\in \mcP} d(u,v) = \poly(n).
\end{align*}
Let $T$ denote the total number of phases in \Cref{alg:PARALLEL_HAC}. Then, there exists some pair of clusters $A,B \subseteq \mcP$ such that
\begin{align*}
    d(A,B) \ge (1+\eps)^T.
\end{align*}
Whereas, by the poly-bounded diameter property of well-behaved linkage $d$ (\Cref{dfn:polyDiam}), we know that
\begin{align*}
    d(A,B) \le \poly\left(\max_{u,v \in \mcP}d(u,v)\cdot n\right).
\end{align*}
Therefore, combining the above gives us
\begin{align*}
    T \le \log_{(1+\eps)} \left(\poly\left(\max_{u,v \in \mcP}d(u,v)\cdot n\right)\right) = O(\log n).
\end{align*}    
\end{proof}

We now proceed to bound the number of rounds in a single phase.
\begin{lemma}\label{lem:BBShellSize}
    Given $A\in C_{t,r}$, the number of clusters in the bounce-back shell of $A$ is
    $$|\pi_A^+| \le O\left(\alpha \cdot c_\Delta^{O(k)}\cdot\ell^{O(k\log(c_\Delta))}\right).$$
\end{lemma}
\begin{proof}
    Let $$\beta = 5c_\Delta^3\cdot\ell^{\log(c_\Delta)}.$$ 
    By definition, a bounce-back shell contains all clusters within the $\beta(1+\eps)^{t+1}$-radius ball centered at each cluster of the corresponding bounce-back path (see \Cref{dfn:BBShell}).  
    By the $\alpha$-packability of well-behaved linkage functions, each such ball contains at most 
    $$O\left(\alpha\cdot\left(\frac{\beta(1+\eps)^{t+1}}{(1+\eps)^t}\right)^{O(k)}\right) = O\left(\alpha\cdot(\beta\cdot(1+\eps))^{O(k)}\right).$$
    Since the number of clusters on a bounce-back path is at most $\ell$, the bound follows.
\end{proof}

\begin{lemma} \label{lem:num_rounds}
    The number of rounds in each phase of \Cref{alg:PARALLEL_HAC} is at most $$O\left(\alpha\cdot h\cdot c_{\Delta}^{O(k)}\ell^{O(k\log(c_\Delta))}\log n\right).$$
\end{lemma}
\begin{proof}
    Fix a phase $t$, and let $\cC_t$ denote the set of clusters at the beginning of the phase. Let 
    $$\eta = O\left(\alpha\cdot c_{\Delta}^{O(k)}\cdot\ell^{O(k\log(c_\Delta))}\right)$$
    be an upper bound on the number of clusters that may appear in the bounce-back shell of any cluster (as proved in \Cref{lem:BBShellSize}). 

    We bound the number of rounds in each phase via a potential argument. Let $R$ denote the total number of rounds in phase $t$.
    For any cluster $C \in \cC_t$, note that
    \begin{align*}
        |\{r : C \in \cS_{t,r} \text{ for all } r\}| \le h,
    \end{align*}
    i.e., cluster $C$ can be picked to merge along its bounce-back path at most $h$ times during the phase.
    For each $C \in \cC_t$, define its potential at round $r$ as
    \begin{align*}
        \Phi_r(C) := \begin{cases}
            h - |\{S_{t,r'} : C \in S_{t,r'}, r' \le r\}|, & \text{if } C \in \cC_{t,r},\\
            0, & \text{otherwise}.
        \end{cases}
    \end{align*}
    Define the total potential at round $r$ as
    \begin{align*}
        \Phi_r = \sum_{C\in \cC_t} \Phi_r(C).
    \end{align*}

    In each round $r$, when a cluster $C$ is picked, it rules out at most $\eta$ other clusters from being picked. Therefore, the number of clusters picked in round $r$ satisfies
    \begin{align*}
        |\cS_{t,r}| \ge \frac{1}{\eta}|\cC_{t,r}|.
    \end{align*}
    Also, since $\Phi_r(C) \le h$ for each $C$, we have
    \begin{align*}
        \Phi_r \le h|\cC_{t,r}| \implies |\cC_{t,r}| \ge \Phi_r / h.
    \end{align*}
    Thus,
    \begin{align*}
        \Phi_{r+1}
        &\le \Phi_r - \frac{1}{\eta}|\cC_{t,r}|
        \le \Phi_r \left(1 - \frac{1}{h\eta}\right)
        \le \Phi_0 \left(1 - \frac{1}{h\eta}\right)^r 
        \le \Phi_0 e^{-r/(h\eta)}.
    \end{align*}
    Since $\Phi_0 \le h|\cC_t|$ and $\Phi_R = 1$, it follows that
    \begin{align*}
        R \le h\eta \log\left(h|\cC_t|\right) = O(h\eta \log n).
    \end{align*}
\end{proof}

\noindent\textbf{Remark.} The above analysis also applies if, instead of using bounce-back paths, we considered the more natural (and perhaps simpler) idea of using the maximal prefix of the locally-optimal path consisting of merges with linkage less than $(1+\epsilon)^{t+1}$---i.e., continuing merges until all linkage values to that cluster exceed the phase threshold.
However, in this case, the length of such a path can be as large as $h$, leading to a weaker bound of $\tO(h^{O(k)})$ rounds per phase. In contrast, using bounce-back paths leads to the stronger bound of $\tO(h\ell^{O(k)})$.


Henceforth, we assume $\alpha=\tO(1)$ and $c_\Delta=O(1)$, and simplify the presentation of bounds accordingly. 
Define $\wnn$ to be a parameter such that the $\tO(\ell^{O(k)})$ nearest neighbors of a cluster can be computed in $\tO(\wnn+\ell^{O(k)})$ work and $\tO(\ell^{O(k)})$ depth---i.e., both the work and depth depend on the output size. Naively, this operation can be implemented via a linear scan over $\cC$ followed by sorting, which requires $\tO(n)$ work and $\tO(1)$ depth; here $\wnn = \tO(n)$. All results that follow are parameterized by $\wnn$, with concrete bounds and implementations provided later for the linkage functions considered in this paper.

\begin{lemma}\label{lem:compute_bbshell}
    Given $A\in \cC_{t,r}$, the bounce back path and the bounce-back shell of $A$ can be computed in $\tO(\wnn\cdot\ell^{O(k)})$ work and $\tO(\ell^{O(k)})$ depth.
\end{lemma}
\begin{proof}
    We first describe how to compute the bounce-back path $\pi_A=\{B_1,B_2,\ldots,B_l\}$ of a cluster $A\in \cC_{t,r}$. 
    
    Let $X = B_d^\cC(A,O(\ell^{O(1)})\cdot(1+\eps)^{t+1})$. By \Cref{lem:BBPathlBall}, we have $\pi_A \subseteq X$. Moreover, by the packability property of well-behaved linkage functions (\Cref{dfn:packableLinkage}), the size of $|X|$ is bounded as $|X| \le \tO(\ell^{O(k)})$. 
    
    We compute $\pi_A$ as follows:
    \begin{itemize}
        \item Compute $X$ using an $\tO(\ell^{O(k)})$-nearest neighbors query.
        \item Let $A_0 = A$.
        \item For each $i\in [1,l]$, 
        \begin{itemize}
            \item Let $\cA_{i-1} = \{B_1,B_2,\ldots, B_{i-1}\}$.
            \item Set $B_i$ to be the nearest neighbor of $A_{i-1}$ in $X\setminus \cA_{i-1}$.
            \item Set $A_i$ to $A_{i-1}\cup B_i$.
        \end{itemize}
    \end{itemize}
    Since \Cref{lem:BBPathlBall} ensures that all required nearest neighbors lie within $X$, each $B_i$ can be computed via a scan over $X$, which takes $O(|X|)$ work and $\tO(1)$ depth. Computing $X$ requires $\tO(\wnn)$ work and $\tO(\ell^{O(k)})$ depth. Therefore, the entire bounce-back path can be computed in $\tO(\wnn \cdot \ell^{O(k)})$ work and $\tO(\ell)$ depth.

    To compute the bounce-back shell, we compute the $\tO(\ell^{O(k)})$ nearest neighbors of each $B_i$ which computes the ball of radius $O(\ell^{O(1)})$ centered at $B_i$, in parallel across all $i \in [1, \ell]$. This step also takes $\tO(\wnn \cdot \ell^{O(k)})$ work and $\tO(\ell^{O(k)})$ depth overall.
\end{proof}

We now analyze the overall work and depth of \Cref{alg:PARALLEL_HAC}.
\begin{lemma}\label{lem:algo_work_depth}
    \Cref{alg:PARALLEL_HAC} runs in $\tO(\wnn\cdot n\cdot h\ell^{O(k)})$ work and $\tO(h\ell^{O(k)})$ depth.
\end{lemma}
\begin{proof}
    Most steps of the algorithm are simple and can be implemented in $O(n)$ work and $O(\log n)$ depth per round within a phase. The only non-trivial step is the computation of the maximal subset $\cS_{t,r}$ (\cref{line:mis}). 

    To compute $\cS_{t,r}$, we construct a graph $H_{t,r}$ with one node for each cluster in $\cC_{t,r}$. For each $A\in \cC_{t,r}$, we compute its bounce-back shell $\pi_A^+$ in $\tO(\wnn\ell^{O(k)})$ work and $O(\ell)$ depth (by \Cref{lem:compute_bbshell}). For every $B\in \pi_A^+$, we add an edge $(A,B)$ to $H_{t,r}$. The set $\cS_{t,r}$ then corresponds to a maximal independent set (MIS) in this graph. 
    An MIS in a graph with $n$ nodes and $m$ edges can be computed in $O(m)$ work and $O(\log n)$ depth~\cite{Fischer2019Tight}. Since each bounce-back shell contains at most $\tO(\ell^{O(k)})$ clusters (by \Cref{lem:BBShellSize}), the total number of edges in $H_{t,r}$ is bounded by $\tO(n\cdot \ell^{O(k)})$.

    Thus, each round requires $O(n\wnn\ell^{O(k)})$ work and $\tO(1)$ depth. By \Cref{lem:algo_num_phases} and \Cref{lem:num_rounds}, which bounds the total number of phases and rounds by $\tO(h\ell^{O(k)})$, the overall work and depth bounds follow.
\end{proof}

\subsection{Final Result}
We now prove the main result of the algorithm.

\conBound*
The proof follows by \Cref{lem:algo_correctness,lem:algo_work_depth}. Plugging our height bounds (\Cref{lem:heightBound}), the naive algorithm for computing $\tO(\ell^{O(k)})$-nearest neighbors, and $\ell \le h$ into \Cref{thm:conBound} immediately gives near-quadratic work parallel algorithms with $\tilde{O}(1)$ depth in constant dimensions for any well-behaved linkage functions, as summarized below.
\begin{restatable}{corollary}{finalAlg}\label{thm:finalAlg}
    Suppose $d$ is a well-behaved linkage function. Then, for any $\eps > 0$, there exists a $(1+\eps)$-approximate parallel HAC algorithm in $\mathbb{R}^k$ with $\tO(n^2k^{O(k^2)}\log^{O(k)}n)$ work and $\tO(k^{O(k^2)}\log^{O(k)}n)$ depth, assuming $\poly(n)$ aspect ratio.
\end{restatable}


\subsubsection{Results for Centroid}
We now describe how to efficiently compute $p$-nearest neighbors during centroid HAC using cover trees~\cite{Elkin2023ANew,Beygelzimer2006CoverTrees,Gu2022ParallelCover}. In particular, we use the following result on parallel batch-dynamic cover trees from \cite{Gu2022ParallelCover}.

\begin{theorem}[Parallel Batch-Dynamic Cover Trees~\cite{Gu2022ParallelCover}]\label{thm:cover_trees}
    There exists a parallel data structure that maintains a dynamically updated set $S \subseteq \R^k$ (initially $S=\emptyset$) and supports the following operations, assuming $\poly(n)$ aspect ratio:
    \begin{enumerate}
        \item \textbf{Batch Insert/Delete:} A batch of $m$ points can be inserted/deleted in $\tO(m\cdot2^{O(k)})$ expected work and $\tO(1)$ depth with high probability.
        \item \textbf{$p$-Nearest Neighbor Search:} Given a query point $q$ and integer $p\ge 1$, the $p$ nearest neighbors of $q$ in $S$ can be returned in $\tO(p\cdot 2^{O(k)})$ work.
    \end{enumerate}
\end{theorem}
The idea is to maintain the centroids of the active clusters $\cC$ using a cover tree. The cover tree can be initialized via a batch insert operation in $\tO(n2^{O(k)})$ expected work and $\tO(2^{O(k)})$ depth with high probability. 
Each $\tO(\ell^{O(k)})$-NNS query on the cover tree takes $\tO(\ell^{O(k)})$ work and depth---i.e., we have $\wnn = \tO(\ell^{O(k)})$. After each round, the cover tree is updated by performing a batch delete (to remove clusters involved in merges) followed by a batch insert (to add the newly created clusters), requiring overall $\tO(n\cdot 2^{O(k)})$ expected work and $\tO(1)$ depth with high probability.
Hence, maintaining the cover tree does not incur an additional (asymptotic) overhead, up to log factors.

Finally, using the bound $\ell \le h$, we obtain the following result for centroid HAC from \Cref{thm:conBound}.

\parAlgCent*

\subsubsection{Results for Ward's}
For Ward's, we can obtain a tighter bound on the parameter $\ell$. To do so, we introduce the notion of \emph{weak-reducibility}, which informally states that as long as a merge does not involve the pair with the largest linkage value in a triple of clusters, the linkage function behaves as if it were reducible.

\begin{restatable}[Weak-Reducible]{definition}{weakRed} \label{dfn:weakRed}
Linkage function $d$ is weak-reducible if for any $A, B, C \subseteq \mathbb{R}^k$ such that $d(A,B) \ge \max(d(A,C),d(B,C))$, we have
\begin{align*}
    d(A \cup C, B) \geq d(B,C) \quad \text{and} \quad d(B\cup C,A) \ge d(A,C).
\end{align*}
\end{restatable}

\begin{lemma}\label{lem:wards_weakRed}
    Ward's linkage function $\dWard$ is weak-reducible (\Cref{dfn:weakRed}).
\end{lemma}
\begin{proof}
    Suppose we have three clusters $A,B,C \subseteq \R^k$ such that 
    $$\dWard(A,B) \ge \max(\dWard(A,C),\dWard(B,C)).$$ 
    Then, by the Lance-Williams form (\Cref{lem:LanceWilliamsWards}), we have $\dWard(A\cup C,B)-\dWard(B,C)$ is
    \begin{align*}
        \frac{|A|+|B|}{|A|+|B|+|C|}&\dWard(A,B)-\frac{|A|}{|A|+|B|+|C|}\dWard(B,C)-\frac{|B|}{|A|+|B|+|C|}\dWard(A,C)\\
        &\ge 0.
    \end{align*}
    A symmetric argument shows that $\dWard(B\cup C,A) \ge \dWard(A,C)$, as required.
\end{proof}

Thus, after a cluster $A \in \cC_{t,r}$ merges with its nearest neighbor, its linkage values with all other clusters remain at least $(1+\epsilon)^t$, implying that $\ell = 1$ for Ward's linkage. 

Since Ward's linkage is not a metric, we cannot apply cover trees directly. Instead, we adopt a bucketing-based approach (also used in~\cite{abboudhac}). Specifically, we group clusters into buckets based on sizes: clusters with sizes in the range $[2^i,2^{i+1})$ are placed in the $i$th bucket, and we maintain a cover tree over the centroids of clusters within each bucket. There are at most $O(\log n)$ such buckets. 

By the approximate form of Ward's linkage (see \Cref{lem:wardAlt}), the nearest neighbor of a cluster within a bucket---measured by centroid distances---is a $2$-approximate nearest neighbor with respect to Ward's linkage. Thus, to compute a set containing the $\tO(\ell^{O(k)})$ nearest neighbors under Ward's linkage, it suffices to compute the $\tO((2\cdot\ell)^{O(k)})$ nearest neighbors by centroid distances within each bucket. Updating the cover trees after merges follows naturally via batched deletions and insertions within each bucket.

Plugging this into \Cref{thm:conBound}, we obtain the following result for Ward's linkage HAC.

\parAlgWard*

\section{Parallel Hardness for HAC in Arbitrary Dimensions}\label{sec:hardness}

In this section, we prove the CC-hardness of centroid HAC. In particular, we show the hardness of the following decision version of HAC.
\decisionHAC*

Recall, CC-hardness is defined as follows.

\CCHard*
Our CC-hardness, then, is given by the following.
\ccHardness*

To prove \Cref{thm:ccHardness}, we reduce from the telephone communication problem (TCP). In the TCP, we receive a series of $n$ phone calls, which each have a start time and an end time, and a capacity $\kappa$ to service calls. 
When a call comes in, if we are currently servicing fewer than $\kappa$ calls we accept the call and service it until its end time. 
Otherwise, we drop the call and do not service it. The question is whether the $t$-th call is serviced. Formally:

\begin{definition}[Telephone Communication Problem]
    We are given $n$ phone calls, $(S_1,F_1), \ldots, (S_n,F_n)$ where $S_i < S_{i+1}$ for all $i \in [n-1]$, a capacity $\kappa$, and $t \in [n]$. If we are servicing less than $\kappa$ calls at time $s_i$, then we will service call $i$ from time $S_i$ to $F_i$. Otherwise, we do not service call $i$. Our goal is to decide whether call $t$ gets serviced.
\end{definition}

TCP was first introduced in \cite{Ramachandran-1991} where they showed that it is CC-hard.

\begin{lemma}[\cite{Ramachandran-1991}]
    TCP is CC-hard.
\end{lemma}


\subsection{Reduction from TCP to Centroid HAC}

We reduce from TCP to HAC. Suppose we are given an instance $\mathcal{I} = ((S_1,F_1), \ldots, (S_n,F_n), \kappa, t)$ of TCP. We will build an instance of HAC in $\mathbb{R}^n$ so that solving it gives us the solution to $\mathcal{I}$. We will use weighted points as it is not difficult to place many points close together to achieve the same result. 
Let $E_1, \ldots, E_{2n}$ be the list of events, $S_i$ and $F_i$ for $i \in [n]$, in the order that they occur in $\mathcal{I}$. Define the following functions:

\begin{itemize}
    \item $e(E_i) = i$ is the number of the event
    \item $a(E_i) = \text{Number of active calls immediately before $e_i$}$
    \item $f(E_i) = \text{Number of calls that finished before $e_i$}$
\end{itemize}

We start by placing a heavy point $C$ at the origin. Each call gets its own axis. For call $i$ we place a point, $S_i$, on the positive $i$-axis. Further along the positive $i$-axis we have a heavy point, $R_i$. We set up the points so that $S_i$ merges with $C$ if call $i$ is serviced in $\mathcal{I}$ and $S_i$ merges with $R_i$ if not. If $i < j$, then it will be the case that $S_i$ is closer than $S_j$ to $C$ so that HAC has to make a decision for earlier phone calls first. 

HAC will be able to distinguish whether $S_j$ should merge with $C$ or $R_j$ based on how many points have merged with $C$. For $i < j$, if $S_i$ merges with $C$, then $C$ gets slightly dragged off the center along the $i$th axis, increasing the distance between $C$ and $S_j$. We might want it to be the case that when $S_j$ has to decide which direction to merge, $C$ is off center in $a(S_j)$ directions so that HAC only has to distinguish between whether $C$ is off center in $\kappa$ coordinates or in less than $\kappa$. However, it is not clear how to do this. Instead, we will add a point $F_i$ on the negative $i$th axis for $i \in [n]$ so that $F_i$ merges with $C$ if and only if $S_i$ does not. To accomplish this, we also add a point $L_i$ outside of $F_i$ that serves a purpose similar to that of $R_i$ for $S_i$. The $S$ and $F$ points will be placed increasingly far from the center in order of their event numbers. Then, when $S_j$ has to make a decision on which direction to merge, $C$ is off center in $f(S_j) + a(S_j)$ coordinates, one for each call that has finished and one for each active call. We will let $r_j$ be the distance between $S_j$ and $R_j$ and set it so that HAC merges $S_j$ with $C$ if and only if $C$ is off center in less than $f(S_j) + \kappa$ directions at the time of $S_j$ merging. On the other side, we will let $l_j$ be the distance between $F_j$ and $L_j$ and set it so that HAC merges $F_j$ with $C$ if and only if $S_j$ did not merge with $C$.

Ideally, we might want the events to merge in order. That is, the $i$th merge is between $E_i$ and either $C$ or $O_i$ where $O_i$ is the outer point of $E_i$ (either an $L$ or $R$). This is almost true but not quite. For an $F_j$, if $S_j$ gets merged then the distance between $C$ and $F_j$ significantly increases and a later event might merge before it. We will call such an $F_j$ \emph{dormant}. A dormant event will eventually merge with its outer point and the rest of the merges will not be affected by when this merge happens. We call any non-merged, non-dormant event \emph{active}. We will see that every merge is either a dormant event merging out or is the earliest active event choosing between $C$ and its outer merge.

We place points as follows. Start by choosing parameters $W$, $\Delta$, $\tau$, $r_i$, and $l_i$ where the last two are for all $i \in [n]$. We place $C$ at the origin with weight $W$. We will think of $W$ as being heavy. For $S_j$ and $F_j$ we will place them roughly $\Delta$ from the center along the $j$th axis with $S_j$ going on the positive side and $F_j$ going on the negative side. We will offset each of these points by $\tau \cdot e(S_j)$ (or $\tau \cdot e(F_j)$ respectively) so that the points merge in the desired order. Each of these points will have weight $1$. We then place $R_j$ and $L_j$ outside of $S_j$ and $F_j$ respectively and set their weight to $W$. See the table below for the specifics of each point and \Cref{fig:tcp} for an illustration of the reduction.\\

\begin{table}[h]
    \centering
\begin{tabularx}{0.9\textwidth} { 
     >{\centering\arraybackslash}X 
     >{\centering\arraybackslash}X 
     >{\centering\arraybackslash}X  }
    \toprule
    \textbf{Point} & \textbf{Position} & \textbf{Weight} \\
    \toprule
    $C$  & $\bf{\hat{0}}$  & W  \\
    \midrule
    $S_j$  & $\left(\Delta + \tau \cdot e(S_j)\right)\bf{\hat{i}_j}$  & 1  \\
    \midrule
    $R_j$  & $\left(\Delta + \tau \cdot e(S_j) + r_j\right)\bf{\hat{i}_j}$  & 1  \\
    \midrule
    $F_j$  & $-\left(\Delta + \tau \cdot e(F_j)\right)\bf{\hat{i}_j}$  & W  \\
    \midrule
    $L_j$  & $-\left(\Delta + \tau \cdot e(F_j) + l_j\right)\bf{\hat{i}_j}$  & W  \\
    \bottomrule
\end{tabularx}
\end{table}

\begin{figure}[ht]
    \centering
    \begin{subfigure}[b]{0.70\textwidth}
        \centering
        \includegraphics[width=\textwidth, trim=0mm 0mm 0mm 120mm, clip]{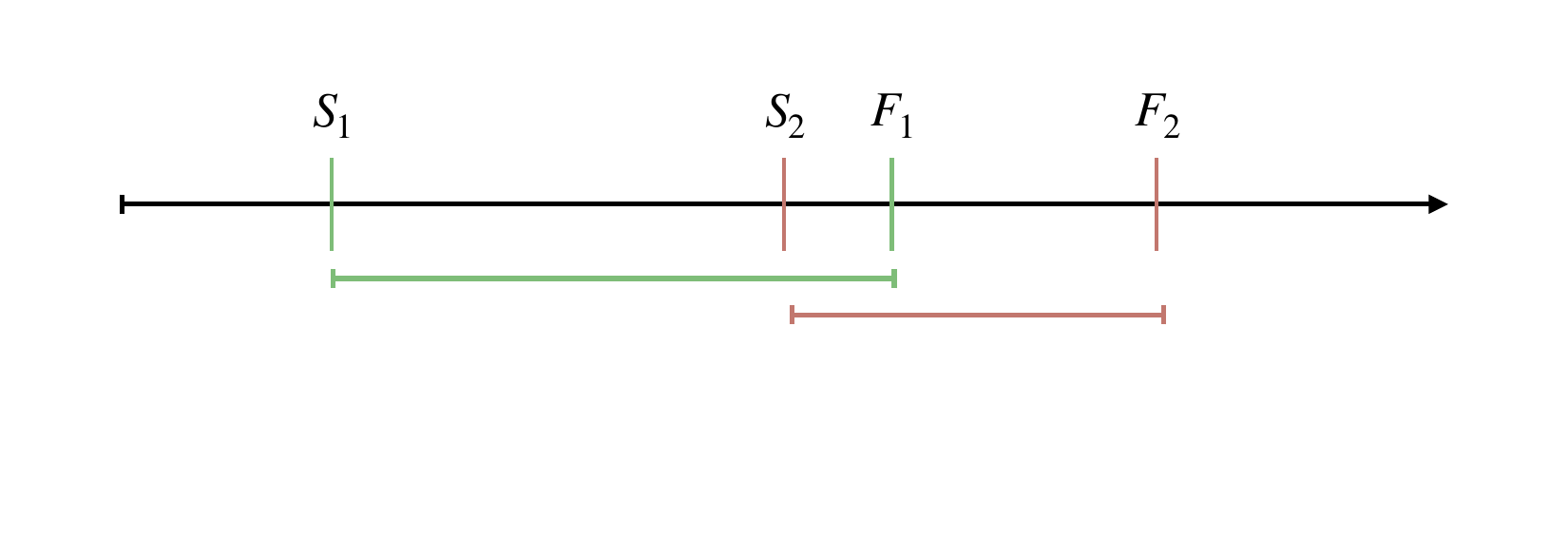}
        \caption{An instance of TCP.}\label{sfig:tcp_setup}
    \end{subfigure} \hspace{1.5em} 
    \begin{subfigure}[b]{0.65\textwidth}
        \centering
        \includegraphics[width=\textwidth, trim=55mm 35mm 60mm 32mm, clip]{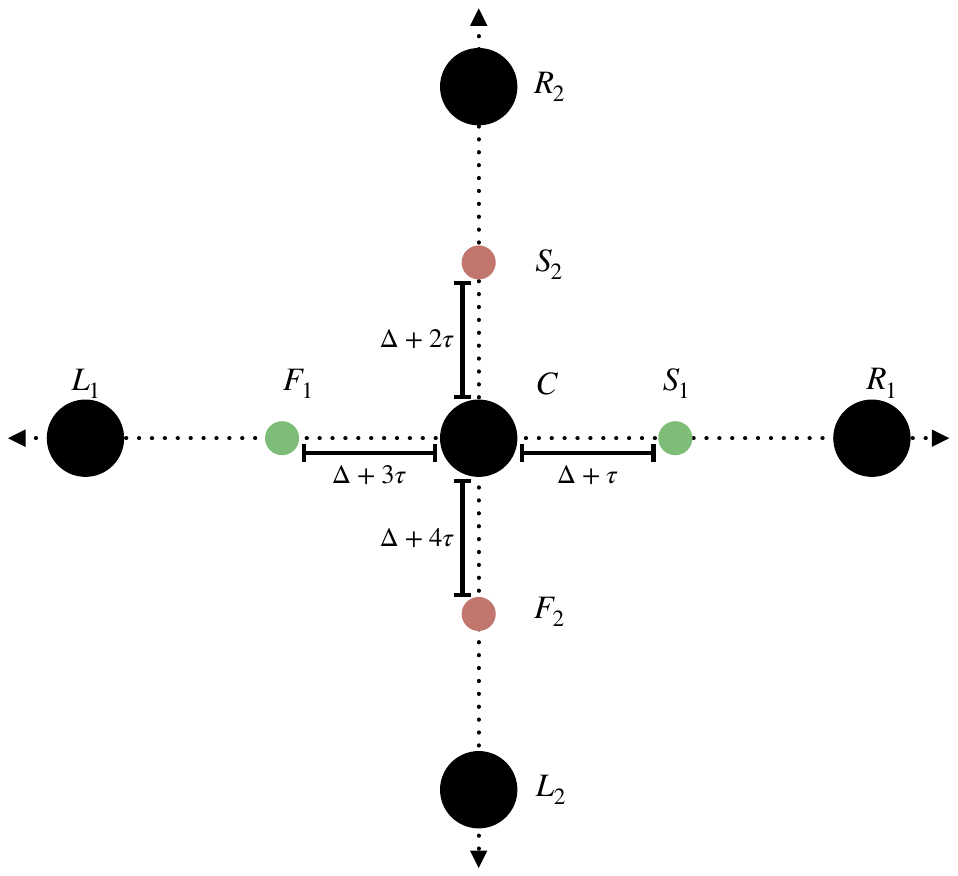}
        \caption{Set of points in $\mathbb{R}^2$ that we will run HAC on.}\label{sfig:tcp_reduction}
    \end{subfigure}  \hspace{1.5em}
    \vspace{-0.5em}
    \caption{An example of the reduction from TCP to HAC for two phone calls.} \label{fig:tcp}
\end{figure}

Since points will move throughout the runtime of HAC, we will use $C^{(i)}$ (respectively $L^{(i)}$ and $R^{(i)}$) to refer to the point $C$ (respectively $L_i$ and $R_i$) immediately before the $i$th merge takes place. When $S_i$ or $F_i$ merge we will think of them as disappearing instead of moving. Let $\delta_x^i = [C^{(i)}]_x$ be the absolute value of the $x$th coordinate of $C^{(i)}$. We will show that throughout the runtime of HAC, $\delta_l < \delta_x^i \leq \delta_u$ where

\begin{itemize}
    \item $\delta_l = \frac{\Delta}{W+n}$ and
    \item $\delta_u = \frac{\Delta + 2n\tau}{W}$.
\end{itemize}

\begin{lemma}
    Suppose one of $S_x$ and $F_x$ has merged with $C^{(i)}$, at most one of $S_y$ and $F_y$ have merged with $C^{(i)}$ for all calls $y$, and that no other points have merged with $C^{(i)}$. Then $\delta_l < \delta_x^i \leq \delta_u$.
\end{lemma}

\begin{proof}
    The weight of $C^{(i)}$ will always be between $W$ and $W+n$. Also, $C^{(i)}$ has exactly one point merged with it that has a non-zero $x$ coordinate. The absolute value of that coordinate is between $\Delta + \tau$ and $\Delta + 2n \tau$. The lemma follows.
\end{proof}

\subsection{Proof of Correctness of Reduction}

We start with a lemma showing what properties we require for our reduction to work.

\begin{lemma} \label{lem:hardness_requirements}
    Suppose we can set $\Delta, W,\tau$, $r_i$ and $l_i$ such that for some $\epsilon > 0$ and all $i \in [2n]$ we have:
    \begin{enumerate}
        \item \textbf{S merges with R:} For some unmerged $S_j$, if for at least $f(S_j) + \kappa$ coordinates $c$ of $C^{(i)}$, $|[C^{(i)}]_c| \in (\delta_l,\delta_u)$ and for the rest including $c=j$, $[C^{(i)}]_c = 0$, then $d(C^{(i)}, S_j) > (1+\epsilon) \cdot  r_j$. \label{item:s_merge_out}
        \item \textbf{S merges with C:} For some unmerged $S_j$, if for less than $f(S_j) + \kappa$ coordinates $c$ of $C^{(i)}$, $|[C^{(i)}]_c| \in (\delta_l,\delta_u)$ and for the rest including $c=j$, $[C^{(i)}]_c = 0$, then $(1+\epsilon) \cdot d(C^{(i)}, S_j) < r_j$. \label{item:s_merge_in}
        \item \textbf{F merges with L:} For some unmerged $F_j$, if $[C^{(i)}]_j \in [\delta_l,\delta_u]$ and for the rest of the coordinates $c \neq j$, $[C]_c \in [-\delta_u, \delta_u]$, then $d(C^{(i)}, F_j) > (1+\epsilon) \cdot l_j$.\label{item:f_merge_out}
        \item \textbf{F merges with C:} For some unmerged $F_j$, if for at most $f(F_j) + \kappa$ coordinates $c$ of $C^{(i)}$, $|[C^{(i)}]_c| \in (\delta_l,\delta_u)$ and for the rest including $c=j$, $[C]_c = 0$, then $(1+\epsilon) \cdot d(C^{(i)}, F_j) < l_j$. \label{item:f_merge_in}
        \item \textbf{Outer distances in order:} Let $E_x, E_y$ be two events with $x < y$. Let $o_x$ and $o_y$ be the $l$ or $r$ value corresponding to $E_x$ and $E_y$ depending on whether they are $S$ or $F$ events. Then $(1 + \epsilon) \cdot o_x < o_y$. \label{item:outer_in_order}
        \item \textbf{Inner distances in order:} Let $E_x, E_y$ be two events with $x < y$. If for all coordinates $c$ of $C^{(i)}$, $[C^{(i)}]_c \in (-\delta_u,\delta_u)$ and $[C^{(i)}]_c = 0$ for the coordinates of $E_x$ and $E_y$, then $(1+\epsilon) \cdot d(C^{(i)}, E_x) < d(C^{(i)}, E_y)$. \label{item:inner_in_order}
        \item \textbf{No bad merges:} If there is some active event, for all coordinates $c$ of $C^{(i)}$, $[C^{(i)}]_c \in [-\delta_u,\delta_u]$, and each outer point has at most merged with its event, then $(1+\epsilon)$-approximate HAC will merge an unmerged event with $C$ or its outer point.  \label{item:no_bad_merges}
    \end{enumerate}

    Then, if $\Delta, W$, $\tau$, $r_i$ and $l_i$ are logspace computable, $(1+\epsilon)$-approximate promise decision centroid HAC is CC-hard in $\mathbb{R}^n$.
\end{lemma}

\begin{proof}
    Suppose we are given an instance $\mathcal{I} = ((S_1,F_1), \ldots, (S_n,F_n), \kappa, t)$ of TCP. Let $P$ be the set of points for our reduction. If $\Delta, W, \tau$, $r_i$ and $l_i$ are logspace computable, then every point in $P$ is logspace computable so we are only concerned that HAC returns true if and only if TCP does. We will show that $S_t$ merges with $C$ if and only if call $t$ is accepted in $\mathcal{I}$, proving the lemma.

    We will prove by induction that
    
    \begin{itemize}
        \item in every round either the earliest active event or a dormant event gets merged,
        \item for each call $j \in [n]$, $S_j$ merges with $C$ if and only if call $j$ is accepted in $\mathcal{I}$ and merges with $R_j$ otherwise, and
        \item $F_j$ merges with $C$ if and only if $S_j$ does not and merges with $L_i$ otherwise.
    \end{itemize}
    
     For the base case, note that before any merges are made that none of these three criteria have been broken. Now consider some round $i$ and assume our criteria are satisfied so far. By our inductive assumptions we know that every coordinate of $C^{(i)}$ is in the range $(-\delta_u,\delta_u)$ and each outer point has at most merged with its event so by \cref{item:no_bad_merges}, we know that some unmerged event will be merged with $C^{(i)}$ or its outer merge. Let $E$ be the earliest active event. By \cref{item:outer_in_order} and \cref{item:inner_in_order} we know that in round $i$ either $E$ or some dormant event, $D$, will be involved in the merge. 

    Assume $E$ is involved in the merge. Either $E = S_j$ or $E = F_j$ for some $j \in [n]$. We first consider the case when $E = S_j$. By our inductive assumptions, every event before $S_j$ has either merged or is dormant. Therefore, for $f(S_j) + a(S_j)$ coordinates $c$ of $C^{(i)}$, $|[C^{(i)}]_c| \in (\delta_l,\delta_u)$ and for the rest including $c=j$, $[C^{(i)}]_c = 0$. Thus, if $a(S_j) < \kappa$, call $j$ is accepted and by \cref{item:s_merge_in}, $S_j$ merges with $C^{(i)}$. On the other hand, if $a(S_j) = \kappa$, call $j$ is dropped and by \cref{item:s_merge_out}, $S_j$ merges with $R_j$. Therefore, $S_j$ merges with $C$ if and only if call $j$ is accepted in $\mathcal{I}$ as desired.
    
    Next, consider the case when $E = F_j$. Since $E$ is active and not dormant, $S_j$ must have merged with $R_j$. Thus, for at most $f(F_j) + a(F_j)$ coordinates $c$ of $C^{(i)}$, $|[C^{(i)}]_c| \in (\delta_l,\delta_u)$ and for the rest including $c=j$, $[C]_c = 0$. Therefore, by \cref{item:f_merge_in}, $F_j$ merges with $C^{(i)}$ as desired.

    Lastly, we consider the case that $D$ is in the merge in round $i$. Since $D$ is dormant, $D = F_j$ for some $j \in [n]$ and $S_j$ merged with $C$. Thus, $[C^{(i)}]_j \in (\delta_l,\delta_u)$ and for the rest of the coordinates $c \neq j$, $[C]_c \in (-\delta_u, \delta_u)$ so by \cref{item:f_merge_out}, $F_j$ merges with $L_j$ as desired. 
    
    Therefore all three of our inductive assumption hold in round $i$. By induction, $S_t$ merges with $C$ before either merge with $R_t$ if and only if call $t$ is accepted in $\mathcal{I}$. Thus, HAC returns true if and only if TCP does, proving the lemma.
\end{proof}

Assign the parameters as follows:

\begin{itemize}
    \item $\epsilon = \frac{1}{n^7}$
    \item $\tau = 1/n$
    \item $W = n^3$
    \item $\Delta = n^5$
    \item $r_i = \left(1+\frac{2}{n^7}\right)\sqrt{(k-1+f(S_i))\delta_u^2 + \left(\Delta + \tau \cdot e(S_i)\right)^2}$
    \item $l_i = \left(1+\frac{2}{n^7}\right)\sqrt{(k+f(F_i))\delta_u^2 + \left(\Delta + \tau \cdot e(F_i)\right)^2}$
\end{itemize}

We now prove that the items from \cref{lem:hardness_requirements} hold with the given parameters in a series of lemmas. We start by showing that $S_j$ merges with $R_j$ if there are $\kappa$ active calls at time $S_j$ in the TCP instance.

\begin{lemma}[\cref{item:s_merge_out}]\label{lem:s_merge_out}
    For some unmerged $S_j$, if for at least $f(S_j) + \kappa$ coordinates $c$ of $C^{(i)}$, $|[C^{(i)}]_c| \in [\delta_l,\delta_u]$ and for the rest including $c=j$, $[C]_c = 0$, then $d(C^{(i)}, S_j) > (1+\epsilon) \cdot  r_j$.
\end{lemma}

\begin{proof}

    By the assumptions of the lemma, we have that $$d(C^{(i)}, S_j) \geq \sqrt{(k+f(S_i))\delta_l^2 + \left(\Delta + \tau \cdot e(S_i)\right)^2}.$$ It follows that

    \begin{align*}
        &d(C^{(i)}, S_j) > (1+\epsilon) \cdot  r_j \\
        &\Leftarrow \sqrt{(k+f(S_i))\delta_l^2 + \left(\Delta + \tau \cdot e(S_i)\right)^2} > (1+\epsilon) \cdot \left(1+\frac{2}{n^7}\right)\sqrt{(k-1+f(S_i))\delta_u^2 + \left(\Delta + \tau \cdot e(S_i)\right)^2} \\
        &\Leftarrow \sqrt{(2n)\delta_l^2 + \left(\Delta + 2n\tau \right)^2} > (1+\epsilon) \cdot \left(1+\frac{2}{n^7}\right)\sqrt{(2n-1)\delta_u^2 + \left(\Delta + 2n \tau \right)^2} \\
        &\Leftarrow (2n)\delta_l^2 + \left(\Delta + 2n \tau \right)^2 > (1+\epsilon)^2 \cdot \left(1+\frac{2}{n^7}\right)^2\left((2n-1)\delta_u^2 + \left(\Delta + 2n \tau \right)^2\right) \\
        &\Leftarrow \frac{2n \cdot \Delta^2}{(W+n)^2} + \left(\Delta + 2n \tau \right)^2 > (1+\epsilon)^2 \cdot \left(1+\frac{2}{n^7}\right)^2\left(\frac{(2n-1)(\Delta + 2n \tau)^2}{W^2} + \left(\Delta + 2n \tau \right)^2\right) \\
        &\Leftarrow 2n \cdot \Delta^2 W^2 + \left(\Delta + 2n \tau \right)^2 W^2 (W+n)^2 \\
        &\hspace{1cm} > (1+\epsilon)^2 \cdot \left(1+\frac{2}{n^7}\right)^2\left((2n-1)(\Delta + 2n \tau)^2 (W+n)^2 + \left(\Delta + 2n \tau \right)^2 W^2 (W+n)^2\right) \\
        &\Leftarrow 2n \cdot n^{10} n^6 + \left(n^5 + 2 \right)^2 n^6 (n^3+n)^2 \\
        &\hspace{1cm} > \left(1+\frac{1}{n^7}\right)^2 \cdot \left(1+\frac{2}{n^7}\right)^2 \left((2n-1)(n^5 + 2)^2 (n^3+n)^2 + \left(n^5 + 2\right)^2 n^6 (n^3+n)^2\right) \\
        &\Leftarrow n^{22} + 2 n^{20} + n^{18} + 6 n^{17} + 8 n^{15} + 4 n^{13} + 4 n^{12} + 8 n^{10} + 4 n^8 \\
        &\hspace{1cm} > \left(1+O\left(\frac{1}{n^7}\right)\right) \cdot \left(n^{22} + 2 n^{20} + n^{18} + 6 n^{17} - n^{16} + 12 n^{15} - 2 n^{14} + O(n^{13}) \right) \\
    \end{align*}

    where the final equation holds for all $n$ greater than some constant. Thus, the initial equation holds as desired.
\end{proof}

Next we show that if there are less than $\kappa$ active calls, $S_j$ merges with $C$.

\begin{lemma} [\cref{item:s_merge_in}] \label{lem:s_merge_in}
    For some unmerged $S_j$, if for less than $f(S_j) + \kappa$ coordinates $c$ of $C^{(i)}$, $|[C^{(i)}]_c| \in [\delta_l,\delta_u]$ and for the rest including $c=j$, $[C]_c = 0$, then $(1+\epsilon) \cdot d(C^{(i)}, S_j) < r_j$.
\end{lemma}

\begin{proof}
    By the assumptions of the lemma, we have that $$d(C^{(i)}, S_j) \leq \sqrt{(k - 1 + f(S_i))\delta_u^2 + \left(\Delta + \tau \cdot e(S_i)\right)^2}.$$ It follows that

    \begin{align*}
        &(1+\epsilon) \cdot d(C^{(i)}, S_j) < r_j \\
        &\Leftarrow \left(1+\frac{1}{n^7}\right) \cdot \sqrt{(k - 1 + f(S_i))\delta_u^2 + \left(\Delta + \tau \cdot e(S_i)\right)^2}\\
        &\hspace{1cm} < \left(1+\frac{2}{n^7}\right)\sqrt{(k-1+f(S_i))\delta_u^2 + \left(\Delta + \tau \cdot e(S_i)\right)^2}
    \end{align*}

    so the initial equation holds as desired.
\end{proof}

We now move on to $F_j$. We start by showing that $F_j$ merges with $L_j$ if $S_j$ merged with $C$.

\begin{lemma}[\cref{item:f_merge_out}] \label{lem:f_merge_out}
    For some unmerged $F_j$, if $[C^{(i)}]_j \in [\delta_l,\delta_u]$ and for the rest of the coordinates $c \neq j$, $[C]_c \in [-\delta_u, \delta_u]$, then $d(C^{(i)}, F_j) > (1+\epsilon) \cdot l_j$.
\end{lemma}

\begin{proof}
    By the assumptions of the lemma, we have that $$d(C^{(i)}, F_j) \geq \Delta + \tau \cdot e(F_i) + \delta_l.$$ It follows that

    \begin{align*}
        &d(C^{(i)}, F_j) > (1+\epsilon) \cdot l_j \\
        &\Leftarrow \Delta + \tau \cdot e(F_i) + \delta_l > (1+\epsilon) \cdot \left(1+\frac{2}{n^7}\right)\sqrt{(k+f(F_i))\delta_u^2 + \left(\Delta + \tau \cdot e(F_i)\right)^2} \\
        &\Leftarrow \Delta + \delta_l > (1+\epsilon) \cdot \left(1+\frac{2}{n^7}\right)\sqrt{(k+f(F_i))\delta_u^2 + (\Delta + 2n \tau)^2} \\
        &\Leftarrow \Delta + \delta_l > (1+\epsilon) \cdot \left(1+\frac{2}{n^7}\right)\sqrt{(2n-1)\delta_u^2 + (\Delta + 2n \tau)^2} \\
        &\Leftarrow \Delta^2 + 2 \Delta\delta_l + \delta_l^2 > (1+\epsilon)^2 \cdot \left(1+\frac{2}{n^7}\right)^2\left((2n-1)\delta_u^2 + (\Delta + 2n \tau)^2 \right) \\
        &\Leftarrow \Delta^2 + \frac{2\Delta^2}{W+n} + \frac{\Delta^2}{(W+n)^2} > (1+\epsilon)^2 \cdot \left(1+\frac{2}{n^7}\right)^2 \left(\frac{(2n-1)(\Delta+2 n \tau)^2}{W^2} +(\Delta + 2n \tau)^2\right) \\
        &\Leftarrow \Delta^2(W+n)^2 W^2 + 2\Delta^2 (W+n) W^2 + \Delta^2 W^2 \\
        &\hspace{1cm} > (1+\epsilon)^2 \cdot \left(1+\frac{2}{n^7}\right)^2 \left((2n-1)(\Delta+2 n \tau)^2 (W+n)^2 + (\Delta + 2n \tau)^2 (W+n)^2 W^2\right) \\
        &\Leftarrow n^{10}(n^3 + n)^2 n^6 + 2n^{10} (n^3 + n) n^6 + n^{10} n^6 \\
        &\hspace{1cm} > \left(1+\frac{1}{n^7}\right)^2 \cdot \left(1+\frac{2}{n^7}\right)^2 \left((2n-1)(n^5 + 2)^2 (n^3 + n)^2 + (n^5+2)^2 (n^3 + n)^2 n^6\right) \\
        &\Leftarrow n^{22} + 2 n^{20} + 2 n^{19} + n^{18} + 2 n^{17} + n^{16} \\
        &\hspace{1cm} > \left(1+O\left(\frac{1}{n^7}\right)\right) \cdot \left(n^{22} + 2 n^{20} + n^{18} + 6 n^{17} - n^{16} + O(n^{15})\right)\\
    \end{align*}

    where the final equation holds for all $n$ greater than some constant. Thus, the initial equation holds as desired.
\end{proof}

We now show that $F_j$ merges with $C$ if $S_j$ did not merge with $C$.

\begin{lemma}[\cref{item:f_merge_in}] \label{lem:f_merge_in}
    For some unmerged $F_j$, if for at most $f(F_j) + \kappa$ coordinates $c$ of $C^{(i)}$, $|[C^{(i)}]_c| \in (\delta_l,\delta_u)$ and for the rest including $c=j$, $[C]_c = 0$, then $(1+\epsilon) \cdot d(C^{(i)}, F_j^{(i)}) < l_j$.
\end{lemma}

\begin{proof}
    By the assumptions of the lemma, we have that $$d(C^{(i)}, F_j) < \sqrt{(k+f(F_i))\delta_u^2 + \left(\Delta + \tau \cdot e(F_i)\right)^2}.$$ It follows that

    \begin{align*}
        &(1+\epsilon) \cdot d(C^{(i)}, F_j) < l_j \\
        &\Leftarrow \left(1+\frac{1}{n^7}\right) \cdot \sqrt{(k+f(F_i))\delta_u^2 + \left(\Delta + \tau \cdot e(F_i)\right)^2} < \left(1+\frac{2}{n^7}\right)\sqrt{(k+f(F_i))\delta_u^2 + \left(\Delta + \tau \cdot e(F_i)\right)^2} \\
    \end{align*}

    where the final equation holds for all $n$ so the initial equation holds as desired.
\end{proof}

Next, we want to show that the events merge in order. We start by showing that the distance the distance between an event and its outer point is at least $(1+\epsilon)$ times that of the same distance for any earlier event.

\begin{lemma}[\cref{item:outer_in_order}] \label{lem:outer_in_order}
    Let $E_x, E_y$ be two events with $x < y$. Let $o_x$ and $o_y$ be the $l$ or $r$ value corresponding to $E_x$ and $E_y$ depending on whether they are $S$ or $F$ events. Then $(1 + \epsilon) \cdot o_x < o_y$.
\end{lemma}

\begin{proof}

    First, note that $f(E_x) \leq f(E_y)$ and if $E_x$ is an $F$ event the $f(E_x) \leq f(E_y) + 1$. Also, $e(E_X) \leq e(E_y) + 1$. It follows that
    
    \begin{align*}
        &(1 + \epsilon) \cdot o_x < o_y \\
        &\Leftarrow (1+\epsilon) \cdot \left(1+\frac{2}{n^7}\right)\sqrt{2n\delta_u^2 + \left(\Delta + \tau \cdot e(E_x)\right)^2} < \left(1+\frac{2}{n^7}\right)\sqrt{2n\delta_u^2 + \left(\Delta + \tau \cdot e(E_y)\right)^2} \\
        &\Leftarrow (1+\epsilon) \cdot \sqrt{2n\delta_u^2 + \left(\Delta + \tau \cdot e(E_x)\right)^2} < \sqrt{2n\delta_u^2 + \left(\Delta + \tau \cdot e(E_y)\right)^2} \\
        &\Leftarrow (1+\epsilon) \cdot \sqrt{2n\delta_u^2 + \left(\Delta + (2n-1)\tau\right)^2} < \sqrt{2n\delta_u^2 + \left(\Delta + 2n\tau \right)^2} \\
        &\Leftarrow (1+\epsilon)^2 \cdot \left(2n\delta_u^2 + \left(\Delta + (2n-1)\tau\right)^2\right) < 2n\delta_u^2 + \left(\Delta + 2n\tau \right)^2 \\
        &\Leftarrow (1+\epsilon)^2 \cdot \left(\frac{2n(\Delta + 2 n \tau)^2}{W^2} + \left(\Delta + (2n-1)\tau\right)^2\right) < \frac{2n(\Delta + 2 n \tau)^2}{W^2} + \left(\Delta + 2n\tau \right)^2 \\
        &\Leftarrow (1+\epsilon)^2 \cdot \left(2n(\Delta + 2 n \tau)^2 + \left(\Delta + (2n-1)\tau\right)^2 W^2\right)\\
        &\hspace{1cm} < \left(2n(\Delta + 2 n \tau)^2 + \left(\Delta + 2n\tau\right)^2 W^2\right) \\
        &\Leftarrow \left(1+\frac{1}{n^7}\right)^2 \cdot \left(2n(n^5 + 2)^2 + \left(n^{5} + 2 - \frac{1}{n}\right)^2 n^{6}\right)\\
        &\hspace{1cm} < \left(2n(n^{5} + 2)^2 + \left(n^{5} + 2 \right)^2 n^{6}\right) \\
        &\Leftarrow \left(1+O\left(\frac{1}{n^7}\right)\right) \cdot \left(n^{16} + 6 n^{11} - 2 n^{10} + 12 n^6 - 4 n^5 + n^4 + 8 n\right)\\
        &\hspace{1cm} < n^{16} + 6 n^{11} + 12 n^6 + 8 n \\
    \end{align*}

    where the final equation holds for all $n$ greater than some constant. Thus, the initial equation holds as desired.
\end{proof}

We now prove the same thing except for the distance between an event and $C$.

\begin{lemma}[\cref{item:inner_in_order}] \label{lem:inner_in_order}
    Let $E_x, E_y$ be two events with $x < y$. If for all coordinates $c$ of $C^{(i)}$, $[C^{(i)}]_c \in [-\delta_u,\delta_u]$ and $[C^{(i)}]_c = 0$ for the coordinates of $E_x$ and $E_y$, then $(1+\epsilon) \cdot d(C^{(i)}, E_x) < d(C^{(i)}, E_y)$.
\end{lemma}

\begin{proof}
    
    \begin{align*}
        &(1+\epsilon) \cdot d(C^{(i)}, E_x) < d(C^{(i)}, E_y) \\
        &\Leftarrow (1+\epsilon) \cdot \sqrt{n\delta_u^2 + \left(\Delta + \tau \cdot e(E_x)\right)^2} < \sqrt{n\delta_u^2 + \left(\Delta + \tau \cdot e(E_y)\right)^2} \\
        &\Leftarrow (1+\epsilon) \cdot \sqrt{n\delta_u^2 + \left(\Delta + \tau \cdot e(E_x)\right)^2} < \sqrt{n\delta_u^2 + \left(\Delta + \tau \cdot (e(E_x) + 1)\right)^2} \\
        &\Leftarrow (1+\epsilon)^2 \cdot \left(n\delta_u^2 + \left(\Delta + \tau \cdot e(E_x)\right)^2\right) < n\delta_u^2 + \left(\Delta + \tau \cdot (e(E_x) +1)\right)^2 \\
        &\Leftarrow (1+\epsilon)^2 \cdot \left(\frac{n(\Delta + 2 n \tau)^2}{W^2} + \left(\Delta + \tau \cdot e(E_x)\right)^2\right) < \frac{n(\Delta + 2 n \tau)^2}{W^2} + \left(\Delta + \tau \cdot (e(E_x) + 1)\right)^2 \\
        &\Leftarrow (1+\epsilon)^2 \cdot \left(n(\Delta + 2 n \tau)^2 + W^2 \left(\Delta + \tau \cdot e(E_x)\right)^2\right)\\
        &\hspace{1cm} < n(\Delta + 2 n \tau)^2 + W^2 \left(\Delta + \tau \cdot (e(E_x) + 1)\right)^2 \\
        &\Leftarrow \left(1+\frac{1}{n^7}\right)^2 \cdot \left(n(n^5 + 2)^2 + n^6 \left(n^5 + e(E_x)/n\right)^2\right)\\
        &\hspace{1cm} < n(n^5 + 2)^2 + n^6 \left(n^5 + (e(E_x)+ 1)/n\right)^2 \\
        &\Leftarrow \left(1+O\left(\frac{1}{n^7}\right)\right) \cdot \left(n^{16} + n^{11} + 2 n^{10} e(E_x) + 4 n^6 + n^4 e(E_x)^2 + 4 n\right)\\
        &\hspace{1cm} < n^{16} + n^{11} + 2 n^{10} e(E_x) + 2n^{10} + 4 n^6 + n^4 e(E_x)^2 + 2 n^4 e(E_x) + n^4 + 4 n \\
    \end{align*}

    where the final equation holds for all $n$ greater than some constant since $1 \leq e(E_x) \leq 2n$. Thus, the initial equation holds
as desired.
\end{proof}

Lastly, we show that all merges will be between an event and either its outer point or the center.

\begin{lemma}[\cref{item:no_bad_merges}] \label{lem:no_bad_merges}
    If there is some active event, for all coordinates $c$ of $C^{(i)}$, $[C^{(i)}]_c \in [-\delta_u,\delta_u]$, and each outer point has at most merged with its event, then $(1+\epsilon)$-approximate HAC will merge an unmerged event with $C$ or its outer point.
\end{lemma}

\begin{proof}
    Consider some event $E$ and its outer point $O_E$. Then $$d(E,O_E) \leq \left(1+\frac{2}{n^7}\right)\sqrt{2n \delta_u^2 + (\Delta + 2n \tau)^2}.$$

    Let $F$ and $G$ be two events. The type of merges we want to rule out are those between

    \begin{enumerate}
        \item $C^{(i)}$ and $O_F$
        \item $X_F$ and $X_G$ where $X_F$ is either $F$ or $O_F$ and $X_G$ is either $G$ or $O_G$
    \end{enumerate}

    First we will rule out the first type of merge. By assumption, we have that $$d\left(C^{(i)}, O_F\right) \geq \left( \frac{W}{W+1} \right)(2 \Delta - \delta_u).$$

    We then want to show that

    \begin{align*}
        &(1 + \epsilon) \cdot d(E,O_E) < d\left(C^{(i)}, O_F\right) \\
        &\Leftarrow (1+\epsilon) \cdot \left(1+\frac{2}{n^7}\right)\sqrt{2n \delta_u^2 + (\Delta + 2n \tau)^2} < \left( \frac{W}{W+1} \right) (2 \Delta - \delta_u) \\
        &\Leftarrow (1+\epsilon)^2 \cdot \left(1+\frac{2}{n^7}\right)^2\left(2n \delta_u^2 + (\Delta + 2n \tau)^2\right) < \left( \frac{W}{W+1} \right)^2(2 \Delta - \delta_u)^2 \\
        &\Leftarrow (1+\epsilon)^2 \cdot \left(1+\frac{2}{n^7}\right)^2\left(2n \delta_u^2 + \Delta^2 + 4n \tau \Delta + 4n^2 \tau^2 \right) < \left( \frac{W}{W+1} \right)^2 (4 \Delta^2 - 4 \Delta \delta_u + \delta_u^2) \\
        &\Leftarrow (1+\epsilon)^2 \cdot \left(1+\frac{2}{n^7}\right)^2\left(\frac{2n(\Delta + 2n\tau)^2}{W^2} + \Delta^2 + 4n \tau \Delta + 4n^2 \tau^2 \right)\\
        &\hspace{1cm} < \left( \frac{W}{W+1} \right)^2 \left(4 \Delta^2 - \frac{4\Delta(\Delta + 2n\tau)}{W} + \frac{(\Delta + 2n\tau)^2}{W^2} \right) \\
        &\Leftarrow (1+\epsilon)^2 \cdot \left(1+\frac{2}{n^7}\right)^2\left(2n(\Delta + 2n\tau)^2 + \Delta^2 W^2 + 4n \tau \Delta W^2 + 4n^2 \tau^2 W^2 \right)\\
        &\hspace{1cm} < \left( \frac{W}{W+1} \right)^2 (4 \Delta^2 W^2 - 4\Delta(\Delta + 2n\tau)W + (\Delta + 2n\tau)^2) \\
        &\Leftarrow \left(1+\frac{1}{n^7}\right)^2 \cdot \left(1+\frac{2}{n^7}\right)^2 \left(2n(n^5 + 2)^2 + n^{10} n^6 + 4 n^5 n^6 + 4 n^6 \right)\\
        &\hspace{1cm} < \left( \frac{n^3}{n^3+1} \right)^2(4 n^{10} n^6 - 4n^5(n^5 + 2) n^3 + (n^5 + 2)^2) \\
        &\Leftarrow \left(1+O\left(\frac{1}{n^7}\right)\right) \cdot \left(n^{16} + 6 n^{11} + 12 n^6 + 8 n \right) < \left( \frac{n^3}{n^3+1} \right)^2(4 n^{16} - 4 n^{13} + n^{10} - 8 n^8 + 4 n^5 + 4) \\
    \end{align*}

    where the final equation holds for all n greater than some constant. Thus, the initial equation holds as desired.

    Next we will rule out the second type of merge. By assumption, we have that $$d\left(X_F, X_G\right) \geq \sqrt{2 \Delta^2}.$$

    It is enough to show that

    \begin{align*}
        &(1 + \epsilon) \cdot d(E,O_E) < d\left(X_F, X_G\right)\\
        &\Leftarrow (1+\epsilon) \cdot \left(1+\frac{2}{n^7}\right)\sqrt{2n \delta_u^2 + (\Delta + 2n \tau)^2} < \sqrt{2 \Delta^2} \\
        &\Leftarrow (1+\epsilon)^2 \cdot \left(1+\frac{2}{n^7}\right)^2 \left(2n \delta_u^2 + (\Delta + 2n \tau)^2\right) < 2 \Delta^2 \\
        &\Leftarrow (1+\epsilon)^2 \cdot \left(1+\frac{2}{n^7}\right)^2\left(2n \delta_u^2 + \Delta^2 + 4n \tau \Delta + 4n^2 \tau^2 \right) < 2 \Delta^2 \\
        &\Leftarrow (1+\epsilon)^2 \cdot \left(1+\frac{2}{n^7}\right)^2 \left(\frac{2n(\Delta + 2n\tau)^2}{W^2} + \Delta^2 + 4n \tau \Delta + 4n^2 \tau^2 \right) < 2 \Delta^2 \\
        &\Leftarrow (1+\epsilon)^2 \cdot \left(1+\frac{2}{n^7}\right)^2 \left(2n(\Delta + 2n\tau)^2 + \Delta^2 W^2 + 4n \tau \Delta W^2 + 4n^2 \tau^2 W^2 \right) < 2 \Delta^2 W^2 \\
        &\Leftarrow \left(1+\frac{1}{n^7}\right)^2 \cdot \left(1+\frac{2}{n^7}\right)^2 \left(2n(n^5 + 2)^2 + n^{10} n^6 + 4 n^5 n^6 + 4 n^6 \right) < 2 n^{10} n^6 \\
        &\Leftarrow \left(1+O\left(\frac{1}{n^7}\right)\right) \cdot \left(n^{16} + 6 n^{11} + 12 n^6 + 8 n \right) < 2 n^{16} \\
    \end{align*}

    where again the final equation holds for all n greater than some constant. Thus, the initial equation holds as desired.

\end{proof}

We now put everything together to prove our main hardness result.

\ccHardness*

\begin{proof}
    We have set $\Delta, W$, $\tau$, $r_i$ and $l_i$ so that they are all logspace computable. Thus, \Cref{lem:hardness_requirements} along with \Cref{lem:s_merge_out}, \Cref{lem:s_merge_in}, \Cref{lem:f_merge_out}, \Cref{lem:f_merge_out}, \Cref{lem:outer_in_order}, \Cref{lem:inner_in_order}, and \Cref{lem:no_bad_merges} prove the theorem.
\end{proof}

\clearpage

\bibliography{main}
\clearpage

\appendix

\section{Proofs from \texorpdfstring{\Cref{sec:packablelinkage}}{Section 3}}\label{sec:DefProofProps}
In this section we give deferred proofs from \Cref{sec:packablelinkage}.
\subsection{Packing Points Proof}
We let $V_k(r)$ give the volume of a radius $r$ ball in $k$ dimensions and use the following well-known closed-form of $V_k(r)$. 
\begin{lemma}[\cite{smith1989small}]\label{lem:ballVol}
    The volume of a radius $r$ ball in $k$ dimensions is
    \begin{align*}
        V_k(r) = \frac{r^k\pi^{k/2}}{\Gamma(\frac{k}{2} + 1)}
    \end{align*}
    where $\Gamma$ is Euler's gamma function.
\end{lemma}

We then have our packing theorem, as follows.
\packPoints*
\begin{proof}
    Let $V_k = V_k(1)$ be the volume of a radius $1$ ball in $\mathbb{R}^k$. By \Cref{lem:ballVol} we have that for any $x \in \mathbb{R}^k$ and $R \geq 0$ that the volume of the radius $R$ ball centered at $x$ is
    \begin{align*}
        \text{Vol}(B(x, R)) = R^k \cdot V_k.
    \end{align*}
    Let $\mcB = \{B(p,r/3) : p \in \mcP \}$ be all balls of radius $r/3$ centered at points of $\mcP$. Observe that the intersection of any two balls of $\mcB$ is empty but each ball of $\mcB$ is contained in $B(x, R)$ and so
    \begin{align*}
        \left(\frac{r}{3}\right)^k \cdot V_k \cdot |P| = \sum_{B \in \mcB}\text{Vol}(B) \leq \text{Vol}(B(x, R)) = R^k \cdot V_k.
    \end{align*}
    Solving for $|P|$ we get
    \begin{align*}
        |P| \leq \left(\frac{3R}{r}\right)^k = \left(\frac{R}{r}\right)^{O(k)}
    \end{align*}
    as required.
\end{proof}

\subsection{Alternate Ward's Form Proofs}

Given $C \subseteq \mathbb{R}^k$, we let $\Delta(C,x) := \sum_{c \in C}\|c-x \|^2$ denote the sum of squared distances from each point in cluster $C$ to some arbitrary point $x$. Then, we have the following identity.
\begin{lemma}\label{lem:ident_sos}
    $\Delta(C,x) = \Delta(C)+|C|\|x-\mu(C)\|^2$.
\end{lemma}
\begin{proof}
We have
\begin{align*}
    \Delta(C,x) &= \sum_{y\in C}\|x-y\|^2\\
     &= \sum_{y\in C} \|(x-\mu(C))-(y-\mu(C))\|^2\\
     &= |C|\|x-\mu(C)\|^2 + \Delta(C)-\sum_{y\in C}\langle x-\mu(C),y-\mu(C) \rangle\\
     &= |C|\|x-\mu(C)\|^2 + \Delta(C)-\langle x-\mu(C),\sum_{y\in C}y-|C|\mu(C) \rangle\\
     &= |C|\|x-\mu(C)\|^2 + \Delta(C)
\end{align*}
as required.
\end{proof}

\noindent Using the above lemma we can get the the following alternate form for Ward's.

\begin{restatable}[Alternate Ward's]{lemma}{altWardForm}\label{lem:wardAlt}
$\dWard(A,B) = \frac{|A||B|}{|A|+|B|}\|\mu(A)-\mu(B)\|^2$.
\end{restatable}
\begin{proof}
By \Cref{lem:ident_sos}, we have
\begin{align*}
    \dWard(A,B) &= \Delta(A\cup B)-\Delta(A)-\Delta(B)\\
     &= \Delta(A)+|A|\|\mu(A)-\mu(A\cup B)\|^2+\Delta(B)+|B|\|\mu(B)-\mu(A\cup B)\|^2-\Delta(A)-\Delta(B)\\
     &= |A|\|\mu(A)-\mu(A\cup B)\|^2+|B|\|\mu(B)-\mu(A\cup B)\|^2\\
     &= \frac{|A||B|^2}{(|A|+|B|)^2}\|\mu(A)-\mu(B)\|^2 + \frac{|B||A|^2}{(|A|+|B|)^2}\|\mu(A)-\mu(B)\|^2\\
     &=\frac{|A||B|}{|A|+|B|}\|\mu(A)-\mu(B)\|^2,
\end{align*}
where the second equality follows by \Cref{lem:ident_sos}, and the fourth equality follows by the fact that when clusters $A$ and $B$ are merged, the centroid $\mu(A\cup B)$ lies on the line joining $\mu(A)$ and $\mu(B)$, and $\|\mu(A)-\mu(A\cup B)\|=\frac{|B|}{|A|+|B|}\|\mu(A)-\mu(B)\|$. 
\end{proof}

We now prove the $2$-approximation for Ward's.
\wardAPX*
\begin{proof}
WLOG, assume $|B|\le|A|$ and apply \Cref{lem:wardAlt}. Then, the left-hand-side follows since $\frac{|A||B|}{|A|+|B|} \ge \frac{|A||B|}{2|A|}$. For the right-hand-side, divide the numerator and denominator by $|A|$. Then, $\frac{|B|}{1+|B|/|A|} \le |B|$. 
\end{proof}

Next, we prove the Lance-Williams form for updated Ward's distances.
\LanceWilliamsWards*
\begin{proof}
    Consider,
    \begin{align*}
        |A|\|\mu(A)-\mu(C)\|^2+|B|\|\mu(B)-\mu(C)\|^2 =& |A|\|\mu(A)-\mu(A\cup B)\|^2+|B|\|\mu(B)-\mu(A\cup B)\|^2\\
        &+ (|A|+|B|)\|\mu(A\cup B)-\mu(C)\|^2\\
        =&\left(\frac{|A||B|^2}{(|A|+|B|)^2}+\frac{|B||A|^2}{(|A|+|B|)^2}\right)\|\mu(A)-\mu(B)\|^2\\
        &+ (|A|+|B|)\|\mu(A\cup B)-\mu(C)\|^2\\
        =& \dWard(A,B)+(|A|+|B|)\|\mu(A\cup B)-\mu(C)\|^2,
    \end{align*}
    where the first equality follows by \Cref{lem:ident_sos}, and the second equality follows by the fact that when clusters $A$ and $B$ are merged, the centroid $\mu(A\cup B)$ lies on the line joining $\mu(A)$ and $\mu(B)$, and
    $\|\mu(A)-\mu(A\cup B)\|=\frac{|B|}{|A|+|B|}\|\mu(A)-\mu(B)\|$. Rearranging and multiplying by $|C|/(|A|+|B|+|C|)$, we get that $\dWard(A\cup B,C)$ is 
    \begin{align*}
        \frac{|A||C|}{|A|+|B|+|C|}\|\mu(A)-\mu(C)\|^2+\frac{|B||C|}{|A|+|B|+|C|}\|\mu(B)-\mu(C)\|^2-\frac{|C|}{|A|+|B|+|C|}\dWard(A,B)
    \end{align*}
    which is
    \begin{align*}
        \frac{|A|+|C|}{|A|+|B|+|C|}\dWard(A,C)+\frac{|B|+|C|}{|A|+|B|+|C|}\dWard(B,C)-\frac{|C|}{|A|+|B|+|C|}\dWard(A,B)
    \end{align*}
    as required.
\end{proof}

\subsection{Approximate Triangle Inequality for Squared Euclidean Distances Proof}

\squaredEucTri*
\begin{proof}
    By the triangle inequality for (non-squared) Euclidean distances we have
    \begin{align}\label{eq:asgas}
        \|a-c\|^2 & \leq (\|a-b\| + \|b- c\|)^2\nonumber\\
        &= \|a-b\|^2 + \|b- c\|^2 + 2 \|a-b\|\|b-c\|.
    \end{align}
    By the AM-GM inequality we have 
    \begin{align*}
        \|a-b\|\|b-c\| \leq \frac{\|a-b\|^2+\|b-c\|^2}{2}
    \end{align*}
    and so
    \begin{align}
        2\|a-b\|\|b-c\| \leq \|a-b\|^2+\|b-c\|^2.\label{eq:asgaass}
    \end{align}
    Plugging \Cref{eq:asgaass} into \Cref{eq:asgas} gives
    \begin{align*}
        \|a-c\|^2 \leq 2 \cdot \left(\|a-b\|^2 + \|b- c\|^2 \right)
    \end{align*}
    as desired.
\end{proof}

\end{document}